\documentclass[reqno,11pt]{amsart}
\usepackage{graphicx}
\usepackage{slashed}
\usepackage{amscd}
\usepackage{tensor}
\usepackage{amssymb}
\usepackage{esint}
\usepackage{enumitem}
\usepackage{accents}
\usepackage{mathtools}
\usepackage{xcolor}
\usepackage{pst-grad} 
\usepackage{pst-plot} 
\usepackage[mathscr]{eucal}
\usepackage{stackengine,graphicx,amssymb}
\stackMath

\numberwithin{equation}{section}
\allowdisplaybreaks[1]

\title[Local Algebras for Causal Fermion Systems]{Local Algebras for Causal Fermion Systems \\ in Minkowski Space}

\author[F.\ Finster]{Felix Finster}
\author[M.\ Oppio]{Marco Oppio \\ \\ April 2020}
\address{Fakult\"at f\"ur Mathematik \\ Universit\"at Regensburg \\ D-93040 Regensburg \\ Germany}
\email{finster@ur.de, marco.oppio@uni-r.de}

\newtheorem{Def}{Definition}[section]
\newtheorem{Thm}[Def]{Theorem}

\newtheorem{Prp}[Def]{Proposition}
\newtheorem{Lemma}[Def]{Lemma}
\newtheorem{Remark}[Def]{Remark}
\newtheorem{Corollary}[Def]{Corollary}

\newtheorem{Part}[Def]{Step}
\newcommand{\Thanks}{\vspace*{.5em} \noindent \thanks}
\newcommand{\beq}{\begin{equation}}
\newcommand{\eeq}{\end{equation}}
\newcommand{\Proof}{\begin{proof}}
\newcommand{\QED}{\end{proof} \noindent}

\newcommand{\la}{\langle}
\newcommand{\ra}{\rangle}

\newcommand{\Sl}{\mbox{$\prec \!\!$ \nolinebreak}}
\newcommand{\Sr}{\mbox{\nolinebreak $\succ$}}

\newcommand{\C}{\mathbb{C}}
\newcommand{\R}{\mathbb{R}}

\newcommand{\N}{\mathbb{N}}

\renewcommand{\O}{{\mathscr{O}}}

\DeclareMathOperator{\supp}{supp}
\renewcommand{\H}{\mathscr{H}}

\newcommand{\Lin}{\text{\rm{L}}}
\newcommand{\F}{{\mathscr{F}}}

\DeclareMathOperator{\im}{Im}

\newcommand{\scrM}{\mycal M}

\newcommand{\A}{\mycal A}
\newcommand{\x}{{\textit{\myfont x}}}
\newcommand{\y}{{\textit{\myfont y}}}
\newcommand{\bitem}{\begin{itemize}[leftmargin=2.5em]}
\newcommand{\eitem}{\end{itemize}}

\newcommand*{\myfont}{\fontfamily{ppl}\selectfont}		
\DeclareFontFamily{OT1}{rsfso}{}
\DeclareFontShape{OT1}{rsfso}{m}{n}{ <-7> rsfso5 <7-10> rsfso7 <10-> rsfso10}{}
\DeclareMathAlphabet{\mycal}{OT1}{rsfso}{m}{n}

\def\scL{L}

\def\scH{\mathscr{H}}

\def\cF{\mathcal{F}}

\def\gR{\mathfrak{R}}
\def\bI{\mathbb{I}}

\def\gR{\mathfrak{R}}

\def\cD{\mathcal{D}}
\def\gG{\mathfrak{G}}
\DeclareMathOperator{\ran}{ran}
\newcommand{\V}[1]{{\bf{#1}}}

\begin{document}

\maketitle

\begin{abstract}
A notion of local algebras is introduced in the theory of causal fermion systems. Their properties
are studied in the example of the regularized Dirac sea vacuum in Minkowski space.
The commutation relations are worked out, and the differences to the
canonical commutation relations are discussed.
It is shown that the spacetime point operators associated to a Cauchy surface satisfy a time slice axiom.
It is proven that the algebra generated by operators in an open set is irreducible as a consequence
of Hegerfeldt's theorem.
The light cone structure is recovered by analyzing expectation values
of the operators in the algebra in the limit when the regularization is removed.
It is shown that every spacetime point operator commutes with
the algebras localized away from its null cone, up to small corrections involving the regularization length.
\end{abstract}

\tableofcontents

\section{Introduction}
{ In relativity, spacetime is described by Minkowski space or, more generally,
by a Lorentzian manifold. Although this description has been highly successful, it is generally believed that in order to reconcile general relativity with quantum theory, the mathematical structure of spacetime should be modified on a microscopic scale (typically thought of as the Planck scale). 
The theory of causal fermion systems is a recent approach to fundamental physics where spacetime is no longer modelled by a Lorentzian manifold but may instead have a nontrivial, possibly discrete structure on a microscopic length scale. In this approach,}
spacetime and all the structures therein { (including causal relations, wave functions,
fields, etc.)} are described by a measure~$\rho$
on a set of linear operators~$\F$ on a Hilbert space~$\H$
(for an introduction and the physical context see the textbook~\cite{cfs} the survey articles~\cite{dice2014, review}
or the website~\cite{cfsweblink}; for the basic definitions see Section~\ref{seccfs} below).
{\em{Spacetime}}~$M$ is defined as the support of this measure,
\[ M := \supp \rho \subset \Lin(\H) \:. \]
Thus the spacetime points are linear operators on~$\H$.
{The physical equations are then formulated via a variational principle, the \textit{causal action principle}, which singles out the measures of physical interest.
In \cite{cfs} it is shown that in a specific limiting case, the so-called \textit{continuum limit}, spacetime goes over to
Minkowski space, and the Euler-Lagrange equations arising from the causal action principle give rise to the classical field equations of the standard model and linearized gravity. 
For these reasons, causal fermion systems are a promising candidate for a unified physical theory.
}

{ Although connections to quantum field theory have already been explored
(see \cite{qft, fockbosonic} or the survey article~\cite{srev}),
the relation between causal fermion systems and the algebraic formulation of quantum theory has not yet been studied.
The objective of this paper is to 
get a first contact between these two approaches. 
More concretely, the general goal is to form operator algebras generated by
sets of spacetime operators and to analyze their
properties, with a special focus on the causal relations among them.
As the first step in this program, we here restrict attention to systems in Minkowski space.

 In a nutshell, the procedure adopted in this work consists in collecting} 
all the operators in an open subset~$\Omega \subset M$ of spacetime,
\[  {\mycal X}_\Omega := \{\x \:|\: \x \in \Omega \} \:, \]
and smear these operators with
continuous test functions of compact support inside~$\Omega$,
\beq \label{smearing}
A_f := \int_\Omega \x\: f(\x)\: d\rho(\x) \:,\qquad \text{$f \in C^0_0(\Omega, \C)$}
\eeq
(the smearing is preferable for our applications in mind).
We then consider the $^*$-algebra~$\A_\Omega$ generated by all these operators.
The questions of interest are: What are the commutation relations between algebras localized in different spacetime regions? Is the causal structure
encoded in the algebras? How can it be retrieved?

We work out the commutation relations and show that, in contrast to the
situation for the algebras of observables in axiomatic quantum field theory, the commutators
are in general non-zero even for space-like separation (see Propositions~\ref{prpcomm}
and~\ref{prpcommmink}).
This result is not surprising because the causal structure of a causal fermion system
is not defined in terms of commutators of spacetime points,
but instead by spectral properties of the operator products (see Definition~\ref{def2}).

In order to study the connection between the local algebras and the causal structure
in more detail, we consider the example of causal fermion systems 
in Minkowski space referred to as {\em{regularized Dirac sea vacua}} (see 
Definition~\ref{defregDir} in Section~\ref{subsectionCFSM}). 
In these examples, the Hilbert space~$(\H, \la \cdot | \cdot \ra)$ is formed of all
the negative-energy solutions of the Dirac equation in Minkowski space.
The causal fermion systems depend on a parameter~$\varepsilon>0$
which describes the ultraviolet regularization. More precisely, to every point~$x$
in Minkowski space we associate the corresponding 
{\em{local correlation operator}}~$F^\varepsilon(x) \in \F$ (for details see~\eqref{defF}). 
Then the measure~$\rho$ on $\F$ is introduced as the push-forward through~$F^\varepsilon$ of the Lebesgue measure~$d\mu=d^4x$ of Minkowski space,
\[ \rho := (F^\varepsilon)_* \mu \:. \]
In these examples, we prove that the operators~$F^\varepsilon(x)$ with~$x$ on a constant-time Cauchy surface
fulfill the time slice axiom (see the end of Section~\ref{sec41}).

Using that the resulting mapping~$F^\varepsilon$ defines a closed homeomorphism to its image (see Theorem \ref{listpropertiesLCF}),
the open subsets of $\supp F^\varepsilon_*(\mu)$ can be identified naturally with the open subsets of Minkowski space. 
In order to clarify the dependence
on the regularization, we denote the corresponding local algebras by~${\mycal A}^\varepsilon_\Omega$, where $\Omega$ now is an open subset of Minkowski space.
The smearing in~\eqref{smearing}, now performed with smooth test functions~$f \in C^\infty_0(\Omega, \C)$,
has the effect of improving the properties of the
algebras in the limit~$\varepsilon \searrow 0$. Indeed, we can even leave out the regularization
and introduce the {\em{unregularized algebra}}~${\mycal A}^\circ_\Omega$,
although the individual spacetime operators $F^\varepsilon(x)$ with $x\in\Omega$
diverge as~$\varepsilon \searrow 0$.

The algebras~$\A^\varepsilon_\Omega$ and~$\A^\circ_\Omega$ turn
out to be irreducible as a consequence of Hegerfeldt's theorem
(see Proposition~\ref{prphegerfeldt}, Proposition~\ref{identification}
and Theorem~\ref{irreducibilityunreg}).
In other words, the von Neumann algebra generated by the smeared operators 
coincides with~$\Lin(\H)$. This means that, taking the weak closure of~$\A^\varepsilon_\Omega$,
all the information on the set~$\Omega$ gets lost. 
Taking the uniform closure of the algebras (i.e.\ the closure in the $\sup$-norm) gives
sensible $C^*$-algebras. However, as it is not obvious how to extend all our results to the uniform closure,
we shall always restrict attention to the $^*$-algebras~$\A^\varepsilon_\Omega$ and~$\A^\circ_\Omega$.

In order to recover the light cone structure, we study the $\varepsilon$-dependence of
expectation values of elements of the algebras~${\mycal A}^\varepsilon_\Omega$ and~${\mycal A}^\circ_\Omega$
with vectors~$u^\varepsilon_{x,\chi}$ which
are ``concentrated'' near a spacetime point~$x$ of Minkowski space and
develop singularities on the null cone centered at~$x$ as~$\varepsilon \searrow 0$
(for details see~\eqref{uchidef} and the explanation thereafter).
Taking the matrix elements with these vectors,
\[ \la u^\varepsilon_{x,\chi} | A u^\varepsilon_{x,\zeta} \ra \:, \]
for the unregularized algebras~$\A^\circ_\Omega$ we derive the following bounds
(see Theorems~\ref{teoremalmite2} and~\ref{theoremblowup}):
\bitem
\item[(i)] If~$\Omega$ intersects the null cone centered at~$x$,
for some~$A \in \A^\circ_\Omega$ and some~$\chi, \zeta \in \C^4$ there is~$c>0$ such that, for sufficiently
small~$\varepsilon>0$,
\[ \big|\la u^\varepsilon_{x,\chi} | A u^\varepsilon_{x,\zeta} \ra \big| \geq
\frac{1}{c\: \varepsilon^2} \:. \]
\item[(ii)] Conversely, if~$\Omega$ does not intersect the null cone centered at~$x$,
then for any~$A \in \A^\circ_\Omega$ there is a constant~$c(A)$ such that for all $\varepsilon>0$,
\[ \big|\la u^\varepsilon_{x,\chi} | A u^\varepsilon_{x,\zeta} \ra\big| \leq
c(A)\:|\chi|\: |\zeta| \qquad \text{for all~$\chi, \zeta \in \C^4$} \:. \]
\eitem
We also derive similar estimates for the regularized algebras
(see Section \ref{seclc2}).

We also prove that every spacetime operator~$F^\varepsilon(x)$ commutes with the algebra~$\A_\Omega^\circ$
if~$\Omega$ does not intersect the null cone centered at~$x$,
up to small corrections involving the regularization length.
More precisely, for every~$A \in \A^\circ_\Omega$ there is a constant~$c(A)>0$ such that
(see Proposition~\ref{Prpcommunreg})
$$
\big\| [A, F^\varepsilon(x) ] \big\| \le c(A)\, \varepsilon^\frac{3}{2}\, \|F^\varepsilon(x)\| \:.
$$
We also derive corresponding estimates for the regularized algebras
(see Theorem~\ref{prpiotaes}). This estimate is motivated by the notion of ``events''
in the ETH formulation of quantum theory~\cite{froehlich2019review, froehlich2019relativistic}.
This connection is explored further in~\cite{eth-cfs}.

The paper is organized as follows.
Section~\ref{secprelim} provides the necessary preliminaries on causal fermion systems in Minkowski space.
In Section~\ref{secalgebras} local algebras for causal fermion systems are introduced abstractly,
and the commutation relations are derived.
In Section~\ref{secregalg} these algebras are worked out more concretely in the example
of regularized Dirac sea vacua. 
In Section~\ref{secunregularized} unregularized local algebras are defined and analyzed.
Finally, in order not to distract from the main ideas, the more technical proofs are given in
the appendices. Finally, in Section~\ref{outlook} we draw the conclusions of the paper and address potential future perspectives which may be undertaken as natural continuations of this work.

\section{Preliminaries} \label{secprelim}

\subsection{Basics on Causal Fermion Systems} \label{seccfs}
We now give a brief summary of the basic mathematical objects in the theory of causal fermion systems.
Since we are here interested in Dirac systems in Minkowski space, we restrict attention to the case of
spin dimension two.

\begin{Def}\label{defCFS} Given a separable complex Hilbert space~$(\mathscr{H}, \la .|. \ra)$, we let~$\mathscr{F} \subset \Lin(\H)$ be the set of all selfadjoint operators of finite rank which - counting multiplicities - have at most two positive and two negative eigenvalues.
\end{Def} \noindent
Note that the set~$\F$ is not a linear space, because the sum of two operators in~$\F$ in general
will have rank larger than four. But~$\F$ has the structure of a {\em{closed double cone}}, meaning
that~$\F$ is closed in the $\sup$-norm topology and that for every~$A \in \F$, the ray~$\R A$ is also
contained in~$\F$. 

Next, we let~$\rho$ be a Borel set on~$\F$, where by a Borel measure we always mean a measure on the Borel
algebra on~$\F$ (with respect to the $\sup$-norm topology) which is finite on every compact set.

\begin{Def}\label{definitioncfs}
The triple~$(\mathscr{H}, \F, \varrho)$ is referred to as a \bf{causal fermion system}.
\end{Def}

A causal fermion system describes a spacetime together with all structures and objects therein.
We only recall those structures which will be needed in this paper
(for a more complete account see~\cite[Section~1.1]{cfs}). 
{\em{Spacetime}}, denoted by~$M$, is defined as the support of~$\rho$,
\[ M := \text{supp}\, \rho \subset \F \:. \]
Endowed with the $\sup$-norm topology, $M$ is a topological space.
The fact that the spacetime points are operators gives rise to additional structures:
For every~$\x \in \F$ we define the {\em{spin space}}~$S_\x$ by~$S_\x= \x(\H)$; it 
is a subspace of~$\H$ of dimension at most four. It is endowed with the {\em{spin scalar product}} defined by
\beq \label{ssp}
\Sl \cdot | \cdot \Sr_\x := -\la \,\cdot \,|\, \x \,\cdot\, \ra \::\: S_\x \times S_\x \rightarrow \C \:,
\eeq
which is an indefinite inner product of signature~$(p,q)$ with~$p,q \leq 2$. Moreover, 
for every~$u \in \H$, we introduce the {\em{physical wave function}}~$\psi^u$ by
projecting the vector~$u$ to the corresponding spin spaces,
\[ 
\psi^u \::\: M \rightarrow \H\:,\qquad \psi^u(\x) := \pi_\x u \in S_\x \:. \]
(where~$\pi_\x$ is the orthogonal projection on the subspace~$S_x \subset \H$).
In this way, similar to a section of a vector bundle, to any spacetime point we associate
a vector in the corresponding spin space. Next, for any~$\x,\y \in M$ we define the
{\em{kernel of the fermionic projector}}~$\mathrm{P}(\x,\y)$ by
\beq \label{Pxydef}
\mathrm{P}(\x,\y) = \pi_\x \,\y|_{S_\y} \::\: S_\y \rightarrow S_\x \:.
\eeq
The kernel of the fermionic projector is a mapping from one spin space to another, thereby
inducing relations between different spacetime points.
Finally, the {\em{closed chain}}~$\mathrm{A}_{\x\y}$ is defined as the product
\beq \label{Axydef}
\mathrm{A}_{\x\y}= \mathrm{P}(\x,\y)\, \mathrm{P}(\y,\x) \::\: S_\x\rightarrow S_\x\:.
\eeq
The spectrum of the closed chain gives rise to the following notion of causality:
\begin{Def} \label{def2} 
For any~$\x, \y \in M$, we denote the eigenvalues
of the closed chain~$A_{\x\y}$ (counting algebraic multiplicities) by~$\lambda^{\x\y}_1, \ldots, \lambda^{\x\y}_{2n}$.
The points~$\x$ and~$\y$ are
called {\bf{spacelike}} separated if all the~$\lambda^{\x\y}_j$ have the same absolute value.
They are said to be {\bf{timelike}} separated if the~$\lambda^{\x\y}_j$ are all real and do not all 
have the same absolute value.
In all other cases (i.e.\ if the~$\lambda^{\x\y}_j$ are not all real and do not all 
have the same absolute value),
the points~$\x$ and~$\y$ are said to be {\bf{lightlike}} separated.
\end{Def} \noindent

All the above objects and structures are {\em{inherent}} in the sense that we only use information already encoded
in the causal fermion system.
The correspondence to the usual notions in Minkowski space has been
worked out in~\cite[Section~1.2]{cfs}.
In what follows, we mainly restrict our attention to causal fermion systems in Minkowski space.
In order to make the paper self-contained, we now give the necessary background.

\subsection{The Dirac Equation in Minkowski Space}\label{sectiondiracequation}
In the present work we mainly focus on causal fermion systems obtained by regularizing
the vacuum Dirac sea in Minkowski space~$\scrM$ as analyzed in detail in~\cite{oppio}.
For notational simplicity, we work in a fixed reference frame and identify Minkowski
space with~$\R^{1,3}$, endowed with the standard Minkowski inner product with signature
convention $(+,-,-,-)$, which we denote by~$u \!\cdot\! v$. 
We denote spacetime indices by~$i,j \in \{0, \ldots, 3\}$ and spatial indices by~$\alpha, \beta \in \{1,2,3\}$.
We use natural units $\hbar = c = 1$.
The Minkowski metric gives rise to a {\em{light cone structure}}:
The null vectors form the double cone $L_0 =
\{ \xi \in \scrM \,|\, \xi \cdot \xi =0\}$, referred to as the
{\em{null cone}}. Physically, the null cone is formed of all light
rays through the origin of~$\scrM$. Similarly, the timelike vectors correspond to velocities
slower than the speed of light; they form the {\em{interior light cone}}
$I_0 = \{ \xi \in \scrM \,|\, \xi \cdot \xi > 0 \}$.
Finally, we introduce the {\em{closed light cone}}
$J_0 = \{ \xi \in \scrM \,|\, \xi \cdot \xi \geq 0 \}$.
With our conventions, the null cone is the boundary of the closed or interior light cones.
By translation, we obtain corresponding cones centered at
any spacetime point $x$. They will be denoted by $L_x,I_x$ and $J_x$, respectively.

Our starting point is the vector space of all smooth solutions of the Dirac equation $i \gamma^j \partial_j \psi=m\psi$ with spatially compact support $\scH_m^{\rm{sc}} \subset C^\infty(\R^{1,3}, \C^4)$.
These Dirac solutions can be described by their initial data at time $t=0$.
More precisely, there is a linear isomorphism
$$
\mathrm{E}:C_0^\infty(\R^3,\C^4)\rightarrow\scH_m^{\rm{sc}}
$$
which propagates the compactly supported initial data, given for example at time $t=0$,
to the whole spacetime.
The linear space $\scH_m^{\rm{sc}}$ can be given a pre-Hilbert space structure by equipping it with the $L^2$ scalar product of the initial data (due to current conservation, integrating over any
surface~$\{t=\text{const}\}$ would give the same result),
$$
(  f \,|\,g ) := \int_{\R^3} f(0, \V{x})^\dagger g(0, \V{x})\: d^3\V{x}\qquad \mbox{for all }f,g\in\scH_m^{\rm{sc}} \:,
$$ 
where the dagger means complex conjugation and transposition.
This makes the mapping~$\mathrm{E}$ to a linear isometry.
The one-particle Hilbert space~$(\scH_m,(\cdot|\cdot))$ is introduced
as the Hilbert space completion of $\scH_m^{\rm{sc}}$. It coincides with the topological completion of $\scH_m^{\rm{sc}}$ within $\scL^2_{\text{loc}}(\R^{1,3},\C^4)$. As a consequence, the isomorphism  $\mathrm{E}$ extends continuously to a unitary operator on $\scL^2(\R^3,\C^4)$, again denoted by $\mathrm{E}$. Finally, we remark that the elements of $\scH_m^-$ admit weak derivatives on $\R^{1,3}$ and satisfy the Dirac equation in the weak sense.

By means of the Fourier transform, to every function  $\psi\in\scL^2(\R^3,\C^4)$ we associate a three-momentum distribution $\cF(\psi)\in\scL^2(\R^3,\C^4)$ and vice versa. Typically, it is more convenient to work in
momentum space. In particular, we use the operator
\begin{equation}\label{ehat}
\hat{\mathrm{E}}:=\mathrm{E}\circ\cF^{-1} \:. 
\end{equation}
Finally, the Hilbert space $\scH_m$ can be decomposed into an orthogonal direct sum of $\scH_m^+$ and $\scH_m^-$, the positive and negative energy subspaces. This is easily formulated in momentum space by means of the following projection operators,
\beq \label{Pmpdef}
\hat{P}_\pm\psi:=p_{\pm}\cdot \psi,\quad \psi\in\scL^2(\R^3,\C^4),\quad p_\pm(\V{k}):=\frac{\slashed{k}+m}{2k^0}\:
\gamma^0 \Big|_{k^0=\pm \omega(\V{k})} \:,
\eeq
where~$\omega(\V{k}):= \sqrt{\V{k}^2+m^2}$ and~$\slashed{k} := k_j \gamma^j$.
A convenient space to work with are the Schwartz functions~${\mathcal{S}}(\R^3, \C^4)$.
It follows from~\eqref{Pmpdef} that the operators~$\hat{P}_\pm$ map Schwartz functions to themselves
and that their image is dense in $\scH_m^\pm$ (for details see~\cite[Lemma 2.17]{oppio}).
On these functions, the action of the operator \eqref{ehat} reads
\begin{equation}\label{expressionEonS}
\hat{\mathrm{E}}(\psi)(t,\V{x})=\int_{\R^3}\frac{d^3\V{k}}{(2\pi)^{3/2}}\left(\psi_+(\V{k})\,e^{-i(\omega(\V{k})t-\V{k}\cdot\V{x})}+\psi_-(\V{k})\,e^{-i(-\omega(\V{k})t-\V{k}\cdot\V{x})}\right)
\end{equation} 
where $\psi_\pm:=\hat{P}_\pm(\psi)$ (see \cite[Proposition 2.19]{oppio}). Note that $\hat{\mathrm{E}}(\psi)\in\H_m\cap C^\infty(\R^{1,3},\C^4)$.

We also point out that the space of solutions $\scH_m$ is equipped with a strongly-continuous unitary representation of the group of translations in spacetime:
\begin{equation}
\mathrm{U}_a:\scH_m\ni u\mapsto u(\,\cdot\,+a)\in\scH_m,\quad a\in\R^{1,3} \:. \label{Udef}
\end{equation}
These operators also preserve the sign of the energy, i.e.\ $\mathrm{U}_a(\scH_m^\pm)\subset\scH_m^\pm$ for all $a\in\R^{1,3}$.

\subsection{Regularization by a Smooth Cutoff in Momentum Space}
In the context of causal fermion systems, in order to take into account the presence of a minimal length scale, a regularization is introduced. This length parameter~$\varepsilon$ can vary in an interval $(0,\varepsilon_{\max})$. 
For technical simplicity, we here regularize by multiplying in momentum space 
by a convergence-generating factor~$e^{-\varepsilon \omega(\V{k})}$ (in~\cite[\S2.4.1]{cfs} this regularization
method is referred to as the $i \varepsilon$-regularization).
This slightly differs from the regularization scheme used in~\cite{oppio}, where a mollification in
spacetime was used. Nevertheless, most of the results of this paper could be extended in a straightforward way
to the regularizations by mollification.
\begin{Def}\label{defreg}
	The {\bf{$i\varepsilon$-regularization operator}} is defined for every $\varepsilon\in (0,\varepsilon_{max})$ by
	$$
	\gR_{\varepsilon} : \scH_m\ni \hat{\mathrm{E}}(\psi)\mapsto \hat{\mathrm{E}}(\mathfrak{g}_\varepsilon \,\psi)\in\scH_m \:,
	$$
	where the regularization function is  $\mathfrak{g}_\varepsilon(\V{k}):=e^{-\varepsilon \omega(\V{k})}$.
\end{Def}
As a consequence of the  boundedness of the cutoff function,
the regularization operators are well-defined continuous linear mappings. Moreover, they map solutions to 
smooth solutions. Finally they preserve the sign of the energy and in the limit $\varepsilon \searrow 0$
converge strongly to the identity. In more detail, we have the following result:

\begin{Prp}\label{propositionregularization}
The $i\varepsilon$-regularization operators are bounded linear functions and have the following properties:
\begin{itemize}[leftmargin=2.5em]
	\vspace{0.4em}
		\item[{\rm{(i)}}] $\gR_\varepsilon(\scH_m^\pm)\subset\scH_m^\pm\cap C^\infty(\R^{1,3},\C^4).$\\[-0.5em]
		\item[{\rm{(ii)}}] $\gR_{\varepsilon} $ is selfadjoint. \\[-0.5em]
		\item[{\rm{(iii)}}] $\ker\gR_{\varepsilon}=\{0\}$ and $\|\gR_{\varepsilon}\|\le 1$\\[-0.5em]
		\item[{\rm{(iv)}}] $\gR_{\varepsilon}u\to u$  as $\varepsilon\searrow 0$ for every $u\in\scH_m$.\\[-0.5em]
		\item[{\rm{(v)}}] If $u\in\hat{\mathrm{E}}(\mathcal{S}(\R^3,\C^4))$, then $\gR_{\varepsilon}u\to u$ uniformly on compact sets  as $\varepsilon\searrow 0$.
\end{itemize}
\end{Prp} \noindent
The proof is given in~Appendix~\ref{secappendix}.

As a final comment, we note that the regularization can be seen as the restriction to the mass shell of the function 
$
e^{-\varepsilon|k^0|}
$
in the four-momentum space.
Since this function is always evaluated on the mass shell, it is possible to replace it by a suitable Schwartz
function~$\mathfrak{G}_\varepsilon$, which coincide with the exponential factor on the hyperboloid (for details see~\cite[Proposition 3.6]{oppio}). With this in mind, in what follows we always assume that the
cutoff function~$\mathfrak{G}_\varepsilon$ is a Schwartz function.

\subsection{The Kernel of the Fermionic Projector}
In what follows, we are mainly interested in the negative-energy spectrum. Therefore we focus our attention on $\scH_m^-$. The restriction of the inner product of $\scH_m$ to $\scH_m^-$ is denoted by $\langle\cdot|\cdot\rangle:=(\cdot|\cdot)\!\restriction_{\scH_m^-\times \scH_m^-}$.
 We introduce the following regularized distribution,
\begin{equation}\label{bidistributionP}
\begin{split}
P^{n\varepsilon}(x,y)&:=\int_{\R^4}\frac{d^4 k}{(2\pi)^4} \:(\slashed{k}+m) \:\delta(k^2-m^2)\:\Theta(-k_0)\:\mathfrak{G}_\varepsilon(k)^n\: e^{-i k\cdot (x-y)}=\\
&\:= -\int_{\R^3}\frac{d^3\V{k}}{(2\pi)^4}\, p_-(\V{k})\:\gamma^0\,\mathfrak{g}_\varepsilon(\V{k})^n \:e^{-i(-\omega(\V{k})(t_x-t_y)-\V{k}\cdot(\V{x-y}))} \:.
\end{split}
\end{equation}
In our case of a $i\varepsilon$-regularization, the $n^\text{th}$ power of the cutoff factor is equivalent to the substitution $\varepsilon\mapsto n\varepsilon$.
We are only interested in $n=0,1,2$. The case $n=0$
is well-defined
in the distributional sense, while the remaining two cases give well-defined smooth functions.
For simplicity, in the case $n=0$ the superscript referring to $\varepsilon$ will be omitted. Also, we introduce the symbol
$$
\hat{P}(k):=\:(\slashed{k}+m) \:\delta(k^2-m^2)\:\Theta(-k_0).
$$

In the next proposition we collect a few properties of these kernels.
In particular, we clarify how the regularized kernels converge to the unregularized
kernel as the cutoff is removed. 

\begin{Prp}\label{propositioncontinuityP1} The following statements hold: \\[-0.2cm]
\begin{itemize}[leftmargin=2.5em]
	\item[\rm{(i)}] For any~$f\in {\mathcal{S}}(\R^{1,3},\C^4)$, $\varepsilon\in (0,\varepsilon_{max})$ and $n=0,1,2$,
	the integral
	\begin{equation}\label{defPf}
	\begin{split}
	\qquad P^{n\varepsilon}(x,f) \,&\!:=\int_{\R^4}d^4x\,P^{n\varepsilon}(x,y)\,f(y) \\
	&=\int_{\R^4}\frac{d^4 k}{(2\pi)^2}\:\hat{P}(k)\:\mathfrak{G}_\varepsilon(k)^n\: \cF(f)(k)\:e^{-ik\cdot x}=\\
	&= -\int_{\R^3}\frac{d^3\V{k}}{(2\pi)^2}\, p_-(\V{k})\:\gamma^0\,\mathfrak{g}_\varepsilon(\V{k})^n\: \cF(f)(-\omega(\V{k}),\V{k}) \:e^{-i(-\omega(\V{k})t-\V{k}\cdot\V{x})} \:.
	\end{split}
	\end{equation}
	is well-defined and gives a smooth solution of the Dirac equation.\\[-0.1cm]
	\item[\rm{(ii)}] Varying~$f$ in~${\mathcal{S}}(\R^{1,3},\C^4)$, the functions $P(\,\cdot\,,f)$
	span a dense subspace of $\scH_m^-$.\\[-0.1cm]
	\item[\rm{(iii)}] For any $\chi\in\C^4$ and $y\in\R^{1,3}$,
\[ 
	P^{\varepsilon}(x,y)\, \chi=P \big( x,T_y(h_\varepsilon \chi) \big) \:, \]
	where $h_\varepsilon:=(2\pi)^{-2}\cF^{-1}(\mathfrak{G}_\varepsilon)$ and $T_y(f)(x):=f(x-y)$ is the translation operator.\\[-0.1cm]
	\item[\rm{(iv)}]  For any $f,g\in {\mathcal{S}}(\R^{1,3},\C^4)$, $\varepsilon\in (0,\varepsilon_{max})$ and $n=0,1,2$, the integral
	\begin{equation*}
	P^{n\varepsilon}(f,g) :=\int_{\R^4}d^4x\int_{\R^4}d^4y\, f(x)^\dagger P^{n\varepsilon}(x,y)\, g(y)=\int_{\R^4}d^4x\, f(x)^\dagger P^{n\varepsilon}(x,g) 
	\end{equation*}
	is well-defined. \\[-0.1cm]
	\item[\rm{(v)}] There are $c,k>0$ such that for all $f,g\in\mathcal{S}(\R^{1,3},\C^4)$, $x\in\R^{1,3}$, $\varepsilon\in (0,\varepsilon_{\max})$ and $n=0,1,2$,
	$$
	|P^{n\varepsilon}(x,g)|\le k\,\|g\|_{6,4}\qquad\mbox{and}\qquad |P^{n\varepsilon}(f,g)|\le c\,\|f\|_{6,0} \,\|g\|_{6,4}\:,
	$$
	(where~$\| . \|_{p,q}$ are the usual Schwartz norms\footnote{We adopt the convention~$\|\phi\|_{p,q}:=\sum_{|\alpha|\le p}\sum_{|\beta|\le q}\|x^\alpha\,D^\beta \phi\|_\infty$.}
	).\\[-0.1cm]
	\item[\rm{(vi)}]  For any $f,g\in\mathcal{S}(\R^{1,3},\C^4)$, $x\in\R^{1,3}$ and $n=0,1,2$,
	\[ 
	\lim_{\varepsilon \searrow 0}P^{n\varepsilon}(x,g)=P(x,g)\qquad\mbox{and}\qquad \lim_{\varepsilon \searrow 0}P^{n\varepsilon}(f,g)=P(f,g) \:. 
	\]
\end{itemize}	
\end{Prp} \noindent
The proof is given in~Appendix~\ref{secappendix}.

The kernels with different exponents of the cutoff functions are related to each other by the following result. The proof is straightforward.
\begin{Prp}\label{prpexponentsP}
	For any~$x\in\R^{1,3}$ and~$\chi\in \C^4$, the following statements hold:
	\begin{itemize}[leftmargin=2.5em]
		\vspace{0.05cm}
	\item[{\rm{(i)}}] $\gR_\varepsilon \big(P(\,\cdot\,,f)\,\chi \big)=P^{\varepsilon}(\,\cdot\,,f)\,\chi\ $\\[-0.5em]
	\item[{\rm{(ii)}}] $\gR_\delta \big(P^{\varepsilon}(\,\cdot\,,x)\,\chi \big)=P^{\varepsilon+\delta}(\,\cdot\,,x)\,\chi$
	\end{itemize}
\end{Prp}

It is worth mentioning that the function $P^\varepsilon(\,\cdot\,,y)$ can actually be calculated explicitly for sufficiently small (but finite) $\varepsilon$
as follows (the details can be found in \cite[Section~1.2.5]{cfs}). We first pull out the Dirac matrices,
\begin{equation}\label{expressionP}
P^{\varepsilon}(x,y)=\big(i\slashed{\partial}_x+m\big)T_{m^2}^\varepsilon(x,y) \:,
\end{equation}
where the function $T^\varepsilon_{m^2}$ is smooth and defined by the modified Bessel function as
$$
T_{m^2}^\varepsilon(x,y):=\frac{m}{(2\pi)^3}\,\frac{K_1\left(m\sqrt{-((y-x)-i\varepsilon e_0)^2}\right)}{\sqrt{-((y-x)-i\varepsilon e_0)^2}}\quad\mbox{for all }x,y\in\R^{1,3},
$$
where  $K_1$ and the square root are defined as usual using a branch cut along $(-\infty,0]$.
In what follows, we also make use of the Bessel functions of first and second kind~$J_1$ and~$Y_1$ (see again \cite[Section~1.2.5]{cfs}) which are analytic in a neighborhood of the positive real line $\R^0_++i(-\delta,\delta)$.

Now, focus for simplicity on $y=0$ and consider two bounded open sets 
$$
\Omega_{\rm{t}}, \Omega_{\rm{s}}
\subset\R^{1,3}\qquad \text{(where $\rm{t}$ stands for ``time-like'' and~$\rm{s}$ for ``space-like'')} \:,
$$
whose closures lie inside the interior light cone~$I_0$ and outside the closed light cone~$J_0$, respectively.
For $\varepsilon$ small enough, we have
\[
\begin{rcases}
+(x+i\varepsilon e_0)^2=+x^2-\varepsilon^2+2ix_0\varepsilon& \text{if }x\in\Omega_{\rm{t}}\\[0.3em] 
-(x+i\varepsilon e_0)^2=-x^2+\varepsilon^2-2i x_0\varepsilon & \text{if }x\in\Omega_{\rm{s}}
\end{rcases} \in \R^0_++i(-\delta,\delta) \]
(where $\{e_j\:|\: j\in\{0,1,2,3\}\}$ denotes the canonical basis of $\R^{1,3}$). In these cases it is possible to write (for details see~\cite[Lemma 1.2.9]{cfs})
\vspace{0.2cm}
\beq \label{Tm2explicit}
T_{m^2}^{\varepsilon}(x,0)=
\begin{cases}
\dfrac{m}{16\pi^2}\,\dfrac{Y_1(m\sqrt{(x+i\varepsilon e_0)^2}}{\sqrt{(x+i\varepsilon e_0)^2}}+\dfrac{im}{16\pi^2}\,\dfrac{J_1(m\sqrt{(x+i\varepsilon e_0)^2}}{\sqrt{(x+i\varepsilon e_0)^2}}\,\varepsilon(-x^0)&\!\!\mbox{if $x\in\Omega_{\rm{t}}$} \\[1.2em]
\dfrac{m}{8\pi^3}\,\dfrac{K_1(m\sqrt{-(x+i\varepsilon e_0)^2}}{\sqrt{-(x+i\varepsilon e_0)^2}}&\!\!\mbox{if $x\in\Omega_{\rm{s}}$}.
\end{cases}
\eeq
\vspace{-0.3cm}

\noindent
This function is well-defined and smooth in the corresponding regions,
even in the limiting case $\varepsilon=0$ (in this case, the function is simply denoted by~$T_{m^2}$).

Now, going back to the general case $y\in\R^{1,3}$, using the above formula, Taylor's theorem, the fact that
the involved functions are analytic and the identity
$$
T_{m^2}^\varepsilon(x,y)=T_{m^2}^\varepsilon(x-y,0)\quad\mbox{for all }x,y\in\R^{1,3},
$$ 
one sees that the functions~$T_{m^2}^\varepsilon(\,\cdot\,,y)$ converge uniformly on compact subsets $K\subset \Omega_{\rm{t,s}}+y$ to the smooth function $T_{m^2}(\,\cdot\,,y)$, as~$\varepsilon\searrow 0$,
$$
T_{m^2}^\varepsilon(\,\cdot\,,y)\!\restriction_K\, \stackrel{\mbox{\tiny{\tiny \rm{\!loc}}}}{\rightrightarrows} T_{m^2}(\,\cdot\,,y)\!\restriction_K.
$$ 
Similar arguments can be used for $P^{\varepsilon}(\,\cdot\,,y)$, taking into account formula~\eqref{expressionP}
(in fact, this can be done for derivatives of any order).
In particular, it follows that~$P(\,\cdot\,,y)$ is smooth in~$\R^{1,3}\setminus L_y$ and coincides with
the smooth function~$\big(i\slashed{\partial}+m\big)T_{m^2}(\,\cdot\,,y)$.
We thus obtain the following result.

\begin{Lemma} \label{lemma26}
	Let $y\in\R^{1,3}$. In the limit~$\varepsilon \searrow 0$,
	the functions~$P^{\varepsilon}(\,\cdot\,,y)$ converge uniformly on any compact subset $K$ of $\R^{1,3}\setminus L_y$ to the smooth
	function~$P(\,\cdot\,,y)$,
	\[
	P^{\varepsilon}(\,\cdot\,,y)\!\restriction_K\ \, \stackrel{\mbox{\tiny{\tiny \rm{\!loc}}}}{\rightrightarrows}\  P(\,\cdot\,,y)\!\restriction_K \:. 
	\]
	The same holds true for the partial derivatives of any order.
\end{Lemma}

We next derive a few properties of the regularized fermionic projector which will turn
out to be useful later on. The smooth function $P^{2\varepsilon}(y,x)$  can be written as
\begin{equation}\label{espressioneP}
\begin{split}
P^{2\varepsilon}(x,y)=\sum_{j=0}^3v_j(x-y)\: \gamma^j+\beta(x-y) \:,
\end{split}
\end{equation}
where we introduced the functions
\vspace{0.05cm}
\begin{equation}\label{functionsvbeta}
	\begin{split}
		v_0(z)&=-\frac{1}{2}\int_{\R^3}\frac{d^3\V{k}}{(2\pi)^4}\,\mathfrak{g}_{\varepsilon}(\V{k})^2\, e^{-i(-\omega(\V{k})t_z-\V{k}\cdot\V{z})}\\
		v_\alpha(z)&= \frac{1}{2}\int_{\R^3}\frac{d^3\V{k}}{(2\pi)^4}\frac{k_\alpha}{\omega(\V{k})}\,\mathfrak{g}_{\varepsilon}(\V{k})^2\, e^{-i(-\omega(\V{k})t_z-\V{k}\cdot\V{z})}\\
		\beta(z)&=\frac{1}{2}\int_{\R^3}\frac{d^3\V{k}}{(2\pi)^4}\frac{m}{\omega(\V{k})}\,\mathfrak{g}_{\varepsilon}(\V{k})^2\, e^{-i(-\omega(\V{k})t_z-\V{k}\cdot\V{z})} \:.
	\end{split}
\end{equation}
In the next lemma we collect a few basic properties of these functions. We denote the complex conjugate of a complex number $a$ by $\overline{a}$.
\begin{Lemma}\label{lemmavbeta}
	The following statements hold:
\begin{itemize}[leftmargin=2.5em]
	\vspace{0.2em}
		\item[\rm{(i)}] The functions $\beta$ and~$v_\mu$ belong to the class~$C^\infty(\R^{1,3},\C)$,\\[-0.4em]
		\item[\rm{(ii)}] For every $z\in\R^{1,3}$ and~$\alpha \in \{1,2,3\}$,
\[ \overline{\beta(x)}=\beta(-x)\:,\quad \overline{v_\alpha(x)}=v_\alpha(-x) \quad \text{and} \quad
v_\alpha(0)=0\:. \]
		\item[\rm{(iii)}] $0<\beta(0)<|v_0(0)|$.
\end{itemize}
\end{Lemma}
\begin{proof}
	Point (i) can be proved using Lebesgue's dominated convergence theorem (or more precisely \cite[Theorem 1.88]{moretti-book}) and \cite[Lemma 8.1]{oppio}. Point (ii) follows by direct inspection, noting that the integrands of $v_\alpha(0)$ are odd with respect to $\V{k}$. Point (iii) follows from the fact that $m/\omega(\V{k})<1$ on $\R^3\setminus\{0\}$. 
\end{proof}

\begin{Remark}\label{remarkbehavior}
	As a consequence of point {\rm{(ii)}} of the above lemma, we see that the diagonal elements read
	$$
	P^{2\varepsilon}(x,x)=\frac{1}{2(2\pi)^4}\left(-\|\mathfrak{g}_\varepsilon^2\|_{\scL^1}\gamma^0+m\left\|\frac{\mathfrak{g}_\varepsilon^2}{\omega}\right\|_{\scL^1}\bI_4\right).
	$$
	The eigenvalues of this matrix are given by
	$$
	\nu^\pm(\varepsilon):=\frac{1}{2(2\pi)^4}\left(\pm\|\mathfrak{g}_\varepsilon^2\|_{\scL^1}+m\left\|\frac{\mathfrak{g}^2_\varepsilon}{\omega}\right\|_{\scL^1}\right).
	$$
	and fulfill the bounds
	$$
	\nu^+(\varepsilon)>0,\  \nu^-(\varepsilon)<0\ \mbox{ and }\ -\nu^-(\varepsilon)< \nu^+(\varepsilon).$$ 
	The dependence on the parameter $\varepsilon$ can be made more explicit  as follows (see Appendix~\ref{secappendix}). There exist functions $f,g\in {\mathscr{O}}(1)$ which are strictly positive together with their limits as~$\varepsilon\searrow 0$, such that
	\begin{equation}\label{eigenvaluesform}
	\nu^\pm(\varepsilon)=\pm \frac{1}{\varepsilon^3}\,g(\varepsilon)+\frac{m}{\varepsilon^2}\,f(\varepsilon).
	\end{equation}
	In particular, the leading order of both eigenvalues scales~$\sim \varepsilon^{-3}.$ It is worth mentioning that a different choice of the regularization operators would only affect the form of the functions $f,g$, but lead to the same scaling behavior as in~\eqref{eigenvaluesform}.
\end{Remark}
\vspace{0.1cm}

For the following analysis, it is useful to introduce the wave functions
\beq \label{uchidef}
u^\varepsilon_{x,\chi} := P^\varepsilon(\,\cdot\,,x)\: \chi \qquad \text{with} \qquad \chi \in \C^4\ \mbox{ and }\ x=(x^0,\V{x})\in\R^{1,3}\:.
\eeq
Clearly, these wave functions are solutions of the Dirac equation. More precisely, they belong to $\H_m^-\cap C^\infty(\R^{1,3},\C^4)$ (see
Proposition~\ref{propositioncontinuityP1}~(i)-(iii)).
Without a regularization, these wave functions are singular on the null cone
centered at~$x$. Indeed, they are the distributional solutions of the Cauchy problem
for initial data on the Cauchy surface $\{t=x^0\}$ obtained by projecting the distribution~$-\gamma^0 \chi \,\delta^{(3)}_{\V{x}}$ to the negative energy subspace
(for a few more details see~\cite[Remark~2.28~(v)]{oppio}).
In this sense, they can be regarded as the Dirac solutions of negative energy which are as far as possible
``concentrated'' at the spacetime point~$x$.
With a regularization, this qualitative picture remains valid, except that the solution is ``smeared out''
on the scale~$\varepsilon$.

We proceed by proving a useful few mathematical identities.
We denote the usual indefinite inner product on spinors by
\beq \label{sspMink}
\Sl \psi | \phi \Sr := \psi^\dagger \gamma^0 \phi
\eeq and refer to it as the {\em{spin scalar product}}.
\begin{Prp}\label{formulaprodottoscalarelocalizzati}
	The following statements hold:
	\vspace{0.1cm}
	\begin{itemize}[leftmargin=2.5em]
		\item[{\rm (i)}] 	For any $\chi,\zeta\in\C^4$ and $x, y\in\R^{1,3}$,
		$$\\[0.2em]
		2\pi\,\langle u_{x,\chi}^\varepsilon\:\big|\:u_{y,\zeta}^\varepsilon\ra=-
		\Sl \chi\,|\, P^{2\varepsilon}(x,y)\, \zeta \Sr \:. $$
		\item[{\rm (ii)}] In particular, for every $\chi\in\C^4$ and $x\in\R^{1,3}$, the norm is bounded by
			$$
		2\pi\,\|u_{x,\chi}^\varepsilon\|^2
		=-\,\Sl \chi \,|\, P^{2\varepsilon}(x,x)\, \chi \Sr \:.$$
		Moreover, the following bounds hold,
\begin{equation}\label{estimatenormlocstate}
		\sqrt{-\nu^-(\varepsilon)}\,|\chi|\le \sqrt{2\pi}\,\|u_{x,\chi}^\varepsilon\|\le \sqrt{\nu^+(\varepsilon)}\,|\chi|
\end{equation}
		(where~$|\chi|$ is the norm in~$\C^4$ of the spinor~$\chi$ and $\nu^\pm(\varepsilon)$ as defined in \eqref{eigenvaluesform})
	\end{itemize}

\end{Prp}
\begin{proof}
We begin with the proof of~(i). In preparation, note that $u^\varepsilon_{z,\phi}\in\hat{\rm E}(\mathcal{S}(\R^3,\C^4))$. Indeed, as an immediate consequence of~\eqref{expressionEonS} and~\eqref{bidistributionP},
	$$
	u^\varepsilon_{z,\phi}=\hat{\mathrm{E}}(\varphi),\quad 
\varphi(\V{k}):=-(2\pi)^{-5/2}\,\mathfrak{g}_\varepsilon\,p_-\,\gamma^0\,\phi\,\,e^{-i(\omega(\V{k})t_z+\V{k}\cdot\V{z})}
	$$ 
	with $\varphi\in\mathcal{S}(\R^3,\C^4)$.
Using that the map $\hat{\mathrm{E}}$ is an isometry of Hilbert spaces, we have
	(see also~\cite[Lemma~1.2.8]{cfs})
	\begin{equation*}\label{azioneAf}
		\begin{split}
\big\la P^{\varepsilon}(\,\cdot\,,x)\,\chi \,|\, P^\varepsilon(\,\cdot\,,y)\,\zeta\big\ra
			&=(2\pi)^{-5}\,\chi^\dagger\int_{\R^3}\, \mathfrak{g}_\varepsilon(\V{k})^2\, \gamma^0\, p_-(\V{k})\, \gamma^0\zeta\ e^{-ik\cdot (x-y)}\, d^3\V{k}=\\
			&=-(2\pi)^{-1}\,\chi^\dagger\gamma^0\, P^{2\varepsilon}(x,y)\, \zeta \:.
		\end{split}
	\end{equation*}
	This concludes the proof of (i). Since the first statement of~(ii) is a direct consequence of~(i), let us prove the final inequalities.
Writing the norm explicitly, we obtain
\begin{align*} 
\|P^{\varepsilon}(\,\cdot\,,x)\chi\|^2&= -(2\pi)^{-1}\,\chi^\dagger\gamma^0\, P^{2\varepsilon}(x,x)\, \chi = \\
&= {(2\pi)^{-1}\,(-\nu^-(\varepsilon)})\,(\chi_0^2+\chi_1^2)+(2\pi)^{-1}\,\nu^+(\varepsilon)\,(\chi_2^2+\chi_3^2) \:.
\end{align*}
The inequalities in point (ii) follow immediately by noticing that $-\nu^-(\varepsilon)< \nu^+(\varepsilon)$.
\end{proof}
\begin{Remark}
The inequalities of Proposition \ref{formulaprodottoscalarelocalizzati}~(ii) show that the norm of the vector $u_{x,\chi}^\varepsilon$ diverges in the limit $\varepsilon \searrow 0$. The order of this divergence is $\sim \varepsilon^{-3/2}$,
as already noted in Remark~\ref{remarkbehavior}.	
\end{Remark}

In the following sections, we will show how the above divergence
can be exploited for the detection of the light cone. Acting on these vectors with operators localized away from the null cone, this blow-up can be suppressed. On the other hand, 
this is no longer possible in a neighborhood of the null cone. This will give a characterization of the light cone in
terms of vectors and algebras.

\subsection{Basics on the Continuum Limit Analysis with $i \varepsilon$-Regularization}
The {\em{formalism of the continuum limit}} gives a systematic procedure for analyzing
composite expressions in the regularized kernel~$P^\varepsilon(x,y)$ in the limit~$\varepsilon \searrow 0$.
We now explain a few basic ideas and results from this analysis
(for the general context and details see~\cite[Chapter~4]{pfp} or~\cite[Sections~2.4 and 3.5]{cfs}).
The first step is the {\em{light-cone expansion}}, where one expands~$P(x,y)$ in orders of the singularity
on the light cone (starting from the most singular contribution, then the next lower singularity, etc.).
Then one regularizes each term of this expansion.
Since every term in the resulting {\em{regularized light cone expansion}} is smooth, one can form
composite expressions (like the close chain its eigenvalues and eigenvectors, etc.).
In the limit~$\varepsilon \searrow 0$, these expressions typically diverge. But with the formalism
of the continuum limit one can determine and compute the leading singular behavior on the light
cone asymptotically in this limit.

In the Minkowski vacuum, the regularized light cone expansion reads
\[ 
P^\varepsilon(x,y)=\frac{i\gamma^j\xi_j^\varepsilon}{2}\sum_{n=0}^\infty\frac{m^{2n}}{n!}\,T_\varepsilon^{(-1+n)}(\xi)+\sum_{n=0}^\infty\frac{m^{2n+1}}{n!}T_\varepsilon^{(n)}(\xi)+f(\varepsilon,\xi) \:, \]
where we set~$\xi:=y-x=(\xi_0,\boldsymbol{\xi})$ and~$\xi^\varepsilon:=\xi-i\varepsilon\,e_0$,
and~$f$ is a smooth function on $(-\varepsilon_{\max},\varepsilon_{\max})\times\R^{1,3}$.
Here the functions~$T_\varepsilon^{(n)}$ are pointwise divergent in the limit $\varepsilon\searrow 0$. The order of the divergence decreases as $n$ increases.
In particular, the two most divergent terms are (see~\cite[equations~(2.4.7) and~(2.4.9)]{cfs})
\begin{equation}\label{explicitformT}
	\begin{split}	
		T_\varepsilon^{(-1)}(\xi)&:=-\frac{1}{2\pi^3}\frac{1}{(\xi^\varepsilon\cdot\xi^\varepsilon)^2}=-\frac{1}{2\pi^3}\frac{1}{\big((\xi_0-i\varepsilon)^2-\boldsymbol{\xi}^2\big)^2}\,,\\
		T_\varepsilon^{(0)}(\xi)&:=-\frac{1}{8\pi^3}\frac{1}{\xi^\varepsilon\cdot\xi^\varepsilon}=-\frac{1}{8\pi^3}\frac{1}{(\xi_0-i\varepsilon)^2-\boldsymbol{\xi}^2}\,.
	\end{split}
\end{equation}
Now one can form the product~$A_{xy}^\varepsilon:=P^\varepsilon(x,y)P^\varepsilon(y,x)$ and
expand.
In particular, expanding up to the linear terms in the mass, we obtain
$$
A_{xy}^\varepsilon=a_0^\varepsilon(\xi)+m\,a_1^\varepsilon(\xi)
$$
where the functions $a_0$ and $a_1$ are given by
\begin{equation*}
\begin{split}
a_0^\varepsilon(\xi)&:=+\frac{\gamma^j\xi_j^\varepsilon\,\gamma^i\overline{\xi_i^\varepsilon}}{4}\,T^{(-1)}_\varepsilon(\xi)\,\overline{T^{(-1)}_\varepsilon(\xi)}\\
a_1^\varepsilon(\xi)&:= \frac{i\gamma^j}{2} \left[\xi^\varepsilon_j\,T^{(-1)}_\varepsilon(\xi)\, \overline{T^{(0)}_\varepsilon(\xi)}-\overline{\xi^\varepsilon_j}\,T_\varepsilon^{(0)}(\xi)\,\overline{T^{(-1)}_\varepsilon(\xi)}\,\right] .
\end{split}
\end{equation*}
These terms have the same scaling behavior in~$\varepsilon$.
In order to justify that the higher orders in the mass can be omitted,
a clean and systematic procedure is to use that in the formalism of the continuum
limit, these contributions are of lower degree on the light cone (for details see~\cite[Section~2.4]{cfs}).
For the purpose of this paper, it suffices to note that all contributions of higher order in the mass
involve at least one scaling factor~$\varepsilon m$
and can thus be absorbed into the error term in Theorem~\ref{theoremasymptotics} below.

Moreover, for our purposes it suffices to consider the matrix
element~$\mathfrak{e}_1^\dagger\, A_{xy}^\varepsilon\,\mathfrak{e}_3$
(where $\{\mathfrak{e}_\mu\:|\: \mu=1,2,3,4\}$ denotes the canonical basis of $\C^4$). 
Thus our goal is to estimate the $\varepsilon$-behavior of integrals of the form
\begin{equation}\label{integral}
\int_{\R^4} f(y)\, \mathfrak{e}_1^\dagger\, A_{0y}^\varepsilon\,\mathfrak{e}_3\, d^4y,
\end{equation}
for functions $f\in C_0^\infty(\R^{1,3},\R)$ which are supported in a neighborhood of a non-zero null vector.
Making use of the definitions \eqref{explicitformT} and noting that 
$
\mathfrak{e}_1\, \gamma^j\,\mathfrak{e}_3=\delta^j_{3}$
and
$\mathfrak{e}_1^\dagger\,\gamma^j\gamma^i\,\mathfrak{e}_3=\delta_{0}^j\,\delta_3^i-\delta_3^j\,\delta_0^i,
$
a direct computation yields
\begin{equation}\label{expressionstoevaluate}
	\begin{split}
		\mathfrak{e}_1^\dagger\, a_0^\varepsilon(\xi)\,\mathfrak{e}_3&=-\frac{i}{8\pi^6}\frac{\xi_3\,\varepsilon}{|(\xi_0-i\varepsilon)^2-\boldsymbol{\xi}^2|^4}	\\
	\mathfrak{e}_1^\dagger\, a_1^\varepsilon(\xi)\,\mathfrak{e}_3&=-\frac{1}{8\pi^6}\frac{\xi_0\,\xi_3\,\varepsilon}{|(\xi_0-i\varepsilon)^2-\boldsymbol{\xi}^2|^4}\:.	
\end{split}	
\end{equation}
We see that the two terms involving~$a_0^\varepsilon(\xi)$ and~$a_1^\varepsilon(\xi)$ are 
can be distinguished because the first term is real, whereas the second is purely imaginary.
This makes it possible to evaluate the corresponding contributions to the integral~\eqref{integral} separately.

In order to analyze the resulting integrals, it is convenient to begin with the following general result. 
\begin{Lemma}\label{teoremaandamentoasintotico}
Let $x_0\in L_0\setminus\{0\}$ and $\delta>0$ such that $0\not\in \overline{B_\delta(x_0)}$. Moreover,
let $f\in C_0^\infty(B_r(x_0),\R^+)$ such that $f(x_0)\neq 0$. Then
	$$
	\int_{\R^4} f(x)\, \frac{1}{|(t-i\varepsilon)^2-\V{x}^2|^4}\, d^4x= h(\varepsilon)\,\frac{1}{\varepsilon^3} + \O(\varepsilon^{-2}) \:,
	$$
	where $h\in \O(1)$ and $\lim_{\varepsilon\searrow 0}h(\varepsilon)\neq 0$.
\end{Lemma}
\begin{proof}
	For simplicity of notation, we denote the norm $|\V{x}|$ by $r$. As a first step, we restate the denominator in the integrand in a more convenient way. Also, without loss of generality we can assume that $x_0$ belongs to the future-directed half of $L_0$. Then
	\begin{equation*}
	\begin{split}
		&\frac{1}{|(t-i\varepsilon)^2-r^2|^4}=\frac{1}{((t-i\varepsilon)^2-r^2)^2\,((t+i\varepsilon)^2-r^2)^2}=\\
		&=\frac{1}{(2r)^4}\left(\frac{1}{(t-i\varepsilon)-r}-\frac{1}{(t-i\varepsilon)+r}\right)^2\,\left(\frac{1}{(t+i\varepsilon)-r}-\frac{1}{(t+i\varepsilon)+r}\right)^2=\\
		&=\frac{1}{(2r)^4}\left[\frac{1}{(t-i\varepsilon)-r}\left(1-\frac{(t-i\varepsilon)-r}{(t-i\varepsilon)+r}\right)\right]^2\,\left[\frac{1}{(t+i\varepsilon)-r}\left(1-\frac{(t+i\varepsilon)-r}{(t+i\varepsilon)+r}\right)\right]^2=\\
		&=\frac{1}{((t-i\varepsilon)-r)^2\,((t+i\varepsilon)-r)^2}\,\underbrace{\frac{1}{(2r)^4}\left[\left(1-\frac{(t-i\varepsilon)-r}{(t-i\varepsilon)+r}\right)^2 \left(1-\frac{(t+i\varepsilon)-r}{(t+i\varepsilon)+r}\right)^2\right]}_{:=h_\varepsilon(t,\V{x})} .
	\end{split}
	\end{equation*}
Now, notice that for any $\varepsilon\in [0,\varepsilon_{\max})$, the function $h_\varepsilon$ is infinitely differentiable on $B_r(x_0)$ and does not vanish on $B_r(x_0)\cap L_0$, i.e. at $t=r$.
In particular we can include the factor $h_\varepsilon$ into the function $f$. The product $f_\varepsilon:=fh_\varepsilon$ is still an element of  $C_0^\infty(B_\delta(x_0),\R)$ and $fh_\varepsilon(x_0)\neq 0$. Moreover, the pointwise limit $\varepsilon\searrow0$ of this function together with all its derivatives is well-defined. In particular, for every $\V{x}\in B_\delta(\V{x}_0)$ we have
$$
f_\varepsilon(|\V{x}|,\V{x})\to \frac{f(|\V{x}|,\V{x})}{16|\V{x}|^4}\ge 0\quad\mbox{as }\varepsilon\searrow 0,
$$ 
which is non-zero at $x_0$.
Finally, notice that 
\begin{equation}\label{boundderivatives}
\sup_{\varepsilon\in [0,\varepsilon_{\max})}\|\partial^\alpha f_\varepsilon\|_\infty<\infty\quad\mbox{for any multi-index $\alpha$.}
\end{equation} 
Therefore, we can focus on the integral
\[ 
S_\varepsilon:=\int_{\R^4} f_\varepsilon(x)\, \frac{1}{((t-i\varepsilon)-r)^2\,((t+i\varepsilon)-r)^2}\, d^4x \:. \]
This integral can be factorized into a time and a spatial integral according to Fubini's theorem,
\begin{equation}\label{integraltoevaluate}
S_\varepsilon=\int_{B_\delta(\V{x_0})}d^3\V{x}\int_{\R}dt\,\frac{f_\varepsilon(t,\V{x})}{((t-i\varepsilon)-r)^2\,((t+i\varepsilon)-r)^2} \:,
\end{equation}
where we chose as effective domain of integration the tube $\R\times B_\delta(\V{x_0})$. 
For any $\V{x}\in B_\delta(\V{x_0})$, the intersection $\R\times\{\V{x}\}\cap L_0$ determines a unique point  $(r,\V{x})\in L_0$, where again $r:=|\V{x}|$. We now consider the Taylor expansion of $f_\varepsilon$ in the time variable on $\R$ up to the first order around the point $r$.
So, for every $(t,\V{x})\in \R\times\{\V{x}\}$ there exists $c(\varepsilon, t,\V{x})\in \R$ such that
\[ 
f_\varepsilon(t,\V{x})=f_\varepsilon(r,\V{x})+\partial_0 f_\varepsilon(r,\V{x})(t-r)+ \frac{1}{2}\,\partial_{0}^2f_\varepsilon(c(\varepsilon, t,\V{x}),\V{x})\,(t-r)^2 \:. \]
The last summand is the remainder, and we denote it by~$R_\varepsilon(t,\V{x})$. Note that the expression $\partial_{0}^2f_\varepsilon(c(\varepsilon,t,\V{x}),\V{x})$ can be bounded uniformly from above by some constant $M$ which is independent of $\V{x}$ and $\varepsilon$, as follows from \eqref{boundderivatives}. Using this expansion in~\eqref{integraltoevaluate}, we obtain
\begin{equation*}
\begin{split}
	\int_{\R}\frac{f_\varepsilon(t,\V{x})}{((t-i\varepsilon)-r)^2\,((t+i\varepsilon)-r)^2}\, dt=f_\varepsilon(r,\V{x})\, I_0(r,\varepsilon)+\partial_0 f_\varepsilon(r,\V{x})\, I_1(r,\varepsilon)+ R_\varepsilon(\V{x})
\end{split}
\end{equation*}
where $R_\varepsilon(\V{x})$ is the part of the integral which contains the remainder, while
\begin{align*}
I_0(r,\varepsilon)&:=\int_{\R}\frac{1}{((t-i\varepsilon)-r)^2\,((t+i\varepsilon)-r)^2}\, dt=\frac{4\pi}{\varepsilon^3} \\
I_1(r,\varepsilon)&:=\int_{\R}\frac{t-r}{((t-i\varepsilon)-r)^2\,((t+i\varepsilon)-r)^2}\, dt=0 \:,
\end{align*}
where the last integrals were computed with residues.

It remains to estimate the remainder. Given $\delta>0$ small enough,
\begin{align*}
|R_\varepsilon(\V{x})|\le&\,\int_\R|R_\varepsilon(t,\V{x})|\,dt\le \frac{M}{2}\int_\R \frac{(t-r)^2}{|((t-i\varepsilon)-r)\,((t+i\varepsilon)-r)|^2}\, dt=\\[0.1em]
=&\ \frac{M}{2}\int_{r-\delta}^{r+\delta} \frac{(t-r)^2}{((t-r)^2+\varepsilon^2)^2}\, dt+ \frac{M}{2}\int_{\R\setminus [r-\delta,r+\delta]}\frac{(t-r)^2}{((t-r)^2+\varepsilon^2)^2}\, dt\le \\[0,1em]
\le&\ \frac{M}{2}\int_{r-\delta}^{r+\delta} \frac{1}{(t-r)^2+\varepsilon^2}\, dt+ \frac{M}{2}\int_{\R\setminus [r-\delta,r+\delta]}\frac{1}{(t-r)^2+\varepsilon^2}\, dt\le \\[0,1em]
\le&\ \frac{M\delta}{\varepsilon^2}+\frac{M}{2}\int_{\R\setminus [r-\delta,r+\delta]}\frac{1}{(t-r)^2}\, dt\le \frac{M\delta}{\varepsilon^2}+\frac{M}{\delta} \:.
\end{align*}
Using these results in the integral \eqref{integraltoevaluate}, we get
\[ S_\varepsilon= \frac{4\pi}{\varepsilon^3} \int_{B_\delta(\V{x}_0)}f_\varepsilon(|\V{x}|,\V{x})\,d^3\V{x}+\int_{B_\delta(\V{x}_0)} R_\varepsilon(\V{x})\, d^3\V{x} \:, \]
where the second integral can be bounded from above by a term proportional to $\varepsilon^{-2}$, while
$$
\int_{B_\delta(\V{x}_0)}f_\varepsilon(|\V{x}|,\V{x})\,d^3\V{x}\to \int_{B_\delta(\V{x}_0)}\frac{f(|\V{x}|,\V{x})}{16|\V{x}|^4}\,d^3\V{x}>0\quad \mbox{as }\varepsilon\searrow 0 \:.
$$ 
This concludes the proof.
\end{proof}

This lemma can be applied to the integral~\eqref{integral}, using the expressions \eqref{expressionstoevaluate}. Note that in \eqref{expressionstoevaluate} an additional $\varepsilon$ factor appears, which lowers the order of the divergence. We thus obtain the following result.
\begin{Thm}\label{theoremasymptotics}
	Let $y_0\in L_0\setminus\{0\}$ with~$y_3\neq 0$. Let $\delta>0$ such that $0\not\in \overline{B_\delta(y_0)}$. Finally, let $f\in C_0^\infty(B_r(y_0),\R_+)$ such that $f(y_0)\neq 0$. Then
	$$
\int_{\R^4} f(y)\, \mathfrak{e}_1^\dagger\, A_{0y}^\varepsilon\,\mathfrak{e}_3\, d^4y= h(\varepsilon)\frac{1}{\varepsilon^2} + \O(\varepsilon^{-1}) \:,
	$$
	where $h\in \O(1)$ and $\lim_{\varepsilon\searrow 0}h(\varepsilon)\neq 0$.
\end{Thm}

\subsection{Causal Fermion Systems in Minkowski Space}\label{subsectionCFSM}
From now on, we always restrict attention to the negative energy subspace $\scH_m^-$. Once the regularization $\gR_\varepsilon$ has been introduced, it is possible to define a function
$$
F^\varepsilon:\R^{1,3}\rightarrow\Lin(\scH_m^-)
$$
which encodes information on the local behavior of the wave functions in~$\scH_m^-$ at any point $x\in\R^{1,3}$
via the identity
\begin{equation}\label{defF}
	\langle u| F^\varepsilon(x)v\rangle=- \Sl \gR_{\varepsilon}u(x) \:|\:\mathfrak{R}_\varepsilon v(x) \Sr\quad\mbox{for any }u,v\in\scH_m^- \:.
\end{equation}
The function~$F^\varepsilon$ is referred to as the \textit{local correlation map}.
The construction and a few properties are summarized in the following theorem.
\begin{Thm}\label{listpropertiesLCF}
	There exists a unique function $F^\varepsilon:\R^{1,3}\rightarrow\Lin(\scH_m^-)$ which fulfills~\eqref{defF}. Moreover the following statements are true.
\vspace{0.2em}
\begin{itemize}[leftmargin=2.5em]
		\item[\rm{(i)}] The image $F^\varepsilon(\R^{1,3})$ is closed. Moreover, the mapping $F^\varepsilon$ is a homeomorphism to its image.\\[-0.2cm]
		\item[\rm{(ii)}] The operator~$F^\varepsilon(x)$ is selfadjoint and
		$$
		F^\varepsilon(x)\,u=2\pi\,P^{\varepsilon}(\,\cdot\,,x)\,\big(\gR_{\varepsilon} u\big)(x) \:.
     	$$\\[-0.9cm]
		\item[\rm{(iii)}] Its image~$\ran F^\varepsilon(x) \subset \H_m^-$ is four-dimensional and 
		$$
	    \ran F^\varepsilon(x)=\mathrm{span}\{	u^\varepsilon_{x,\chi}\:|\:\chi\in \C^4\}\:.
		$$
		\item[\rm{(iv)}] The operator~$ F^\varepsilon(x)|_{\ran F^\varepsilon(x)}$ has two two-fold degenerate eigenvalues given by $\nu^\pm(\varepsilon)$,
		with corresponding eigenvectors\footnote{For simplicity of notation, we here denote the vector $u^\varepsilon_{x,\mathfrak{e}_\mu}$ by $u^\varepsilon_{x,\mu}$.}
\begin{align*} 
		u^\varepsilon_{x,\mu}:&=P^{\varepsilon}(\,\cdot\,,x)\,\mathfrak{e}_\mu,\quad \mu\in\{1,2,3,4\}\\ 
		F^\varepsilon(x)u^\varepsilon_{x,\mu}&=
		\begin{cases}
			2\pi\,\nu^-(\varepsilon)\,u^\varepsilon_{x,\mu} & \mu=1,2\\ 
			2\pi\,\nu^+(\varepsilon)\,u^\varepsilon_{x,\mu} & \mu=3,4
		\end{cases}
\end{align*}
(where $\{\mathfrak{e}_\mu\:|\: \mu\in \{1,2,3,4\}\}$ denotes the canonical basis of $\C^4$). In particular, 
		$$
		\|F^\varepsilon(x)\|=2\pi\,\nu^+(\varepsilon) \:.
		$$\\[-0.8cm]
		\item[\rm{(v)}] The operator $F^\varepsilon(x)$ belongs to $\F$.\\[-0.2cm]
		\item[\rm{(vi)}] The unitary operators~$\mathrm{U}_a$ defined in~\eqref{Udef}
		describe translations of the operators~$F^\varepsilon(x)$ by
\[ F^\varepsilon(x+a)=\mathrm{U}^\dagger_a\,  F^\varepsilon(x)\, \mathrm{U}_a \:. \]
	\end{itemize}
\end{Thm} 
\begin{proof}
The proof of all the statements can be carried out by adapting the results in \cite[Section 1.2]{cfs} and in \cite[Sections 4.2, 5.4 and 6.2.1]{oppio} to the cutoff function chosen in this paper.
\end{proof}

From point (iv) we see that every $ F^\varepsilon(x)$ has the spectral decomposition
\begin{equation}\label{decompositionF}
F^\varepsilon(x)= F^\varepsilon_+(x)+ F_-^\varepsilon(x):=2\pi\,\nu^+(\varepsilon)\pi_x^++2\pi\,\nu^-(\varepsilon)\pi_x^-
\:,
\end{equation}
where $\pi_x^\pm$ are the projection operators on the
corresponding two-dimensional eigen\-spa\-ces. In particular, we obtain
\beq \label{trvacdef}
\mbox{tr}_{\rm{vac}}^\varepsilon :=\mbox{tr}( F^\varepsilon(x))= 4\pi(\nu^+(\varepsilon)+\nu^-(\varepsilon))=\frac{2 m}{(2\pi)^{3}}\left\|\frac{\mathfrak{g}^2_\varepsilon}{\omega}\right\|_{\scL^1}= 8\pi\,\frac{m}{\varepsilon^{2}}\,f(\varepsilon) \:.
\eeq
The trace of the local correlation operators is independent of the spacetime point. This is an obvious consequence of translation invariance as made explicit by point (vi) above.

By means of the topological identification of $\R^{1,3}$ and its image under~$F^\varepsilon$,
it is possible to introduce a causal fermion system by taking the push-forward of the Lebesgue measure of $\R^{1,3}$,
$$\rho_{\rm{vac}}:=F^\varepsilon_*(\mu).$$
\begin{Def} \label{defregDir}
	The causal fermion system $(\scH_m^-, \F, \rho_{\text{vac}})$ is referred to as the
	{\bf{regularized Dirac sea vacuum}}.
\end{Def}
In particular, thanks to point (i) of Theorem \ref{listpropertiesLCF}, we see that the mapping~$F^\varepsilon$
gives a one-to-one correspondence between points in Minkowski space
and points in the support of the measure~$\rho_{\rm{vac}}$,
\beq \label{id1}
M_{\rm{vac}}:=\supp\rho_{\rm{vac}}=F^\varepsilon(\R^{1,3}) \:.
\eeq
Moreover, for every $x\in\R^{1,3}$ there is a canonical identification of the space of Dirac spinors
with the spin space at the corresponding point of~$F^\varepsilon(x) \in M_{\rm{vac}}$
\beq \label{id2}
\Phi_x : S_{F^\varepsilon(x)}=\ran F^\varepsilon(x) \ni u\mapsto \gR_{\varepsilon} u(x)\in  \C^4
\eeq
(for details see~\cite[Proposition~1.2.6]{cfs} or~\cite[Theorem~4.16]{oppio}).
This identification preserves the spin scalar products as given by~\eqref{ssp} and~\eqref{sspMink};
cf.~\eqref{defF}. 
Every spin space $S_{F^\varepsilon(x)}$ is four-dimensional, hence it has maximal rank as an element of $\F$ (see Definition \ref{defCFS}). This fact is referred to as the \textit{regularity} of the causal fermion system. Equivalently, this property can be restated as follows (see \cite[Lemma 4.17 and Section 5.1]{oppio}):
\begin{center}
	\textit{For every $x\in\R^{1,3}$ and for every $\mu\in\{1,2,3,4\}$ there exists  \\[0.2em] a smooth $u_\mu\in \scH_m^-$
		such that $u_\mu(x)=\mathfrak{e}_\mu\:,$}
\end{center}
where~$\{\mathfrak{e}_\mu,\ \mu\in \{1,2,3,4\}\}$ denotes again the canonical basis of~$\C^4$.

As a final remark, we point out that the identification~\eqref{id2} allows for an explicit realization of the kernel of the fermionic projector as defined abstractly in~\eqref{Pxydef} in terms of the regularized bi-distribution introduced in~\eqref{bidistributionP} (for details see for example~\cite[Theorem 5.18]{oppio}),
\begin{equation}\label{id3}
\Phi_x\, \mathrm{P}\big(F^\varepsilon(x), F^\varepsilon(y) \big)\,\Phi_y^{-1}=2\pi\,P^{2\varepsilon}(x,y) \:.
\end{equation}
This identity will be exploited in the next section in the proof of irreducibility of the local algebras.
Indeed, this formula allows us to translate the irreducibility of the Dirac matrices on $\C^4$ which appear in $P^{2\varepsilon}(x,y)$ to  irreducibility of the operators $ F^\varepsilon(x)$ on~$\scH_m^-$.

The identification~\eqref{id3} also gives a corresponding realization of the
closed chain~$\mathrm{A}_{\mathrm{x}\mathrm{y}}$ as defined in~\eqref{Axydef} in terms of the product of the two regularized bi-distributions,
\begin{equation}\label{id4}
\begin{split}
A_{xy}^\varepsilon&:=P^{2\varepsilon}(x,y)P^{2\varepsilon}(y,x)=(2\pi)^{-2}\,\Phi_x\, \mathrm{P}\big(F^\varepsilon(x), F^\varepsilon(y) \big)\,\mathrm{P}\big(F^\varepsilon(y), F^\varepsilon(x) \big)\,\Phi_x^{-1}=\\
&=\Phi_x\,\big((2\pi)^{-2}\, \mathrm{A}_{F^\varepsilon(x)F^\varepsilon(y)}\big)\,\Phi_x^{-1}.
\end{split}
\end{equation}
Computing the eigenvalues of this matrix, one finds that
the causal structure of Definition~\ref{def2} agrees with the causal structure of Minkowski space
in the limit~$\varepsilon \searrow 0$
(see~\cite[Proposition~1.2.10]{cfs}).

When working in Minkowski space, we will often make use of the above identifications~$x \simeq F^\varepsilon(x)$ and~$S_x \simeq \C^4$.
For example, the orthogonal projector onto the image of~$F^\varepsilon(x)$ will be denoted simply by $\pi_x$ instead of $\pi_{F^\varepsilon(x)}$. Likewise, the the spin spaces will simply be denoted by~$S_x$.

\section{Local Algebras for Causal Fermions Systems} \label{secalgebras}
\subsection{Definition of the Local Algebras}
Consider a causal fermion system $(\scH,\F,\rho)$ and a open set $\Omega\subset M:=\supp\rho$.
In view of the later constructions in Minkowski space, it is preferable to
denote all the operators corresponding to spacetime points in~$\Omega$ by
\[  {\mycal X}_\Omega := \{x \:|\: x \in \Omega \} \:. \]
In order to construct an algebra generated by these operators, it is of technical advantage
to smear out the operators by continuous functions with compact support within~$\Omega$,
\[ 
A_f := \int_M\: f(x)\, x\: d\rho(x) \quad f\in C_0^0(\Omega,\C)
\]
(in the sense of a Bochner integral).
It is actually not necessary to consider the whole class of compactly supported continuous function on $\Omega$. In general, it suffices to consider a subset of~$C^0_0(\Omega, \C)$ which is rich enough to allow for a reconstruction of the operators $x$ as limit points, as we now explain.
\begin{Def}
Given~$x_0\in M$, a {\bf{Dirac sequence}} at $x_0$ is a sequence~$(f_n)_{n\in\N}$
of functions in~$C^0_0(M,\R^+)$ such that $\|f_n\|_{\scL^1}=1$ and $\supp f_n\subset B_{1/n}(x_0)$. 
\end{Def} \noindent
It is always possible to construct such an object. Indeed, choosing a mollifier $h\in C^\infty_0((-1,1))$, 
the function
\[ 
f_n(x) := \frac{h\left( n \: \|x-x_0\| \right)}{\int_M h(n\|x-x_0\|)\,d\rho(x)}
\]
is a Dirac sequence at~$x_0$.

We choose a subspace~$C^\flat_0(\Omega, \C)$ in~$C^0_0(\Omega, \C)$ 
(which later on will be the smooth functions in Minkowski space) which contains a Dirac sequence
at every point of $\Omega$ and introduce the set
$$
L_\Omega := \{ A_f \:|\: f \in C^\flat_0(\Omega, \C) \} \:.
$$
We now take the $^*$-algebra generated by this set, defined for a generic subset $\mathcal{A}\subset \Lin(\scH)$
by
	$$
	\langle \mathcal{A} \rangle := \left\{ \sum_{k=1}^N\sum_{i_1,\dots,i_k=1}^{n_k}\lambda ^{i_1,\dots,i_k}A_{i_1}\cdots A_{i_k}\ \bigg|\  N,n_k\in\N,\ A_i\in\mathcal{A}\right\}.
	$$

\begin{Remark}\label{remarknonunital}
	Note that the identity operator is not included in this definition of generated $^*$-algebra $\langle \mathcal{A} \rangle$. This will be crucial later in the analysis of the connection between the local algebras and the light-cone structure in Minkowski space (see for example Remark \ref{remarkreasonnidentity}).
\end{Remark}
\begin{Def}
Let $\Omega\subset M$ be open, then the $^*$-algebra
\[ \A_\Omega := \la L_\Omega \ra \]
is referred to as the {\bf{local algebra}} corresponding to~$C^\flat_0(\Omega, \C) \subset C^0_0(\Omega, \C)$.
\end{Def} \noindent

We now prove that the uniform closure of this local algebra contains all spacetime operators of $\mycal{X}_\Omega$.
\begin{Prp}\label{closurecoincide} The uniform closures of~$\mycal{A}_\Omega$ contains~$\mycal{X}_\Omega$.
Even more,
\[ \overline{\mycal{A}_\Omega} = \overline{\la \mycal{X}_\Omega \ra} \;\subset\; K(\H)\:, \]
where~$K(\H)$ are the compact operators on~$\H$.
\end{Prp}
\begin{proof} We first show that~$\mycal{X}_\Omega \subset \overline{L_\Omega}$. This will imply that~$\la {\mycal X}_\Omega \ra \subset \langle\, \overline{L_\Omega}\,\rangle\subset \overline{\mycal{A}_\Omega}$, where the last inclusion follows from the fact that the closure of a $^*$-algebra is again a $^*$-algebra.
For any~$x_0 \in \Omega$, we choose a Dirac sequence~$f_n$. Then
\begin{equation*}
		\begin{split}
			\left\|A_{f_n}- x_0\right\|&=\left\|\int_\Omega f_n(x)\, x\, d\rho(x)-\int_\Omega f_n(x)\,x_0\, d\rho(x)\right\|=\\
			&=\left\|\int_\Omega f_n(x)( x-x_0)\, d\rho(x)\right\|\le \int_\Omega f_n(x)\| x-x_0 \|\, d\rho(x).
		\end{split}
	\end{equation*}
Since the function $f$ is supported in $B(x_0,1/n)$, it is easy to see that the right-hand side above converges to zero
as $n \rightarrow \infty$. This proves that $ x_0\in \overline{L_\Omega}$. 
	
The opposite inclusion follows immediately, once we have shown that $L_\Omega\subset \overline{ \la \mycal{X}_\Omega \ra}$. So, let $f\in C_0^\flat(\Omega,\C)$. For simplicity of notation, let us denote the (compact) support of $f$ by $R$.  Since the function $F: M\ni  x\mapsto f(x)\, x\in \F$ is continuous on the compact set $R$, it is also uniformly continuous, i.e. for any $\varepsilon>0$ there exists $\delta>0$ such that for any $x,y\in R$, $\|x-y\|\le \delta$ implies $\|F(x)-F(y)\|\le \varepsilon$. Now, choose for any $n\in\N$ a set of points $\{x_i\}_{i=1,\dots,N_n}\subset R$ such that $\{B(x_i,1/n)\}_{i=1,\dots,N_n}$ covers $R$ (this can always be done, thanks to the compactness of $R$). Define the measurable sets $C_{n,i}:= R\cap B(x_i,1/n)\setminus\bigcup_{k=1}^{i-1}B(x_k,1/n)$. These sets give a partition of $R$ into disjoint measurable sets. At this point we define the measurable step function on $M$:
$$
S_n:= \sum_{i=1}^{N_n} F(x_i) \: \chi_{C_{n,i}} \:.
$$
We claim that $S_n$ converges pointwise to $F$. Indeed, if $x\not\in R$, then both $S_n(x)$ and $F(x)$ vanish, so there is nothing to prove. Suppose instead that $x\in R$ and fix $\varepsilon$. Let $\delta>0$ given as above and choose $\bar{n}$ in a way that $1/\bar{n}<\delta$. Then $x\in C_{\bar{n},i_0}$ for some index $i_0$. Then $\|x-x_{i_0}\|\le 1/\bar{n}<\delta$ and so:
$$
\|S_{\bar{n}}(x)-F(x)\|=\|F(x_{i_0})-F(x)\|\le \varepsilon.
$$
The same is of course true for any $n\ge \bar{n}$, showing that $S_n(x)$ converges to $F(x)$.
Finally, notice that $\sup_{x\in R}\|S_n(x)\|\le \sup_{x\in R}\|F(x)\|< \infty $, again by continuity of $F$ and compactness of $R$. Since $R$ is compact and $\rho$ is finite on compact sets,
we can apply Lebesgue's dominated convergence theorem and prove that 
$$
\sum_{i=1}^{N_n}x_i\, f(x_i)\,\rho(C_{n,i})=\sum_{i=1}^{N_n}F(x_i)\rho(C_{n,i})=\int_\Omega S_n(x)\, d\rho(x)\to \int_\Omega F(x)\, d\rho(x)=A_f
$$
Since the elements on the left-hand side belong to $\la \mycal{X}_\Omega \ra$, the result follows.

It remains to show that the operators in~$\overline{\la {\mycal X}_\Omega \ra}$ are compact.
All operators in~${\mycal X}_\Omega$ have finite rank.
Therefore, also the operators in~$\la \mycal{X}_\Omega \ra$ have finite rank.
Taking their closure, we obtain compact operators.
\end{proof}

\begin{Remark}\label{remarkidentity}
We point out that, if~$\H$ is infinite-dimensional, then the closures of the $^*$-algebras~$\mycal{A}_\Omega$
and~$\la \mycal{X}_\Omega \ra$ do not contain the identity operator. 
This follows immediately from the fact that the identity is not a compact operator.
\end{Remark}

\subsection{The Commutation Relations}\label{commutatinrelationsection}

In this short section we compute the abstract commutation relations between different spacetime points of $M$.

\begin{Prp} \label{prpcomm}
For any~$u \in \H$ and~$\x,\y \in M$,
\[ \la u \,|\, [\x,\y] \, u \ra = -2 i\, \im \Sl \psi^u(\x) \,|\, \mathrm{P}(\x,\y)\, \psi^u(\y) \Sr_{\x} \:. \]
\end{Prp}
\Proof The identity follows from the computation
\begin{align*}
\la u \,|\, \x \y \, u \ra &= -\Sl \pi_{\x} u \,|\, \pi_{\x} \y\, \pi_{\y} u \Sr_{\x}
= -\Sl \psi^u(\x) \,|\, \mathrm{P}(\x,\y)\, \psi^u(\y) \Sr_{\x} \\
\la u \,|\, \y\x \, u \ra &= \la \x\y \,u \,|\, u \ra
= \overline{ \la u \,|\, \x\y \,u \ra }
= \overline{-\Sl \psi^u(\x) \,|\, \mathrm{P}(\x,\y)\, \psi^u(\y) \Sr_{\x}} \:.
\end{align*}
\QED

In the next sections we analyze the local algebras and the commutation relations
in the example of the regularized Dirac sea vacuum.
In this example, we shall see that the operators corresponding to spacelike separated
points do in general not commute.
This illustrates that {\em{our local algebras are not to be interpreted as algebras of local observables}}
as considered in quantum field theory.
A major difference between our algebras and the usual algebras of local observables
is that our algebras act on the one-particle Hilbert space instead of the Fock space.
As we shall see, despite these major differences, also our local algebras encode
information on the causal structure. However, this information is not retrieved by looking
at the commutation relations, but instead by considering suitable expectation values.

We remark that the usual algebras of local observables can also
be introduced in the setting of causal fermion systems starting
from the Fock space constructions in~\cite{fockbosonic}.
For bosonic fields, one gets the usual canonical commutation relations,
which vanish for spacelike separated points (see~\cite[Section~7.2]{fockbosonic}).
The generalization to include fermionic fields is currently under investigation.
However, these constructions have a different meaning and purpose
than the algebras constructed here.

\section{Regularized Local Algebras in Minkowski Space} \label{secregalg}
Using the identifications~\eqref{id1} and~\eqref{id2}, 
for the case of the regularized Dirac sea vacuum 
we can work directly in Minkowski space~$\R^{1,3}$.
Therefore, in this concrete example the set $\Omega$ will always refer to an open subset of $\R^{1,3}$. Bearing this identification in mind, we keep the notation $\mycal{X}_\Omega$ and $\mycal{A}_{\Omega}$.

\subsection{The Local Set of Spacetime Operators and the Time Slice Axiom} \label{sec41}
Following our general construction, we collect all the local correlation operators
which belong to a given set $\Sigma \subset \R^{1,3}$,
\[ 
{\mycal X}_\Sigma^\varepsilon :=\{F^\varepsilon(x)\:|\: x\in\Sigma \}  \subset \Lin(\scH_m^-) \:. \]
We are interested in understanding
how much information of the system can be extracted from this collection of operators. 
What could one expect?
Recall that the negative-energy solutions of the Dirac equation have the
strong common property that they vanish almost nowhere:
\begin{Prp} \label{prphegerfeldt}
	Let $\Sigma\subset\R^{1,3}$ be an open set. Let~$u\in\scH_m$ vanish on~$\Sigma$, then~$u$ vanishes identically (up to sets of measure zero),
\beq \label{usigzero}
u\!\restriction_\Sigma\,=0 \qquad \Longrightarrow \qquad u = 0\:.
\eeq
In particular, if two solutions $u,v\in\scH_m^-$ coincide on $\Sigma$, then they coincide everywhere.
The statement~\eqref{usigzero} holds true if $\Sigma$ is a Cauchy surface and $u\in \scH_m\cap C^\infty(\R^{1,3},\C^4)$.
\end{Prp}

\begin{proof}
In the case that~$\Sigma$ is a Cauchy surface,
the result follows from the existence and uniqueness of global solutions of linear hyperbolic systems.
The case of~$\Sigma$ being an open set follows from a generalization of
a result by Hegerfeldt in~\cite{hegerfeldt1974remark}. Indeed, Hegerfeldt proves that
any element $u\in\scH_m^-$ which is supported in a \textit{bounded} set must vanish identically. Here we use a more general statement saying that it suffices for the function to be supported in the complement of a non-empty open set in order to vanish identically. 
Nevertheless, we will keep referring to this result as \textit{Hegerfeldt's theorem}. For the proof, we make use of well-known results in Fourier theory and complex analysis. More precisely, we need the following statement (see \cite[Corollary 3.6]{CB}):
\begin{quote}
\label{propositionBeck}
	{\em{Let $\lambda$ be a complex Borel measure on $\R^4$ with support in the past closed light cone $J_0^-:=\{k\in J_0\:|\: k^0\le 0 \}$. Consider the function
	$$
	f(x):=\int_{\R^4}e^{-ik\cdot x}\, d\lambda(k) \:.
	$$
	If $f$ vanishes on an open connected subset $\Omega\subset \R^4$, then $f \equiv 0$ on~$\R^4$.}}
\end{quote}
\noindent We now apply this statement to our proof. 
Following the discussion in \cite[Section 3.2]{oppio},
 we first regularize the elements of $\scH_m^-$ by convolution with a suitable mollifier $h\in C^\infty_0(B(0,\delta),\R_+)$. This gives rise to a bounded injective operator
\[ 
\scH_m^-\ni v\mapsto h*v\in\scH_m^-\cap C^\infty(\R^{1,3},\C^4) \:. \]
Moreover, this mapping can be reformulated in momentum space similar as in Definition \ref{defreg}.
This gives rise to a function $\mathfrak{g}\in\mathcal{S}(\R^3,\C^4)$ such that
$$
\scH_m^-\ni\hat{\mathrm{E}}(\psi)\mapsto \hat{\mathrm{E}}(\mathfrak{g}\,\psi)\in\scH_m^-\cap C^\infty(\R^{1,3},\C^4)
$$
(cf.\ the end of Section~\ref{sectiondiracequation} for the definition of the operator~$\hat{\mathrm{E}}$). Using that $\mathfrak{g}\,\psi\in L^1(\R^3,\C^4)$, an argument similar as in the proof of Proposition \ref{propositionregularization}~(i) below shows that $\hat{\mathrm{E}}(\mathfrak{g}\,\psi)$ can be written explicitly as 
\begin{equation}\label{measureintegral}
\hat{\mathrm{E}}(\mathfrak{g}\,\psi)(x):=\int_{\R^3}\frac{d^3\V{k}}{(2\pi)^{3/2}}\, \mathfrak{g}(\V{k})\,\psi(\V{k})\,e^{-ik\cdot x}\quad \mbox{(with $k^0=-\omega(\V{k})$)} \:.
\end{equation}
Now, assume there exists $\psi\in\hat{P}_-(\scL^2(\R^3,\C^4))$  such that $\hat{\mathrm{E}}(\psi)\in\scH_m^-$ vanishes almost everywhere on an open set $\Omega\subset\R^{1,3}$. Because the regularization is given as a convolution with a mollifier supported in $B(0,\delta),$ the smooth solution $\hat{\mathrm{E}}(\mathfrak{g}\,\psi)$ vanishes on   the open set
$
\Omega_\delta:=\{x\in\Omega\:|\: d(x,\partial \Omega)>\delta \}\subset\Omega.
$
Now, it is possible to restate the identity~\eqref{measureintegral} as
$$
\hat{\mathrm{E}}(\mathfrak{g}\,\psi)(x)=\int_{\R^4}\,e^{-ik\cdot x}\,d^4\lambda(k) \:,
$$
where $\lambda$ is defined for all Borel subsets $\Delta\subset\R^4$ by
$$
\lambda(\Delta):=\int_{\pi_{\R^3}(\Delta\cap \Omega^-_m)}\frac{\mathfrak{g}(\V{k})\,\psi(\V{k})}{(2\pi)^{3/2}}\, d^4\V{k}\,,
$$
with $\pi_{\R^3}:\R^4\mapsto \R^3$ the canonical projection $(k^0,\V{k})\mapsto \V{k}$ and 
$$
\Omega_m^-:=\{k\in\R^4\:|\:k^2=m^2,\,k^0\le 0\:\}\subset J_0^-
$$  the mass hyperboloid of negative energy. 
Each spinorial component of $\lambda$ defines a complex measure on $\R^4$ with support in $J_0^-$. Therefore, as $\hat{\mathrm{E}}(\mathfrak{g}\,\psi)$ vanishes on $\Omega_\delta$, 
the statement of~\cite[Corollary 3.6]{CB} given above
ensures that $\hat{\mathrm{E}}(\mathfrak{g}\,\psi)=0$ identically. Since the chosen regularization is injective, we
conclude that~$\hat{\mathrm{E}}(\psi)=0$, giving the result.
\end{proof} 
\begin{Remark}
Chapter 3 in \cite{CB} provides a thorough compendium of the main results which concern the impossibility of proper localization schemes in relativistic quantum mechanics.  
A nice analysis on the localization problem for the Dirac equation can be found in~\cite[Section~1.8]{thaller}.
\end{Remark}

\noindent
Proposition \ref{prphegerfeldt} shows that the knowledge of negative-energy solutions of the Dirac equation on any fixed open set
suffices to characterize them globally. This statement is no longer true if the restriction to negative-energy solutions
is dropped. The statement in the case that~$\Sigma$ is a Cauchy surface, however, remains true
for any subspace of Dirac solutions.

The next two results show that the above property is also true on the operator level when a regularization is involved. 
\begin{Lemma}\label{lemmasuffiienza}
Let $\Sigma\subset\R^{1,3}$ be either an open set or the Cauchy surface~$\{t=\text{const}\}$. Then
the spin spaces on~$\Sigma$ span a dense subspace of~$\scH_m^-$,
\[ \overline{ \text{\rm{span}} \bigcup_{x\in\Sigma}S_x } = \scH_m^- \:. \]
\end{Lemma}
\begin{proof}
Let $S_\Sigma$ denote the closure of the span of the spin spaces on $\Sigma$ and suppose there is $v\in
(S_\Sigma)^\perp$. Then for every $x\in\Sigma$ and every $ u\in\scH_m^- $, we have
$0=(v| F^\varepsilon(x) u)=-\Sl \gR_{\varepsilon} v(x) | \, \gR_{\varepsilon}u(x)\Sr$.
Since the regularized Dirac see vacuum is regular (see the end of Section \ref{subsectionCFSM}), for every~$x\in\Sigma$ there exist~$u_\mu\in\scH_m^-$ with~$\gR_\varepsilon u_\mu (x)=\mathfrak{e}_\mu$ for $\mu\in \{1,2,3,4\}$. Plugging these vectors in the identity above we see that $\gR_\varepsilon v(x)=0$ for all $x\in\Sigma$. Now, since $v\in\scH_m^-$ we have $\gR_{\varepsilon} v\in\scH_m^-\cap C^\infty(\R^{1,3},\C^4)$ (see Proposition \ref{propositionregularization}~(i)) and therefore $\gR_{\varepsilon} v$ is a smooth solution of the Dirac equation. Since it vanishes on a Cauchy surface or on a open set, it must vanish everywhere in spacetime and thus $\gR_{\varepsilon} v=0$. Finally, since $\ker\gR_{\varepsilon}=\{0\}$
(see Proposition~\ref{propositionregularization}~(iii)), it follows that~$v=0$, concluding the proof.
\end{proof}

In the language of operator algebras, this result can be stated in terms of
local correlation operators as follows.
\begin{Thm}\label{maintheorem}
Let $\Sigma\subset\R^{1,3}$ be either an open set or the Cauchy surface~$~\{t=\text{const}\}$.
Then the set ${\mycal X}_\Sigma^\varepsilon $ is irreducible, i.e. $({\mycal X}_\Sigma^\varepsilon)'=\C\, \bI$.
\end{Thm} \noindent
The proof is given in Appendix~\ref{secappreg}.

As a direct corollary, we conclude that the von Neumann algebra generated by an open set or a Cauchy surface $\Sigma\subset\R^{1,3}$ is maximal. 
Before stating the result, we introduce the following notation.  For a given subset $\mathcal{A}\subset\Lin(\scH)$ we define
$$
\mathcal{A}(\scH):=\{Au\:|\: A\in\mathcal{A},\ u\in\scH \}=\bigcup_{A\in \mathcal{A}}\ran A(\scH)
$$

\begin{Corollary}\label{corollaryvonNeum}
Let $\Sigma\subset\R^{1,3}$ be either an open set or the Cauchy surface~$~\{t=\text{const}\}$.
Then
\begin{itemize}[leftmargin=2.5em]
	\item[{\rm{(i)}}] $\scH_m^-=\overline{ \text{\rm{span}} \; {\mycal X}_\Sigma^\varepsilon (\scH_m^-)  }$\\[-0.2cm]
	\item[{\rm{(ii)}}] $({\mycal X}_\Sigma^\varepsilon)''=\overline{  \la \mycal{X}_\Sigma^\varepsilon \ra}^s=\overline{ \la \mycal{X}_\Sigma^\varepsilon \ra}^w=\Lin(\scH_m^-)$\:.
\end{itemize}
\end{Corollary}

\begin{Remark}
We remind the reader that the identity operator is not included in the definition of  generated $^*$-algebra (see also Remark \ref{remarknonunital}).
Otherwise, point (ii) of the above result would be an immediate consequence of von Neumann's double commutant theorem. Nevertheless, point (i)  shows that the identity lies in the strong closure of the $^*$-algebra generated by the local operators.
\end{Remark}
\begin{proof}[Proof of Corollary \ref{corollaryvonNeum}]
	By construction, $\langle  {\mycal X}_\Omega^\varepsilon \rangle$ is a $^*$-algebra. Applying Lemma \ref{lemmasuffiienza}, we obtain
\[ \scH_m^-=\overline{ \text{\rm{span}} \bigcup_{x\in\Sigma}\ran F^\varepsilon(x) }=\overline{ \text{\rm{span}} \; {\mycal X}_\Omega^\varepsilon (\scH_m^-)  } \subset \overline{ \text{\rm{span}}\,\langle  {\mycal X}_\Omega^\varepsilon\rangle  (\scH_m^-) } \:. \]
At this point, we can apply Corollary 1 on page 45 in~\cite{dixmier}
to infer that $({\mycal X}_\Omega^\varepsilon)''=\langle  {\mycal X}_\Omega^\varepsilon \rangle''=\overline{\langle {\mycal X}_\Omega^\varepsilon \rangle}^s=\overline{\langle {\mycal X}_\Omega^\varepsilon \rangle}^w$. This concludes the proof.
\end{proof}

In the case of a Cauchy surface, the above results
can be understood as a version of the\\[0.2cm]
{\bf{Time Slice Axiom}}: \textit{The von Neumann algebra generated by a Cauchy surface $\Sigma$ 
(or an open neighborhood~$U$ thereof) coincides with the algebra generated by the whole spacetime:}
\begin{equation*}
	\big(({\mycal X}_U^\varepsilon)''=\big)\,({\mycal X}_\Sigma^\varepsilon)''=({\mycal X}_{\R^{1,3}}^\varepsilon)''=\Lin(\scH_m^-).
\end{equation*}

\noindent
This version of the time slice axiom remains true if the Hilbert space is extended to also contain all positive-energy solutions. On the other hand, in this case the algebra associated with a generic open set no longer needs to be irreducible.

\subsection{The Commutation Relations} \label{seccommmink}
In this section we go back to the analysis of the commutation relations started in Section \ref{commutatinrelationsection} and study it in the concrete example of the regularized Dirac sea vacuum.
More precisely, given two points $x,y\in\R^{1,3}$, we want to evaluate the expectation values of the commutator of the corresponding operators $F^\varepsilon(x)$ and~$F^\varepsilon(y)$ making use of 
Proposition \ref{prpcomm}.

As a first step, we make the form the physical wave functions more explicit. As one might expect, in this concrete example the abstract physical wave functions $\psi^u$ should coincide with the vectors $u$, modulo the action of the regularization operator. In order to compare them, it is useful to consider again the identification mapping $\Phi_z$ introduced in \eqref{id2}. A proof of the following result can be found in \cite[Proposition 1.2.6]{cfs}.
\begin{Prp}
	Let $u\in\scH_m^-$, then for every $x\in\R^{1,3}$,
	$$
	\psi^u(F^\varepsilon(x))=(\Phi_x)^{-1}\big(\gR_{\varepsilon} u(x)\big).
	$$
\end{Prp}

At this point, we can apply Proposition \ref{prpcomm} 
to any $u\in\scH_m^+$ and get:
\begin{equation*}\label{expectationvalue}
\begin{split}
&\langle u\big|\big[F^\varepsilon(x),F^\varepsilon(y)\big]u\rangle = -2 i\, \im \Sl \psi^u(\x) \,|\, \mathrm{P}(\x,\y)\, \psi^u(\y) \Sr_{\x}=\\
&=-2i\,\im \Sl (\Phi_x)^{-1}\big(\gR_{\varepsilon} u(x)\big) \,|\, \mathrm{P}(F^\varepsilon(x),F^\varepsilon(y))\, (\Phi_y)^{-1}\big(\gR_{\varepsilon} u(y)\big) \Sr_{\x}=\\
&=-2(2\pi)i\,\im \Sl \gR_{\varepsilon} u(x) \,|\, P^{2\varepsilon}(x,y)\, \gR_{\varepsilon} u(y) \Sr
\end{split}
\end{equation*}
where we used~\eqref{id3} and the fact that the function $\Phi_z$ is a unitary operator with respect to the corresponding indefinite inner products. 
Our results are summarized as follows.
\begin{Prp} \label{prpcommmink}
	For any $u\in\scH_m^-$ and for any $x,y\in\R^{1,3}$,
	$$
	\langle u\big|\big[F^\varepsilon(x),F^\varepsilon(y)\big]u\rangle=-2(2\pi)i\,\im \Sl \gR_{\varepsilon} u(x) \,|\, P^{2\varepsilon}(x,y)\, \gR_{\varepsilon} u(y) \Sr.
	$$
\end{Prp} \noindent
From~\eqref{expressionP} and~\eqref{Tm2explicit} one sees that~$P^{2\varepsilon}(x,y)$ does not vanish
in spacelike direction. As a consequence, the above commutator is in general non-zero,
even if~$x$ and~$y$ are spacelike separated and have distance much larger than~$\varepsilon$.
At least, using the asymptotics of the Bessel functions in~\eqref{Tm2explicit}, one sees that
the commutator decays exponentially for large distances like
\[ P(x,y) \sim \exp(-m \sqrt{|(y-x)^2|}) \:. \]
Clearly, this exponential tail is again a manifestation of Hegerfeldt's theorem.

We finally remark that, using the method in the proof of Proposition~\ref{prpcomm} iteratively,
one could also compute the commutator between arbitrary elements of 
the local algebras.

\subsection{Where is the Light Cone?} \label{seclc1}
In preparation of the analysis of the algebras, we prove that and explain why the underlying light-cone structure is already encoded in the local correlation operators.
Due to translation invariance (see Theorem~\ref{listpropertiesLCF} (vi)), there is no loss of generality in restricting our analysis to the origin of $\R^{1,3}$. The question is: is it possible to reconstruct the null cone centered at the origin by looking at the matrix elements of the selfadjoint operator~$F^\varepsilon(x)$
on the spin space~$S_0$.
The vectors of $S_0$, i.e. the states 
\beq \label{localize}
u^\varepsilon_{0,\chi}:=P^{\varepsilon}(\,\cdot\,,0)\chi\ \mbox{ for }\chi\in\C^4
\eeq
are the wave functions which are as far as possible localized at the origin and as such they will propagate
almost with the speed of light (see the discussion after Remark \ref{remarkbehavior}).
As a consequence, their contribution to the local correlation operators $F^\varepsilon(x)$
corresponding to spacetime points $x$ which lie away from the null cone~$L_0$ is expected to be negligible compared to
their contribution on~$L_0$.
Of course, we do not expect these contributions to be identically zero for points $x$ which are spatially separated from the origin, due to fact that a regularization parameter is introduced and that we restrict our attention to negative-energy solutions, which have infinite tails as a consequence of Hegerfeldt's theorem
(see Proposition \ref{prphegerfeldt} and Proposition \ref{prpcommmink}). 
Finally, we point out that the system under consideration has strictly positive mass $m>0$. Therefore, the Huygens
principle applies only approximately, giving rise to small but non-vanishing contributions away from the null cone.
This qualitative picture is made more precise by considering the following cases:
\vspace{0.3cm}

\begin{itemize}[leftmargin=2.5em]
\item[{\rm{(a)}}] \textit{The spacetime point $x$ does not belong to null cone~$L_0$.} \\[0.2em]
	\noindent As mentioned in the discussion before Lemma~\ref{lemma26}, the function $P(\,\cdot\,,0)$ is smooth on the complement of $L_0$. Moreover, $P^{2\varepsilon}(\,\cdot\,,0)$ converges uniformly to $P(\,\cdot\,,0)$ on any compact set $B\subset \R^{1,3}\setminus L_0$ in the limit $\varepsilon \searrow 0$ (see Lemma \ref{lemma26}).  Therefore, there exists a constant $c_B$ which does not depend on $\varepsilon$ such that, for any $x\in B$,
	$$
	|\langle u^\varepsilon_{0,\chi}| F^\varepsilon(x) \, u^\varepsilon_{0,\zeta}\rangle |\le |P^{2\varepsilon}(x,0)\chi||P^{2\varepsilon}(x,0)\zeta|\le c_B \:|\chi|\,|\zeta| \:,
	$$
	where we used the definition of $F^\varepsilon(x)$ and Proposition \ref{prpexponentsP}.
	If the regularization is removed, the right side remains finite, although each individual factor on the left
	diverges,
	$$
	\|u_{0,\chi}^\varepsilon\|,\|u_{0,\zeta}^\varepsilon\|\to \infty\quad\mbox{and}\quad\|F^\varepsilon(x)\|\to \infty\quad\mbox{as }\varepsilon\searrow 0
	$$
	This can be inferred from Theorem \ref{listpropertiesLCF}, Proposition \ref{formulaprodottoscalarelocalizzati} and Remark \ref{remarkbehavior}.
	\vspace{0.4cm}
	
	\item[{\rm{(b)}}] \textit{The spacetime $x$ point does belong to the null cone~$L_0$ and $x_3\neq 0$.} \\[0.2em]
In this case, the matrix elements can be written in terms of the closed chain as
\beq \label{melement} \begin{split}
\langle u^\varepsilon_{0,\chi} | F^\varepsilon(x) u^\varepsilon_{0,\zeta} \rangle 
&= -\Sl \chi \,|\, P^{2\varepsilon}(0,x)\, P^{2\varepsilon}(x,0) \, \zeta \Sr = \\
&=-\Sl \chi \,|\,A_{0x}^{2\varepsilon}\, \zeta \Sr \:,
\end{split}
\eeq
where we exploited the definition of the closed chain in \eqref{id4}.
	By choosing $a=\mathfrak{e}_1$ and $b=\mathfrak{e}_3$ (where~$\mathfrak{e}_\mu$ always denotes the canonical basis of~$\C^4$) we can estimate the $\varepsilon$-behavior of this quantity by means of the identities \eqref{expressionstoevaluate}. These two expressions have the same $\varepsilon$-behavior for points on the null cone. Moreover they are purely real and purely imaginary, respectively, so there is no risk of cancellation effects. Therefore, as $x_3\neq 0$, there exists a constant $c(x)>0$ and for $\varepsilon$ small enough,
	\begin{equation*}
	\begin{split}
	|\Sl \mathfrak{e}_1| A_{0x}^{2\varepsilon}\,\mathfrak{e}_4\Sr|=|\mathfrak{e}_1^\dagger\, A_{0x}^{2\varepsilon}\,\mathfrak{e}_4|&\ge \frac{c(x)\, \varepsilon}{|(t-i\varepsilon)^2-\V{x}^2|^4}=\frac{c(x) \varepsilon}{|t^2-\V{x}^2-\varepsilon^2-2i\varepsilon t|^4}=\\
	&
	=\frac{c(x)\, \varepsilon}{|\varepsilon^2+2i\varepsilon t|^4}\ge c'(x)\frac{1}{\varepsilon^3}
	\end{split}
	\end{equation*}
where in the last identity we used the fact that $x\in L_0$ and where $c'(x)$ is some suitable strictly positive constant.
We conclude that, for a suitable choice of the spinors~$\chi$ and~$\zeta$, 
the matrix element~\eqref{melement} diverges as~$\varepsilon \searrow 0$. 
\end{itemize}
\vspace{0.2cm}

The above computations show that, by considering the matrix elements
of the operator~$F^\varepsilon(x)$ on vectors of~$S_0$ (which are of the form~\eqref{localize})
and analyzing the limit~$\varepsilon \searrow 0$, we can detect whether the point~$x$ lies
on the null cone~$L_0$ or not.

\subsection{The Regularized Algebra}
A similar but more practical way to consider the problem is to smear out the local correlation operators on the set $\Omega$ by means of continuous compactly supported functions,
\begin{equation}\label{defsmeared}
A^\varepsilon_f:=\int_{\R^{1,3}}f(x)\,  F^\varepsilon(x)\, d^4x\in\Lin(\scH_m),\quad f\in C_0^0(\Omega,\C),
\end{equation}
as was carried out abstractly in Section~\ref{secalgebras}.
In this way, it is possible to avoid some divergences and to extend the analysis to the unregularized case. This will be discussed in what follows. In the case of functions that are positive valued and with unit $\scL^1$-norm the integral~\eqref{defsmeared} can be interpreted as an average of the local correlation operators on $\supp f$.

In this concrete case we work with the space $$C^\flat_0(\Omega,\C):=C_0^\infty(\Omega,\C).$$
The operator $A_f^\varepsilon$ can be uniquely determined in terms of its matrix elements as follows.
\begin{Prp}
	For any $f\in C_0^\infty(\Omega,\C)$, the operator $A_f^\varepsilon$ is the unique operator~$\Lin(\scH_m^-)$ such that for all $u,v\in\scH_m^-$,
\[ 
	\langle u|A^\varepsilon_f v\ra =\int_{\R^4} f(x) \la u| F^\varepsilon(x) v\ra \, d^4x=-\int_{\R^4} f(x)\,
	\Sl \gR_{\varepsilon} u(x) \,|\, \gR_{\varepsilon} u (x)\Sr \, d^4 x\:. \]
\end{Prp}
\begin{proof}
	The above property is fulfilled by $A_f^\varepsilon$, thanks to the properties of the Bochner integral. Uniqueness follows immediately from the arbitrariness of $u$ and~$v$.
\end{proof}
It is useful to collect a few properties of these operators. Take any $f\in C_0^\infty(\R^{1,3},\C)$.  Since the operators $ F^\varepsilon(x)$ have finite rank and all have the same trace (as they are unitarily equivalent, see Theorem \ref{listpropertiesLCF}~(vi)), one might expect that also the integral $A_f^\varepsilon$, which involves a smooth, compactly supported function, is trace-class. Also, the trace should coincide with the trace of the local correlation operators, multiplied by the integral of the function. This is indeed true, as shown in the following proposition.
\begin{Prp}\label{proprietaAf}
	For all~$f,g\in C_0^\infty(\Omega,\C)$ and $a,b\in\C$, the following statements~hold:\\[-0.8em]
	\begin{itemize}
		\item[{\rm{(i)}}] $A^\varepsilon_{af+bg}=aA_f^\varepsilon+bA_g^\varepsilon$ and $(A^\varepsilon_{f})^*=A^\varepsilon_{\overline{f}}$,\\[-0.5em]
		\item[{\rm{(ii)}}] If $\supp f\subset  \{(t, \V{x})\:|\: |t| \leq T \}$, then $\|A_f^\varepsilon\|\le 4T\|f\|_\infty$,\\[-0.5em]
		\item[{\rm{(iii)}}] The operator $A_f^\varepsilon$ is trace-class and 
\[ \mbox{\rm{tr}}\,(A^\varepsilon_f)=\mbox{\rm{tr}}_{\rm{vac}}^\varepsilon\: \int_{\R^4}f(x)\, d^4x \]
(with~$\mbox{\rm{tr}}_{\rm{vac}}^\varepsilon$ as defined in~\eqref{trvacdef}).
	\end{itemize}
\end{Prp} \noindent
The proof is given in Appendix~\ref{secappendix}.

Before proceeding with further properties of these operators, it is interesting to see how they act on the
dense set of smooth solutions (see Proposition~\ref{propositioncontinuityP1}),
$$
P(\,\cdot\,,\varphi):\R^{1,3}\ni x\mapsto \int_{\R^4}\frac{d^4k}{(2\pi)^2}\,\hat{P}(k)\,\hat{\varphi}(k)\, e^{-ik\cdot x}\in \C^4,\quad \varphi\in {\mathcal{S}}(\R^{1,3},\C^4).
$$

\begin{Thm}\label{theoremeactionAf}
For any~$f\in C_0^\infty(\R^{1,3},\C)$ and~$\varphi\in {\mathcal{S}}(\R^{1,3},\C^4)$,
\[ A_f^\varepsilon \big(P(\,\cdot\,,\varphi)\big) 
=2\pi\ P^\varepsilon \big(\,\cdot\,,f\,P^\varepsilon(\,\cdot\,,\varphi) \big) \:. \]
\end{Thm} \noindent
The proof is given in Appendix~\ref{secappendix}.

At this point, similar as in Section~\ref{secalgebras},
we collect all the operators $A_f^\varepsilon$ supported within a given set $\Omega$,
$$
L^\varepsilon_\Omega:=\left\{A^\varepsilon_f,\ f\in C_0^\infty(\Omega,\C)  \right\} , $$
where the superscript~$\varepsilon$ clarifies the dependence on the regularization.
We remark that the set~$L^\varepsilon_\Omega$ is a complex linear space and a
$^*$-representation of $C_0^\infty(\Omega,\C)$ thanks to Proposition~\ref{proprietaAf}~(i). 

\begin{Def}
For any open subset~$\Omega\subset M$, the $^*$-algebra
\[ \mycal{A}_\Omega^\varepsilon:= \langle L_\Omega^\varepsilon\rangle \]
is referred to as the {\bf{regularized local algebra}} corresponding to~$C^\infty_0(\Omega, \C)$.
\end{Def}

\begin{Remark}\label{remarkregidentity}
	As already pointed out in Remark \eqref{remarkidentity}, these $^*$-algebras, as well as their closures in the $\sup$-norm topology, do not contain the identity operator. 
\end{Remark}
We now collect a few basic properties of these algebras.
\begin{Prp}\label{propositionrelationsreg}
	The following statements hold:
	\vspace{0.1cm}
	\begin{itemize}[leftmargin=2.5em]
		\item[{\rm{(i)}}] Let $\Omega_1\subset\Omega_2\subset\R^{1,3}$ be open sets, then $\A^\varepsilon_{\Omega_1}\subset\A^\varepsilon_{\Omega_2}$.\\[-0.5em]
		\item[{\rm{(ii)}}] Let $\Omega\subset\R^{1,3}$ be an open set and $a\in\R^{1,3}$, then $\A^\varepsilon_{\Omega+a}=(\mathrm{U}_a)^\dagger\, \A^\varepsilon_\Omega\, \mathrm{U}_a$.\\[-0.5em]
		\item[{\rm{(iii)}}] Let $\Omega\subset\R^{1,3}$ be an open set, then
\[
		\A^\varepsilon_\Omega=\left\langle\bigcup_{\Omega'\in\mycal{O}(\Omega)}\A^\varepsilon_{\Omega'}\right\rangle \:, \]
		where $\mycal{O}(\Omega)$ denotes the family of open bounded subsets of $\Omega$.
	\end{itemize}
\end{Prp}
\begin{proof}
		Points (i) and~(iii) are obvious. Let us prove point (ii). For any~$u,v\in\scH_m^-$ and $f\in C_0^0(\R^{1,3},\C)$, 
		we have (note that $\gR_{\varepsilon}$ and $\mathrm{U}_a$ commute)
\begin{align*}
		\la u|(\mathrm{U}_a)^\dagger A^\varepsilon_f \mathrm{U}_a v\ra &= \la\mathrm{U}_a u| A^\varepsilon_f \mathrm{U}_a v\ra = -\int_{\R^4} f(x)\,
		\Sl \gR_{\varepsilon}u(x+a) | \gR_{\varepsilon}v(x+a) \Sr\, d^4x=\\
		&= -\int_{\R^4} f(x-a)\, \Sl \gR_{\varepsilon}u (x) | \gR_{\varepsilon}v(x)\Sr \, d^4x = \la u| A^\varepsilon_{f_{-a}} v\ra \:,
\end{align*}
		where $f_{-a}(x):=f(x-a)$. The function $f_{-a}$ is supported within $\Omega +a$. 
		The claim follows by taking products and linear combinations and using 
		that~$\mathrm{U}_a$ is unitary.
\end{proof}

We now study the von Neumann algebra generated by the operators in $L_\Omega^\varepsilon$. Note that, in view of Remark \ref{remarkregidentity}, point (ii) of the next statement is not obvious.

\begin{Prp}\label{identification}
For any~$\Omega\subset\R^{1,3}$ the following statements hold:
\vspace{0.1em}
\begin{itemize}[leftmargin=2.5em]
	\vspace{0.05cm}
\item[{\rm{(i)}}]  $\scH_m^-=\overline{ \text{\rm{span}} \,L^\varepsilon_\Omega (\scH_m^-)}$,\\[-0.5em]
\item[{\rm{(ii)}}] $(L_\Omega^\varepsilon)''=\overline{\mycal{A}_\Omega^\varepsilon}^s=\overline{ \mycal{A}_\Omega^\varepsilon}^w=\Lin(\scH_m^-)$, \\[-0.5em]
\item[{\rm{(iii)}}] $\overline{\la \mycal{X}^\varepsilon_\Omega \ra}=\overline{\mycal{A}_\Omega^\varepsilon}.$\\[-0.9em]
\end{itemize}
In particular, both $L_\Omega^\varepsilon$ and the $^*$-algebra $\mycal{A}_\Omega^\varepsilon$ are irreducible.
\end{Prp}
\begin{proof}
Point (iii) follows directly from Proposition~\ref{closurecoincide}. In particular, Theorem \ref{maintheorem} implies that $L_\Omega^\varepsilon$ (and also its generated algebra) is irreducible.
	Point (ii) follows directly from the irreducibility of $L_\Omega^\varepsilon$, once we have proved
	that the identity lies in the strong closure of $\A_\Omega^\varepsilon$. More precisely, we must show that	$\overline{\text{span}\, \A_\Omega^\varepsilon (\scH_m^-)}=\scH_m^-$ and apply again Corollary~1 on page~45 of~\cite{dixmier}. To this end, assume that there is a vector $v$ which is orthogonal to $\overline{ \text{span} \,L^\varepsilon_\Omega (\scH_m^-)}$. Then, by definition of the smeared operators,
	$$
	0=-\int_{\R^4} f(x)\: \Sl \gR_{\varepsilon} u(x) \:|\: \gR_{\varepsilon} v(x) \Sr\: d^4x\quad\mbox{for all } f\in C_0^\infty(\Omega,\C) \:.
	$$
The theorem by Du Bois-Reymond ensures that $\Sl \gR_{\varepsilon} u(x) | \gR_{\varepsilon} v(x) \Sr=0$ for every $u\in\scH_m^-$ and $x\in\Omega$. At this point, we can proceed just as in the proof of Lemma \ref{lemmasuffiienza}
to infer that~$v=0$. This concludes the proof of point (i).
\end{proof}

\subsection{Once Again: Where is the Light Cone?}\label{seclc2}
Working with the regularized local algebras, we can ask the same question as in Section~\ref{seclc1}:
Where is the light cone? How can we detect it by looking at the elements of ${\mycal A}^\varepsilon_\Omega$?
In order to analyze this question, we fix an open set $\Omega\subset\R^{1,3}$.
Again, the idea is to evaluate the matrix elements of the operators~$A \in {\mycal A}^\varepsilon_\Omega$ 
on the spin space~$S_0$ and to study their dependence on $\varepsilon$. The elements of $S_0$ are given 
again by the wave functions 
\beq\label{localize2}
u^\varepsilon_{0,\chi}:=P^{\varepsilon}(\,\cdot\,,0)\chi\ \mbox{ for }\chi\in\C^4.
\eeq
We now state our result; the proof (which also uses results on the unregularized algebra
to be proved in Section~\ref{secunregularized} below) is given in Appendix~\ref{secappendix}.
By translation invariance, we restrict attention to the light cone centered at the origin.
In order to relate the algebras for different values of~$\varepsilon$ to each other,
we introduce the free algebra~$\A(C^\infty_0(\Omega))$ generated by $C^\infty_0(\Omega)$
and define the linear mapping
\beq \label{iotaepsdef}
\iota^\varepsilon : \A(C^\infty_0(\Omega)) \rightarrow\mycal{A}_\Omega^\varepsilon 
\quad \text{defined on monomials by} \quad
f_{i_1}\cdots f_{i_k} \mapsto A_{f_{i_1}}^\varepsilon \cdots A_{f_{i_k}}^\varepsilon \:.
\eeq

\begin{Prp} \label{prpregularized} The following statements hold:
\begin{itemize}[leftmargin=2.5em]
	\vspace{0.05cm}
\item[{\rm{(i)}}] If the set $\Omega$ does not intersect the null cone~$L_0$,
	for any~$a \in \A(C^\infty_0(\Omega))$ there is a constant~$c(a)>0$ such that
	for all~$\chi,\zeta \in \C^4$,
\[ |\langle u^\varepsilon_{0,\chi}\, |\,\iota^\varepsilon(a)\, u^\varepsilon_{0,\zeta}\rangle| \le c(a)\: |\chi|\:
|\zeta| \:, \]
uniformly in~$\varepsilon>0$.\\[-0.5em]
\item[{\rm{(ii)}}] If the set $\Omega$ intersects the null cone~$L_0$, there
are an element~$a \in \A(C^\infty_0(\Omega))$, a constant~$c>0$ and spinors~$\chi,\zeta \in \C^4$ such that
for sufficiently small~$\varepsilon>0$,
\[	|\langle u^\varepsilon_{0,\chi}\, |\,\iota^\varepsilon(a)\, u^\varepsilon_{0,\chi}\rangle|\geq \frac{1}{c\: \varepsilon^2}
\:. \]
\end{itemize}
\end{Prp} \noindent
This proposition shows that, by considering the matrix elements
of operators in the algebra on vectors of~$S_0$ (which are of the form~\eqref{localize2})
and analyzing the limit~$\varepsilon \searrow 0$, we can detect whether $\Omega$
intersects the null cone~$L_0$ or not.

In Section~\ref{secunregularized} we shall introduce and analyze local algebras without regularization.
We will see that the above proposition holds analogously for these unregularized local algebras.

\subsection{The Commutator of Spacetime Points and Local Algebras}\label{commutators}
The goal of this section is to show that every spacetime point operator commutes with
the algebras supported away from the corresponding null cone, up to small corrections involving the regularization length. 
More generally, given a point $x\in\R^{1,3}$ we estimate the $\varepsilon$-behavior of the commutators
$$
[F^\varepsilon(x),A]\quad\ \mbox{for }\  A\in\A_\Omega^\varepsilon\ \mbox{ and }\ \Omega\subset\R^{1,3}.
$$ 
As we shall see, the scaling of this commutator depends on whether $\Omega$ does or does not intersect the null cone
centered at~$x$, making it possible to recover the light-cone structure from the local algebras.
Again, by translation invariance, it suffices to focus on the origin of Minkowski space.

We now state our first estimate. Since the proof uses results of Section~\ref{secunregularized}, it is given in
Appendix~\ref{secappendix}.
\begin{Thm} \label{prpiotaes}
	Let $\Omega$ be an open set which does not intersect the null cone~$L_0$. Then
	for any $a\in \A(C^\infty_0(\Omega))$ there is a constant $c(a)>0$ such that for all
	sufficiently small~$\varepsilon>0$,
	$$
	\|F^\varepsilon(0)\,\iota^\varepsilon(a)\|\le c(a)\, \varepsilon^{\frac{3}{2}}\, \|F^\varepsilon(0)\| \:.
	$$
	In particular, the commutator of the spacetime point operator and the algebra element fulfills the inequality
	$$
	\|\,[F^\varepsilon(0),\iota^\varepsilon(a)]\,\|\le  2 c(a)\, \varepsilon^{\frac{3}{2}}\,\|F^\varepsilon(0)\| \:.
	$$
\end{Thm} \noindent
The factor~$\|F^\varepsilon(0)\|$ is needed in order for the estimate to be
invariant under scalings of~$F^\varepsilon$.

Instead, if $\iota^\varepsilon(a)$ instead belong to the algebra localized on an open set which intersect the null cone, we get different results:

\begin{Prp}\label{prpiotaes2}
Let $\Omega$ be an open set which intersects the null cone~$L_0$. Then
then there exists $a\in \A(C^\infty_0(\Omega))$ and a constant $c>0$ such that
for all sufficiently small $\varepsilon$,
$$
	\big\|\,[F^\varepsilon(0),\iota^\varepsilon(a)]\,\big\|\ge c\:\varepsilon\,\|F^\varepsilon(0)\| \:.
$$
\end{Prp}
\begin{proof}
Let $f$ a smooth function with compact support in $\Omega$. Then
	\begin{align*}
	\la u^\varepsilon_{0,1}\,|\,[A_f^\varepsilon,F^\varepsilon(0)]\, u^\varepsilon_{0,3})\ra&=\la u^\varepsilon_{0,1}\,|\,A_f^\varepsilon\,F^\varepsilon(0)\,u^\varepsilon_{0,3}\ra-\la u^\varepsilon_{0,1}\,|\,F^\varepsilon(0)\,A_f^\varepsilon\, u^\varepsilon_{0,3}\ra=\\
	&=\la u^\varepsilon_{0,1}\,|\,A_f^\varepsilon\,F^\varepsilon(0)\,u^\varepsilon_{0,3}\ra-\la F^\varepsilon(0)\,u^\varepsilon_{0,1}\,|\,A_f^\varepsilon\, u^\varepsilon_{0,3}\ra=\\
	&=2\pi\big(\nu^+(\varepsilon)-\nu^-(\varepsilon)\big)\la u^\varepsilon_{0,1}\,|\,A_f^\varepsilon \,u^\varepsilon_{0,3}\ra =\\
	&=4\pi\,g(\varepsilon)\,\varepsilon^{-3}\int_{\R^4} f(x)\, \mathfrak{e}_1^\dagger\, A_{0x}^{2\varepsilon}\, \mathfrak{e}_3\, d^4 x \:,
	\end{align*}
where we used Remark \ref{remarkbehavior} and Theorem \ref{listpropertiesLCF}. At this point, using the fact that the operator norm can be expressed as $\|A\|=\sup_{u,v}\frac{|(u|Av)|}{\|u\|\|v\|}$ 
together with the inequalities in~\eqref{estimatenormlocstate}, we get
\begin{equation}\label{equazionestima}
\begin{split}
\big\|[A_f^\varepsilon,F^\varepsilon(0)]\big\|\ge (2\pi)^2\, \frac{|2\,g(\varepsilon)\,\varepsilon^{-3}|}{|\nu^+(\varepsilon)|}\int_{\R^4} f(x)\, \mathfrak{e}_1^\dagger\, A_{0x}^{2\varepsilon}\, \mathfrak{e}_3\, d^4 x \:.
\end{split}
\end{equation}
Noting that $2\,g(\varepsilon)\,\varepsilon^{-3}\sim\nu^+(\varepsilon)\sim\varepsilon^{-3}$, 
we can choose a point $x\in \Omega\cap L_0\setminus\{0\}$ and a function~$f$ as in
Theorem \ref{theoremasymptotics}. With this choice, the integral on the right-hand side of \eqref{equazionestima}
scales $\sim \varepsilon^{-2}$. The proof follows by noticing that $\|F^\varepsilon(x)\|\sim \varepsilon^{-3}$.
\end{proof}

The $\varepsilon$-behaviors of the two estimates of Theorem~\ref{prpiotaes} and Proposition~\ref{prpiotaes2}, respectively, are incompatible in the limit $\varepsilon\searrow 0$. On the null cone centered at $x$, the vectors $P^{\varepsilon}(\,\cdot\,,x) \,\chi$ develop singularities, and this becomes manifest on the algebra level in the commutators.

These results can also be interpreted in terms of the loss of information carried away on waves propagating
``almost with the speed of light.'' The reader interested in this interpretation
and in the connection to the ETH formulation of quantum theory is referred to~\cite{eth-cfs}.

\section{Unregularized Local Algebras in Minkowski Space} \label{secunregularized}
In this section we turn our attention to the following question: which of the above structures
remain well-defined if the regularization is removed? 
Clearly, the spacetime operators~$F^\varepsilon(x)$ diverge in the limit~$\varepsilon \searrow 0$, because
(see Remark \ref{remarkbehavior})
$$
\| F^\varepsilon(x)\|=2\pi\nu^+(\varepsilon)=\frac{2\pi m}{\varepsilon^2}\,f(\varepsilon)+\frac{2\pi}{\varepsilon^3}\,g(\varepsilon)\stackrel{\varepsilon\searrow 0}{\longrightarrow} +\infty \:.
$$
However, we will show that for the smeared operators in~\eqref{defsmeared} this limit
does exist, making it possible to introduce unregularized local algebras.

Our method is based on the following simple consideration:
We cannot work directly with~\eqref{defsmeared} because the integrand diverges as~$\varepsilon \searrow 0$.
The integrand is ill-defined pointwise even if we consider expectation values, because in
$$
\la u| F^\varepsilon(x) v\ra =-\Sl  \gR_{\varepsilon} u(x) | \gR_\varepsilon v(x) \Sr
$$
the limit does not exist if~$u$ or~$v$ are not continuous.
However, if we integrate
against a smooth, compactly supported function, the integral does converge as~$\varepsilon \searrow 0$.
We now enter the detailed constructions.

\subsection{The Unregularized Smeared Operators}
The unregularized smeared operator $A^\circ_f$ can be defined using the Riesz representation theorem.
\begin{Prp}
	For any $f\in C_0^\infty(\R^{1,3},\C)$ there exists a unique operator $A^\circ_f\in \Lin(\scH_m^-)$ such that for every $u,v\in\scH_m^-$,
	$$
	\langle u|A^\circ_f v\rangle= -\int_{\R^4} f(x)\,\Sl u(x) | v(x)\Sr\, d^4x \:.
	$$
\end{Prp}
\begin{proof}
	Consider two vectors $u,v\in\scH_m^-$ and a function $f\in C_0^\infty(\R^{1,3},\C)$. By choosing $T>0$ suitably, we can assume that $f$ is supported inside the time strip
\[ R_T := \{(t, \V{x})\:|\: |t| \leq T \} \:. \]	
Using Hölder's inequality, we get
	\begin{align*}
	\int_{\R^4} &\big|f(x)\, \Sl u(x) | v(x)\Sr \big|\, d^4x=\int_{R_T} \big|f(x)\, \Sl u(x) | v(x)\Sr \big|\, d^4x \le\\
	&\le \left(\int_{R_T}|f(x)|^2|u(x)|^2\, d^4x\right)^{1/2}\left(\int_{R_T}|v(x)|^2\, d^4x\right)^{1/2}\le\\
	&\le \|f\|_\infty \left(\int_{R_T}|u(x)|^2\, d^4x\right)^{1/2}\left(\int_{R_T}|v(x)|^2\, d^4x\right)^{1/2} = \\
	&= 2T\, \|f\|_\infty\, \|u\|\|v\| <\infty \:,
	\end{align*}
where in the last step we used the theorem by Fubini-Tonelli and the fact that the spatial integral is time independent due to current conservation (see \cite[Lemma 2.7]{oppio}).
As a consequence, the sesquilinear form
	$$
	d_f:\scH_m^- \times\scH_m^- \rightarrow \C\:, \qquad (u,v) \mapsto  -\int_{\R^4} f(x)\,\Sl u(x) | v(x)\Sr \, d^4x
	$$
is well-defined and continuous. The Riesz Representation Theorem yields a unique operator $A^\circ_f$ with~$d_f(u,v)=(u|A^\circ_fv)$.
\end{proof}

We now collect a few properties of these operators. The proof is similar as in the regularized case
(see Proposition~\ref{proprietaAf}).
\begin{Prp}\label{proprietaAfunreg}
Let $\Omega\subset\R^{1,3}$ be open, $f,g\in C_0^\infty(\R^{1,3},\C)$ and $a,b\in\C$. Then
the following statements hold:
	\begin{itemize}[leftmargin=2.5em]
		\vspace{0.12em}
		\item[{\rm{(i)}}] $A^\circ_{af+bg}=aA^\circ_f+bA^\circ_g$ and $(A^\circ_{f})^*=A^\circ_{\overline{f}}$\\[-0.5em]	\item[{\rm{(ii)}}] If $\supp f\subset R_T$, then $\|A^\circ_f\|\le 4T\:\|f\|_\infty$
	\end{itemize}
\end{Prp}

At this point, one may wonder if the operator~$A_f^\circ$ can be obtained as the limit $\varepsilon\searrow 0$ of the regularized operators $A_f^\varepsilon$. This is indeed possible, as shown in the following theorem.
\begin{Thm}\label{properreg}
For any~$f\in C_0^\infty(\R^{1,3},\C)$, the following statements hold:
	\begin{itemize}[leftmargin=2.5em]
		\vspace{0.2em}
		\item[{\rm{(i)}}] $A_f^\varepsilon = \gR_{\varepsilon}\, A^\circ_f\,\gR_{\varepsilon}$,\\[-0.6em]
		\item[{\rm{(ii)}}] $\lim_{\varepsilon\to 0}A_f^\varepsilon= A^\circ_f$ in the strong topology.
	\end{itemize}
\end{Thm}
\begin{proof} Point (i) follows immediately from the fact that, for any ~$u,v\in\scH_m^-$,
$$
\langle u\,|\,\gR_{\varepsilon}\, A^\circ_f\, \gR_{\varepsilon} u\rangle=  \langle \gR_{\varepsilon} u \,|\, A^\circ_f\, \gR_{\varepsilon} u\rangle = -\int_{\R^4} f(x)\, \Sl \gR_{\varepsilon} u(x) | \gR_{\varepsilon}  u(x)\Sr\, d^4x = \langle u|A_f^\varepsilon v\rangle \:.
$$ 
In order to prove point (ii), note that, for any $u\in\scH_m^-$,
	\begin{equation*}
	\begin{split}
	\|A_f^\varepsilon u-A^\circ_f u\|&\le \|\gR_{\varepsilon} (A^\circ_f\, \gR_{\varepsilon} u)-\gR_{\varepsilon}(A^\circ_f u)\|+\|\gR_{\varepsilon}(A^\circ_f u)- A^\circ_f u\|\le \\
	&\le  \|A^\circ_f\, \gR_{\varepsilon} u-A^\circ_f u\|+\|\gR_{\varepsilon}(A^\circ_f u)- A^\circ_f u\| \:,
	\end{split}
	\end{equation*}
	where in the last inequality we used the bound~$\|\gR_{\varepsilon}\|\le 1$ together with~(i). The fact that $A^\circ_f$ is bounded and that $\gR_{\varepsilon}$ is strongly continuous concludes the proof. 
\end{proof}

As we did in the regularized case, we now study how these operators act on the dense set of smooth
solutions~$P(\,\cdot\,,\varphi)$. Notice that the following result is in agreement with Theorem~\ref{properreg} and with Theorem~\ref{theoremeactionAf}. The proof is analogous to that of Theorem~\ref{theoremeactionAf}, with obvious changes.
\begin{Thm}\label{theoremeactionAfunreg}
	Let $f\in C_0^\infty(\R^{1,3},\C)$ and $\varphi\in {\mathcal{S}}(\R^{1,3},\C^4)$. Then
\[	A^\circ_f \big(P(\,\cdot\,,\varphi)\big)=2\pi\,P(\,\cdot\,,fP(\,\cdot\,,\varphi)) \:. \]
\end{Thm}

As in the regularized case, we again collect all the operators $A_f$ supported within a given set $\Omega$,
$$
L^\circ_\Omega:=\left\{A^\circ_f,\ f\in C_0^\infty(\Omega,\C)  \right\} . $$
Again, the set~$L^\circ_\Omega$ is a complex linear space and a
$^*$-representation of $C_0^\infty(\Omega,\C)$ thanks to Proposition \ref{proprietaAfunreg}~(i). 
We are now ready to introduce the unregularized local algebras.
\begin{Def}
For any open subset~$\Omega\subset M$, the $^*$-algebra
\[ \mycal{A}^\circ_\Omega := \langle L^\circ_\Omega\rangle \]
is referred to as the {\bf{unregularized local algebra}} corresponding to~$C^\infty_0(\Omega, \C)$.
\end{Def}

\begin{Remark}\label{remarkmaybeidentityunreg}
At present, it is unknown whether the~$^*$-algebras $\A_\Omega^\circ$, as well as their closures in the $\sup$-norm topology, contain the identity operator or not. Arguments as in Remarks \ref{remarkidentity} or \ref{remarkregidentity} cannot be applied, because the unregularized operators $A_f^\circ$ are not known to be compact.
As a partial result, obtained with different methods, it will be shown in Proposition~\ref{propositionidentity} that the identity does not belong to $\A_\Omega^\circ$ if~$\Omega$ is contained in the interior light cone $I_x$ of some spacetime point $x\in\R^{1,3}$.
\end{Remark}

The unregularized algebra can be obtained
from the regularized algebra in the limit~$\varepsilon \searrow 0$, as we now make precise.
\begin{Remark}
Applying Proposition~\ref{properreg} inductively, it readily follows that every element of the algebra $\mycal{A}_\Omega^\varepsilon$ converges strongly to a corresponding element of the unregularized algebra,
	$$
	\sum_{k=1}^N\sum_{i_1,\dots,i_k=1}^{n_k}\lambda ^{i_1,\dots,i_k}A^\varepsilon_{f_{i_1}}\cdots A_{f_{i_k}}^\varepsilon\stackrel{s}{\longrightarrow} \sum_{k=1}^N\sum_{i_1,\dots,i_k=1}^{n_k}\lambda ^{i_1,\dots,i_k}
A^\circ_{f_{i_1}}\cdots A^\circ_{f_{i_k}} \:.
	$$
In terms of the operators~$\iota^\varepsilon$ introduced in~\eqref{iotaepsdef},
this can be restated that these operators converge pointwise, i.e.
\[ \iota^\varepsilon(a) \rightarrow \iota^\circ(a)\quad\mbox{if $\varepsilon\searrow 0$} \quad \text{for all~$a \in \A(C^\infty_0(\Omega))$} \:, \]
where~$\iota^\circ$ is defined in analogy to~\eqref{iotaepsdef} simply by replacing the indices~$\varepsilon$ with $\circ$.
\end{Remark}

The following basic properties are an immediate consequence of our definitions. 
\begin{Prp}\label{propositionrelations}
The following statements hold:
	\begin{itemize}[leftmargin=2.5em]
		\vspace{0.1cm}
		\item[{\rm{(i)}}] Let $\Omega_1\subset\Omega_2\subset\R^{1,3}$ be open sets. Then $\A^\circ_{\Omega_1}\subset\A^\circ_{\Omega_2}$,\\[-0.6em]
		\item[{\rm{(ii)}}] Let $\Omega\subset\R^{1,3}$ be an open set and $a\in\R^{1,3}$. Then $\A^\circ_{\Omega+a}=(\mathrm{U}_a)^\dagger\, \A^\circ_\Omega\, \mathrm{U}_a$\\[-0.6em]
		\item[{\rm{(iii)}}] For any open set~$\Omega\subset\R^{1,3}$, then
\[ \A^\circ_\Omega=\left\langle\bigcup_{\Omega'\in\mycal{O}(\Omega)}\A^\circ_{\Omega'}\right\rangle \:, \]
where $\mycal{O}(\Omega)$ denotes the family of all open bounded subsets of~$\Omega$.
	\end{itemize}
\end{Prp}
\begin{proof}
The proof is analogous to that of Proposition \ref{propositionrelationsreg}.
\end{proof}

Let us now consider the von Neumann algebra generated by~$L^\circ_\Omega$.
In analogy to Proposition~\ref{identification}~(i), it can be written as the
strong or weak closure of the $^*$-algebra generated by~$L^\circ_\Omega$.
The proof follows the same strategy as in the regularized case.
Nevertheless, we give the proof in detail, because there are subtleties concerning
the regularity of the Dirac wave functions. 

\begin{Prp}\label{propositionspan} For any open set~$\Omega\subset\R^{1,3}$, the following identities hold:
	\begin{itemize}[leftmargin=2.5em]
		\vspace{0.2em}
		\item[{\em{(i)}}] $\scH_m^-=\overline{\text{\rm{span}}\,  L^\circ_\Omega (\scH_m^-)}$, \\[-0.6em]
		\item[{\em{(ii)}}] $(L^\circ_\Omega)'' = \overline{\langle L^\circ_\Omega\rangle}^s=\overline{\langle L^\circ_\Omega\rangle}^w$.
	\end{itemize}
\end{Prp}
\begin{proof} let $(x_n)_{n \in \N}$ be a sequence which is dense in~$\Omega$. Given~$n$,
	we choose four smooth solutions $w_\mu \in \H_m^-$ with $w_\mu(x_n)=\mathfrak{e}_\mu$ for $\mu=1,2,3,4$
	(this is possible because our causal fermion system is regular; see the end of Section \ref{subsectionCFSM}). Then in particular, with obvious notation,
	$\det(w_0(x_n)|w_1(x_n)|w_2(x_n)|w_3(x_n))\neq 0$. By continuity of the functions $w_\mu$ and the determinant function, there is a neighborhood $B_{\varepsilon_n}(x_n)\subset\Omega$ on which the functions $w_\mu$ are
	pointwise linearly independent, and therefore define a basis of $\C^4$ at  any point $z\in B_{\varepsilon_n}(x_n)$.

	In order to prove~(i), it suffices to show that the orthogonal complement of the space~$L^\circ_\Omega(\scH_m^-)$
	is trivial. Given~$v\perp L^\circ_\Omega(\scH_m^-)$, we know that $v\perp A^\circ_f w$ for any $w\in\scH_m^-$ and $f\in C_0^\infty(\Omega)$. In particular, choosing~$w$ as one of the functions~$w_\mu$ constructed above,
\[ 
	0 = (v|A^\circ_fw_\mu)=\int_\Omega f(x)\, \Sl v(x) | w_\mu(x) \Sr\, d^4 x \:. \]
	The arbitrariness of $f$ implies that the functions $\Sl v | w_\mu \Sr$ vanish almost everywhere in $\Omega$.
	Since the vectors~$w_\mu(z)$ form a basis of the spinors at every point~$z\in B_{\varepsilon_n}(x_n)$, we conclude
	that the function~$v$ vanishes almost everywhere in~$B_{\varepsilon_n}(x_n)$.
	This statement holds for any~$n \in \N$. Since the open balls~$B_{\varepsilon_n}(x_n)$ provide a countable covering of~$\Omega$,
	we conclude that~$v$ vanishes in~$\Omega$ almost everywhere.
%
At this point,	Hegerfeldt's theorem (see Proposition \ref{prphegerfeldt}) implies that this function vanishes identically almost everywhere, i.e. $v=0$. 
	
	Part~(ii) is an immediate consequence of~(i) and
	Corollary 1 on page 45 in~\cite{dixmier}, just as explained in the proof of Corollary~\ref{corollaryvonNeum}.
\end{proof}

To conclude this section, we state a last important result. As in the regularized case (see Proposition \ref{identification}), the local algebra associated with any open subset of Minkowski space
is irreducible. Again, this is a manifestation of Hegerfeldt's theorem (see Proposition \ref{prphegerfeldt}) and translates its content from the level of wave functions to operators. The proof of this result can be found in
Appendix~\ref{secappunreg}. It develops on ideas and techniques different from those used in the regularized case. This was necessary because the notions like the spacetime point operator~$F(x)$ and the spin space $S_x$ cannot introduced if no regularization is present,
making it impossible to use again the strategy behind the proof of Theorem \ref{maintheorem}.
Also, several subtleties arise due to the fact that the functions in~$\scH_m^-$ are in general not
differentiable or even continuous.

\begin{Thm}\label{irreducibilityunreg}
For any open subset~$\Omega\subset\R^{1,3}$, the set~$L^\circ_\Omega$ is irreducible. In particular, the algebra $\A^\circ_\Omega$ is irreducible. 
\end{Thm}	

\begin{Remark}
Point~{\rm{(i)}} of Proposition \ref{propositionspan} shows in particular that
for every $x\in\R^{1,3}$, $\varepsilon\in (0,\varepsilon_{max})$ and $\chi\in\C^4$,
the states $u_{x,\chi}^\varepsilon:=P^{\varepsilon}(\,\cdot\,,x)\chi$ can be approximated by
vectors in the image of the smeared operators. 
Similarly, point~{\rm{(ii)}} and the irreducibility of $L^\circ_\Omega$ imply that
the regularization operators (and therefore also the regularized operators $A_f^\varepsilon$) could
be reconstructed starting from the knowledge on the unregularized algebra~$\mycal{A}_\Omega^\circ$.
\end{Remark}

In the next section we proceed as in the regularized case and show how the light-cone structure is encoded in the local algebras.

\subsection{Detecting the Light Cone}\label{detectionunreg}
In this section, we analyze how the causal structure can be retrieved by looking at suitable features of the unregularized
local algebras.
Again, due to translation invariance, it suffices to consider light cone centered at the origin of Minkowski space.

Let us start our analysis with an open set which does not intersect the null cone.
In this case, as already anticipated in the regularized case, the quantities of interest remain bounded in the limit $\varepsilon \searrow 0$. The key is the following estimate.

\begin{Lemma}\label{teoremalmite}
	Let $\Omega\subset\R^{1,3}$ be an open set which does not intersect the null cone~$L_0$.
	Then  for every $A\in\A^\circ_\Omega$ there exists a constant~$c(A)>0$ such that for every $\chi\in\C^4$,
	\begin{equation}\label{equazione}
	\sup_{\varepsilon\in (0,\varepsilon_{max})} \left\|A\,u_{0,\chi}^\varepsilon\right\|\le c(A)\:|\chi| \:.
	\end{equation}
\end{Lemma}
\begin{Remark}\label{remarkreasonnidentity}
	Before proceeding with the proof, we point out that estimate~\eqref{equazione} would not be true if the algebra $\A^\circ_\Omega$ was chosen to contain the identity operator. Indeed, in this case the element~$A=\bI$ (but also others) would fail to fulfill \eqref{equazione}, because the norm of the vectors $u_{0,\chi}^\varepsilon$ blows up in the limit $\varepsilon\searrow 0$ (see Proposition \ref{formulaprodottoscalarelocalizzati}(ii)).
\end{Remark}
\begin{proof}
We begin by considering the generating elements $A=A^\circ_f\in \A_\Omega$. Notice that $A^\circ_f=\lim_{\delta\to 0}A_f^\delta$ in the strong topology, as proven in Theorem \ref{properreg}. Then, we have
\begin{align}
	&\|A_f^\delta\, u_{0,\chi}^\varepsilon\|^2=\la A_f^\delta\, u_{0,\chi}^\varepsilon| A_f^\delta\, u_{0,\chi}^\varepsilon\ra=\notag \\
&=	\int_{\R^4} d^4x \int_{\R^4}d^4 y \,f(x)^*\,f(y)\, \left\langle F^\delta(x)u_{0,\chi}^\varepsilon\big| F^\delta(y)u_{0,\chi}^\varepsilon\right\rangle =\notag \\
	&=(2\pi)^2\int_{\R^4}\!d^4x\int_{\R^4}d^4 y \,f(x)^*\,f(y)\,\langle P^{\delta}(\,\cdot\,,x)\gR_\delta (u_{0,\chi}^\varepsilon)(x)\big|P^{\delta}(\,\cdot\,,y)\gR_\delta (u_{0,\chi}^\varepsilon)(y)\rangle =   \notag \\
	&=-2\pi\int_{\R^4}d^4x\int_{\R^4}d^4 y\; \Sl f(x)\,\gR_\delta(u_{0,\chi}^\varepsilon)(x)\:|\: P^{2\delta}(x,y)\, f(y)\, \gR_\delta(u_{0,\chi}^\varepsilon)(y) \Sr =\notag \\
	&=-2\pi\, P^{2\delta}(\gamma^0 f\,\gR_\delta (u_{0,\chi}^\varepsilon),f\,\gR_\delta(u_{0,\chi}^\varepsilon))
	\label{inequalitylimit}
\end{align}
(cf.~Theorem \ref{listpropertiesLCF} (ii) for the first equality and Proposition \ref{propositioncontinuityP1} for the definition of $P^{\sigma}(h,g)$).
The function~$u_{0,\chi}^\varepsilon$
	belongs to $\hat{\mathrm{E}}({\mathcal{S}}(\R^3,\C^4))$. Indeed, using~\eqref{bidistributionP},
	$$
	u_{0,\chi}^\varepsilon(x)=P^\varepsilon(x,0)\chi=-\int_{\R^3}\frac{d^3\V{k}}{(2\pi)^{4}}\,\mathfrak{g}_\varepsilon(\V{k})\,p_-(\V{k})\gamma^0\chi\, e^{-ik\cdot x} \:,
	$$ 
showing that~$u_{0,\chi}^\varepsilon=\hat{E}(-(2\pi)^{-5/2}\mathfrak{g}_\varepsilon\,p_-\gamma^0\chi)\in\hat{\mathrm{E}}(\mathcal{S}(\R^3,\C^4))$.
	Therefore, applying Proposition~\ref{propositionregularization}~(v) gives
	$$
	\gR_\delta\big(u_{0,\chi}^\varepsilon\big)\to u_{0,\chi}^\varepsilon\ \mbox{ uniformly on compact subsets of $\R^{1,3}$ if $\delta\searrow 0$} \:.
	$$
	 In particular, as the function $f$ is compactly supported within $\Omega$, we see that 
	 $
	 f\,\gR_\delta\big( u_{0,\chi}^\varepsilon\big)$  converges uniformly to $fu_{0,\chi}^\varepsilon$ if $\delta\searrow 0$.
	Of course, the same holds true for $\gamma^0\,f\,\gR_\delta \big( u_{0,\chi}^\varepsilon\big)$ which then converges uniformly to $\gamma^0\,f\, u_{0,\chi}^\varepsilon$. To summarize,
\[	\begin{cases}
	F_\delta&:=\gamma^0 f\,\gR_\delta\big(u_{0,\chi}^\varepsilon\big)\to \gamma^0f\,u_{0,\chi}^\varepsilon:= F\\[0.2em]
	G_\delta&:=f\,\gR_\delta\big(u_{0,\chi}^\varepsilon\big)\to f\,u_{0,\chi}^\varepsilon:= G
	\end{cases}
	\quad\mbox{in the uniform topology}\:. \]

We next verify that the above convergence holds in the stronger sense of the topology of the test function space  $\mathcal{D}(\Omega)$. As a first step, integration-by-parts shows gives
	$$
	\partial_\alpha \big(\gR_\delta\big(u_{0,\chi}^\varepsilon\big)\big) = \gR_\delta \big(\partial_\alpha \big(u_{0,\chi}^\varepsilon\big)\big)\quad\mbox{for any multi-index $\alpha\in\N^4$} \:.
	$$ 
Moreover, since the derivative in position space corresponds to a polynomial multiplication in momentum space it can be shown that (see also \cite[Lemma 8.1]{oppio}) 
$$
\partial_\alpha\big(u_{0,\chi}^\varepsilon\big)\in \hat{\mathrm{E}}( {\mathcal{S}}(\R^3,\C^4))\quad\mbox{for any multi-index $\alpha\in\N^4$}.
$$ Therefore, reasoning as above, it follows that for every multi-index $\alpha\in\N^4$
$$
\partial_\alpha \big(\gR_\delta\big(u_{0,\chi}^\varepsilon\big)\big)\to \partial_\alpha\big(u_{0,\chi}^\varepsilon\big)\ \mbox{uniformly on compact subsets of $\R^{1,3}$ if ~$\delta \searrow 0$} .
$$  
Using this result together with the Leibniz rule, it follows that,
	for any multi-index $\alpha$, the functions~$\partial_\alpha F_\delta$ and $\partial_\alpha G_\delta$ converge uniformly  to $\partial_\alpha F$ and~$\partial_\alpha G$, respectively. Moreover, 
	as the supports of the functions~$F_\delta$, $G_\delta$ and~$F$, $G$ are contained in the support of~$f$,
	we conclude that 
\[ F_\delta\to F\quad\mbox{and}\quad
G_\delta\to G
	\quad\mbox{in the topology of  $\ \mathcal{D}(\R^{1,3},\C^4)$} \:. \]
	Therefore, the convergence holds also in the topology of $\mathcal{S}(\R^{1,3},\C^4)$, as the embedding $\mathcal{D}(\R^{1,3},\C^4)\subset \mathcal{S}(\R^{1,3},\C^4)$ is continuous.

At this point, we have:
\begin{align*}
	&|P^{2\delta}(F_\delta, G_\delta)-P(F,G)|\le \notag \\
	&\le 	|P^{2\delta}(F_\delta, G_\delta)-P^{2\delta}(F_\delta, G)| + |P^{2\delta}(F_\delta,G)-P^{2\delta}(F,G)|+|P^{2\delta}(F,G)-P(F,G)|\le \notag \\
	&\le 	|P^{2\delta}(F_\delta, G_\delta-G)| + |P^{2\delta}(F_\delta-F,G)|+|P^{2\delta}(F,G)-P(F,G)|\le \notag \\
	&\le  c\:\|F_\delta\|_{6,0}\|G_\delta-G\|_{6,4}+c\|F_\delta-F\|_{6,0}\|G\|_{6,4}+|P^{2\delta}(F,G)-P(F,G)|\stackrel{\delta\searrow 0}{\longrightarrow} 0
\end{align*}
where we used Proposition \ref{propositioncontinuityP1}(v)-(vi).
We conclude that~$P^{2\delta}(F_\delta, G_\delta)$ converges to $P(F, G)$ as~$\delta\searrow 0$.
Going back to~\eqref{inequalitylimit}, we have
\begin{align*}
	\|A^\circ_f\, u_{0,\chi}^\varepsilon\|^2&=\lim_{\delta\to 0}\|A^\delta_f\, u_{0,\chi}^\varepsilon\|^2=-2\pi P(\gamma^0\,f\,u_{0,\chi}^\varepsilon,f\,u_{0,\chi}^\varepsilon)\le \\
	&\le 2\pi\sum_{\mu,\nu=1}^4  |\overline{\chi^\mu}\,\chi^\nu| \big|P(\gamma^0\,f\,u_{0,\mu}^\varepsilon,f\,u_{0,\nu}^\varepsilon)\big|\le\\
	&\le \left(2\pi\,c\,\sum_{\mu,\nu=1}^4  \left\|fu_{0,\mu}^\varepsilon\right\|_{6,0}\,\left\|fu_{0,\nu}^\varepsilon\right\|_{6,4}\right)|\chi|^2 \:,
\end{align*}
where we again used the notation $u_{0,\kappa}^\varepsilon:=P^\varepsilon(\,\cdot\,,0)\mathfrak{e}_\kappa$.
In order to conclude the proof in the case~$A=A^\circ_f$, it remains to show that the sum in the last line converges in the limit $\varepsilon \searrow 0$. To this end, we note that the function $P(\,\cdot\,,0)\mathfrak{e}_\kappa$ is smooth on $\R^{1,3}\setminus L_0$ (see the discussion before Lemma~\ref{lemma26}).  Moreover, $P^{\varepsilon}(\,\cdot\,,0)\mathfrak{e}_\kappa$ converges locally uniformly to $P(\,\cdot\,,0)\mathfrak{e}_\kappa$ on $\Omega$
(see Lemma~\ref{lemma26}).
Reasoning in a similar way as before and using Lemma \ref{lemma26}, we conclude that~$f\,P^{\varepsilon}(\,\cdot\,,0)\mathfrak{e}_\kappa\to f\,P(\,\cdot\,,0)\mathfrak{e}_\kappa$ in~$\cD(\R^{1,3},\C^4)$, and therefore also in~$ {\mathcal{S}}(\R^{1,3},\C^4)$.
By continuity of the Schwartz norms $\|\cdot\|_{p,q}$ it follows that the above sum converges in the limit $\varepsilon\searrow 0$.

Finally, we generalize the above proof to the case of arbitrary $A\in\mycal{A}^\circ_\Omega$.  By linearity, it suffices to consider 
monomials of operators $A_f$. We have
	$$\|A^\circ_{f_{1}}\cdots A^\circ_{f_{n}}\left(u_{0,\chi}^\varepsilon\right)\|\le \|A^\circ_{f_{1}}\|\cdots\|A^\circ_{f_{{n-1}}}\|\|A^\circ_{f_{n}}\, u_{0,\chi}^\varepsilon\| \:.
	$$
The first part of the proof implies that the last factor is uniformly bounded in $\varepsilon$ by $c(f_{k})|\chi|$. The other factors are finite constants, concluding the proof.
\end{proof}

We are now ready to specify how to detect the light cone.
The strategy is the same as in Sections~\ref{seclc1} and~\ref{seclc2}: we evaluate the matrix elements of the operators in $\A^\circ_\Omega$ on the space $S_0$ and then study the dependence on $\varepsilon$.

\begin{Thm}\label{teoremalmite2}
Let $\Omega\subset\R^{1,3}$ be an open set which does not intersect the null cone~$L_0$.
Then for every $A\in \A^\circ_\Omega$ there exists constant $c(A)>0$ such that for all $\chi,\zeta\in\C^4$,
	$$
	\sup_{\varepsilon\in(0,\varepsilon_{max})}\big|\langle u_{0,\chi}^\varepsilon|A\, u_{0,\zeta}^\varepsilon\rangle|\le c(A)|\chi||\zeta| \:.
	$$
\end{Thm}
\begin{proof}
	It again suffices to consider monomials~$A^\circ_{f_1}\cdots A^\circ_{f_n}\in\A^\circ_\Omega$.
	In the case~$n\ge 2$,
\begin{align*}
	|\langle u_{0,\chi}^\varepsilon\,\big|\,A^\circ_{f_1}\cdots A^\circ_{f_n} u_{0,\zeta}^\varepsilon\rangle| &= |\langle A^\circ_{\overline{f_1}}\,u_{0,\chi}^\varepsilon\,\big|\,A^\circ_{f_2}\cdots A^\circ_{f_{n-1}} A^\circ_{f_n}u_{0,\zeta}^\varepsilon\rangle|\le\\
	&\le \big\|A^\circ_{\overline{f_1}}\,u_{0,\chi}^\varepsilon\big\|\|A^\circ_{f_2}\cdots A^\circ_{f_{n-1}}\|\big\|A^\circ_{f_n}\,u_{0,\zeta}^\varepsilon\big\|,
\end{align*}
	which is uniformly bounded in $\varepsilon$ by some $c(f_1,\dots,f_n)|\chi||\zeta|$ thanks to Lemma \ref{teoremalmite}. In the remaining case~$n=1$,
\begin{align*}
\left|\langle \,u_{0,\chi}^\varepsilon\,|A^\circ_{f_1} \,u_{0,\zeta}^\varepsilon\rangle\right|&\le\int_{\R^4} |f_1(x) \:
	\Sl P^{\varepsilon}(x,0)\chi | P^{\varepsilon}(x,0)\zeta\Sr| \, d^4x\le \\
	&\le \int_{\R^4}|f_1(x)|\big(\|P^\varepsilon{(0,x)}\|_2\big)^2\,|\chi||\zeta| \:.
\end{align*}
Again invoking Lemma \ref{lemma26}, the integral converges in the limit $\varepsilon \searrow 0$.
This concludes the proof.
\end{proof}

We next consider the case when~$\Omega$ intersects the null cone.
In this case, it is indeed possible to find a suitable function~$f$ supported near the null cone
in~$\Omega$ for which some of the above matrix elements do not go to zero in the limit~$\varepsilon \searrow 0$.

\begin{Thm}\label{theoremblowup}
Let $\Omega\subset\R^{1,3}$ be an open set which intersects the null cone~$L_0$.
Then there exists $f\in C_0^\infty(\Omega)$ and $\chi,\zeta\in\C^4$ such  that
	$$
	\sup_{\varepsilon\in (0,\varepsilon_{max})}|\langle u_{0,\chi}^\varepsilon\big|A^\circ_f\, u_{0,\zeta}^\varepsilon\rangle|=\infty \:.
	$$
	More precisely, there exists a positive constant $c$ such that for all sufficiently small $\varepsilon>0$,
	$$
	|\langle u_{0,\chi}^\varepsilon\big|A^\circ_f\, u_{0,\zeta}^\varepsilon\rangle|\ge \frac{1}{c\,\varepsilon^2} \:.
	$$
\end{Thm}
\begin{proof}
Let $f\in C_0^\infty(\R^{1,3},\C)$ be arbitrary for now, and take $a=\mathfrak{e}_1$ and $b=\mathfrak{e}_3$. Then
\begin{align*}
	\langle u_{0,1}^\varepsilon\big|A^\circ_f\, u_{0,3}^\varepsilon\rangle&=-\int_{\R^4} f(x)\Sl P^{\varepsilon}(x,0)\mathfrak{e}_1| P^{\varepsilon}(x,0)\mathfrak{e}_3\Sr\, d^4x=\\
	&=-\int_{\R^4} f(x)\, \Sl \mathfrak{e}_1|P^{\varepsilon}(0,x)\, P^{\varepsilon}(x,0)\,\mathfrak{e}_2\Sr\,d^4x=\\
	&=-\int_{\R^4}f(x)\, \mathfrak{e}_1^\dagger\, A_{0x}^\varepsilon\, \mathfrak{e}_3\, d^4x \:.
\end{align*}
The result follows by choosing $f$ as in Theorem \ref{theoremasymptotics}~for some point $x\in\Omega\cap L_0\setminus\{0\}$.
\end{proof}

To conclude this section we prove that the unregularized algebras $\A_\Omega^\circ$ localized within interior light cones do indeed not contain the identity, as already anticipated in Remark \ref{remarkmaybeidentityunreg}.
\begin{Prp}\label{propositionidentity}
	Let $\Omega\subset\R^{1,3}$ be open and contained in the interior light cone $I_x$ of a spacetime point $x\in\R^{1,3}$. Then $\A_\Omega^\circ$ does not contain the identity operator.
\end{Prp}
\begin{proof}
Assume by contradiction that $\bI\in\A_\Omega^\circ.$ Let~$x$ and~$\Omega$ as in the hypothesis, so that $L_x\cap \Omega=\varnothing$. Lemma \ref{teoremalmite} implies the existence of a constant $c>0$ such that, for any $\chi\in\C^4$,
	$$
	\|u_{x,\chi}^\varepsilon\|=\|\bI\,u_{x,\chi}^\varepsilon\|\le c|\chi|\quad\mbox{for any }\varepsilon>0 \:.
	$$
	This is a contradiction, because the norm of $u_{x,\chi}^\varepsilon$ diverges if $\varepsilon\searrow 0$ (see \eqref{estimatenormlocstate}).
\end{proof}

\subsection{The Commutator of Spacetime Points and Local Algebras} \label{seccommunreg}
Similar as in Section \ref{commutators}, we now analyze the limit $\varepsilon\searrow 0$ of commutators of the type
$$
[F^\varepsilon(x),A]\quad\ \mbox{for }\  A\in\A_\Omega^\circ\ \mbox{ and }\ \Omega\subset\R^{1,3}
$$ 
for some fixed $x\in\R^{1,3}$.
As in the regularized case, the scaling of this commutator depends on whether $\Omega$ does or does not intersect the null cone
centered at~$x$, making it possible to recover the light-cone structure from the local algebras. Again, it suffices to focus on the null cone centered at the origin of Minkowski space.
\begin{Prp}\label{Prpcommunreg}
	Let $\Omega$ be an open set which does not intersect the null cone~$L_0$. Then
	for any $A\in \A_\Omega^\circ$ there is a constant $c(A)>0$ such that for all~$\varepsilon>0$,
$$
\|\,[F^\varepsilon(0),A]\,\|\le  2\, c(A)\, \varepsilon^{\frac{3}{2}}\,\|F^\varepsilon(0)\| \:.
$$
On the other hand, if $\Omega$ does intersect the null cone~$L_0$, then
there exists $A\in \A_\Omega^\circ$ and a constant $c>0$ such that
$$
\big\|\,[F^\varepsilon(0),A]\,\big\|\ge c\:\varepsilon\,\|F^\varepsilon(0)\| \:.
$$
for sufficiently small $\varepsilon$.
\end{Prp} \noindent
The proof is analogous to the unregularized case (see Theorem~\ref{prpiotaes} and Proposition~\ref{prpiotaes2})
and will be omitted.

As in the regularized case, the $\varepsilon$-behaviors of the two estimates in this proposition are incompatible
in the limit $\varepsilon\searrow 0$, making it possible to recover the light cone structure from
the commutators.

	\section{Conclusions and Outlook} \label{outlook}
This paper is a first step towards an algebraic description of the theory of causal fermion systems. It was 
shown, in the example of a regularized Dirac sea vacuum in Minkowski space, that it is possible to introduce  a meaningful notion of a local algebra in the theory of causal fermion systems.  These algebras
show similarities and share elementary features with the typical structures present in the modern algebraic approaches to quantum theory.  However, substantial differences arise in regard to the principle of causality, as the canonical commutation relations fail to hold. This is not a physical issue because these algebras should not be understood as formed by local observable quantities. Rather, these new structures provide distinguished sets of operators which carry information on the local structure of spacetime. We point out for clarity that, although the canonical commutation relations are not fulfilled, the light cone structure of spacetime can still be reconstructed with a different method, namely by studying specific expectation values of the elements of the local algebras.

In conclusion, the notion of a local algebra introduced here departs from the standard structure at the heart of the algebraic approaches to quantum theory and provides a distinguished way to encode the information on the local structure of spacetime in the language of operators.
 
 Despite the obvious physical limitations of the example considered here (which describes the Minkowski vacuum) the results in this work lay solid mathematical and conceptual foundations for future developments. 
Indeed, we expect that most of our results could be carried over to causal fermion systems
describing physical spacetimes of a more interesting nature, like in the presence of gravitational fields and matter. 
In this context, one could ask whether the algebra also encodes information on
the geometry and/or the matter fields (like particle densities, field strengths, curvature of spacetime, etc.).
Generally speaking, the operator algebras contain part of the information encoded in the causal fermion system.
It seems an interesting program to explore what information other than causality can be extracted
from them.

We finally remark that carrying out this program seems quite challenging. For example, here we made essential use
of the homogeneity of the fermionic projector of the Minkowski vacuum, making it possible to use Fourier methods.
Clearly, this is no longer possible in curved spacetime, making the whole analysis considerably harder.
Nevertheless, this program is of physical and mathematical relevance and seems worth being
pursued in the future.

\appendix
\section{Technical Proofs} \label{secappendix}

\begin{proof}[Proof of Proposition \ref{propositionregularization}]
We begin with point (i). The regularization operator is defined in three-momentum space by multiplication with a scalar function. Therefore, it does not affect the sign of the energy. 
Let us prove that $\gR_{\varepsilon} u$ is indeed smooth. 
For simplicity, we only consider the negative energy case. The proof can easily be extended to the whole space $\scH_m$.
	By construction, there exists $\psi\in \hat{P}_-(\scL^2(\R^3,\C^4))$ such that $u=\hat{\mathrm{E}}(\psi)$ and thus, by definition, $\gR_{\varepsilon} u = \hat{\mathrm{E}}(\mathfrak{g}_\varepsilon\psi)$.
By H\"{o}lder's inequality it follows that $\mathfrak{g}_\varepsilon\psi\in\scL^1(\R^3,\C^4)$. 
Thus we may define
	\begin{equation}\label{vpsi}
	v_\psi(x):= \int_{\R^3}  \frac{d^3\V{k}}{(2\pi)^{3/2}}\,\mathfrak{g}_\varepsilon(\V{k})\psi(\V{k})\, e^{-ik\cdot x}\quad\mbox{for all } x\in\R^{1,3} \:,
	\end{equation}
	where we used the compact notation $k\cdot x=-\omega(\V{k})t-\V{k}\cdot\V{x}$. For simplicity of notation, we also introduce
	$$
	{\varphi}_x:= \mathfrak{g}_\varepsilon\,\psi\, e^{-ik\cdot x}\in\scL^1(\R^3,\C^4)\quad\mbox{for any } x\in\R^{1,3}.
	$$
	Notice that the mapping $x\mapsto \varphi_x(\V{k})$ is smooth for every fixed $\V{k}\in\R^3$. Moreover, differentiating with respect to $x$, for any multi-index $\alpha \in\N^4$ we get
	\begin{equation}\label{identityderiv}
	|D^\alpha \varphi_x(\V{k})| = |k^\alpha \mathfrak{g}_\varepsilon\psi (\V{k})|\quad\mbox{for every }\V{k}\in\R^3,
	\end{equation}
	where $k^0=-\omega(\V{k})$. 
	Notice that for any multi-index $\alpha$, the function
	$k^\alpha \mathfrak{g}_\varepsilon$ is in~$ {\mathcal{S}}(\R^3,\C^4)$. This follows from the fact that the Schwartz space is closed under multiplication by polynomials and by \cite[Lemma 8.1]{oppio}.
	In particular, as $\psi$ is a $\scL^2$ function, we have $k^\alpha \mathfrak{g}_\varepsilon\, \psi \in\scL^1(\R^3,\C^4)$. From this fact and identity \eqref{identityderiv} Lebesgue's dominated convergence theorem (or more precisely \cite[Theorem 1.88]{moretti-book}) implies that the function $v_\psi$ is differentiable to every order and that partial derivatives and the integral may be interchanged.
We conclude that $v_\psi$ is a smooth solution of the Dirac equation. At this point, the proof of point (i) is concluded once we prove that $\hat{\mathrm{E}}(\mathfrak{g}_\varepsilon\psi)=v_\psi$.
	By denseness, there exists a sequence $\psi_n\in \hat{P}_-( {\mathcal{S}}(\R^3,\C^4))$ which converges to $\psi$ in the $\scL^2$-norm. Then, by continuity and boundedness of the cutoff function, we get simultaneously:
	\begin{equation}\label{twoconverg}
	\hat{\mathrm{E}}(\mathfrak{g}\psi_n)\to \hat{\mathrm{E}}(\mathfrak{g}\psi)\quad\mbox{and}\quad v_{\psi_n}(x)\to v_\psi(x)\ \mbox{for all }x\in\R^{1,3},
	\end{equation}
	where $v_{\psi_n}$ is defined analogously to  \eqref{vpsi}. 
	From \eqref{expressionEonS} we know that $v_{\psi_n}=\hat{\mathrm{E}}(\mathfrak{g}_\varepsilon\psi_n)$ for every $n\in\N$. Now, current conservation (see \cite[Lemma 2.7]{oppio}) implies that the limit
	$\hat{\mathrm{E}}(\mathfrak{g}_\varepsilon\psi_n)\to \hat{\mathrm{E}}(\mathfrak{g}_\varepsilon\psi)$ can be restated as a $L^2$- convergence on any stripe $R_T=\{(t,\V{x})\:|\: |t|\le T \}$:
	$$
	\int_{R_T} |\hat{\mathrm{E}}(\mathfrak{g}_\varepsilon\psi_n)(x)- \hat{\mathrm{E}}(\mathfrak{g}_\varepsilon\psi)(x)|^2=\sqrt{2T}\|\hat{\mathrm{E}}(\mathfrak{g}_\varepsilon\psi_n)- \hat{\mathrm{E}}(\mathfrak{g}_\varepsilon\psi)\|^2\to 0
	$$ 
As a consequence, choosing a countable exhaustion of $\R^{1,3}$ made of such stripes, it is possible to find a subsequence $\hat{\mathrm{E}}(\mathfrak{g}_\varepsilon\psi_{n_k})$ which converges pointwise to $\hat{\mathrm{E}}(\mathfrak{g}_\varepsilon\psi)$ almost everywhere on $\R^{1,3}$. The limits \eqref{twoconverg}, then, imply that 
 $v_\psi=\hat{\mathrm{E}}(\mathfrak{g}\psi)$ almost everywhere concluding the proof.

Point  (ii) can be proved directly using the fact that $\mathfrak{g}_\varepsilon$ is real-valued.

We now prove point (iii). The first statement follows from the definition of $\gR_{\varepsilon}$ and from the fact that $\mathfrak{g}_\varepsilon$ is strictly positive everywhere. The second statement can be proved by using  Parseval's identity and the fact that $|\mathfrak{g}_\varepsilon(\V{k})|\le 1$. 

Let us prove point (iv).
Let $u=\hat{\mathrm{E}}(\psi)$ with $\psi\in L^2(\R^3,\C^4)$ be any element of $\H_m^-$. Then  $|\mathfrak{g}_\varepsilon\psi-\psi|^2\in\scL^1(\R^3,\C^4)$ and $\mathfrak{g}_\varepsilon\psi\to \psi$ pointwise in the limit $\varepsilon\searrow 0$. Since $|\mathfrak{g}_\varepsilon|\le 1$, we can apply Lebesgue's dominated convergence theorem to infer that $\|\mathfrak{g}_\varepsilon\psi-\psi\|_{\scL^2}\to 0$. In other words $\gR_{\varepsilon}u\to u$. 

It remains to prove point (v). Again, it suffices to focus on the negative-energy subspace. So, take  $u=\hat{\mathrm{E}}(\psi)$ for some $\psi\in \hat{P}_-( {\mathcal{S}}(\R^3,\C^4))$. The functions $u$ and $\gR_{\varepsilon} u$ are both continuous (cf.~\eqref{expressionEonS}). Now, fix any compact subset $K\subset\R^{1,3}$. Using  that $\mathfrak{g}_\varepsilon\psi\in\mathcal{S}(\R^3,\C^4)$, we get
	\begin{equation*}
	\begin{split}
	\sup_{x\in K}|\gR_{\varepsilon}u(x)-u(x)|&=\sup_{x\in K}\left|\int_{\R^3}\frac{d^3\V{k}}{(2\pi)^{3/2}}\,\big( \mathfrak{g}_\varepsilon(\V{k})\psi(\V{k})- \psi(\V{k})\big) e^{-ik\cdot x} \right|\le\\
	&\le \int_{\R^3} \frac{d^3\V{k}}{(2\pi)^{3/2}}|\mathfrak{g}_\varepsilon(\V{k})-1||\psi(\V{k})| \:.
	\end{split}
	\end{equation*}
Now the result follows from Lebesgue's dominated convergence theorem, noting that $\psi\in\scL^1$, $\mathfrak{g}_\varepsilon(\V{k})\to 1$ as $\varepsilon\searrow 0$, and $|\mathfrak{g}_\varepsilon|\le 1$ uniformly in $\varepsilon$.
\end{proof}

\begin{proof}[Proof of Proposition \ref{propositioncontinuityP1}]
	The proof of points (i), (ii) and (iii) can be found in \cite[Propositions~2.23 and~3.19]{oppio} (with minor adjustments due to a different choice of the cutoff function). 
We now prove two statements that are needed for the proof of the remaining points. Fix any $g\in \mathcal{S}(\R^{1,3},\C^4)$. By definition of the Schwartz space, there exists a constant $A$ such that:
	$$
	(1+|k|^4) \:|\cF(g)(k)|\le A\:\|\cF(g)\|_{4,0}
	$$
	(where $\|\cdot\|_{p,q}$ are the usual Schwartz norms, see the footnote after Proposition \ref{propositioncontinuityP1}). 
	At this point, let $n=0,1,2$. Then (see \eqref{defPf})
	\begin{equation*}
	\begin{split}
	|P^{n\varepsilon}(x,g)|&\le \bigg|\int_{\R^4}\frac{d^4k}{(2\pi)^2}\, \delta(k^2-m^2)\Theta(-k^0)(\slashed{k}+m) \gG_\varepsilon(k)^n\cF(g)(k)\, e^{-ik\cdot x}\bigg|\le\\
	&\le\|\mathfrak{g}_\varepsilon^n\|_\infty\int_{\R^3}\frac{d^3\V{k}}{(2\pi)^2 2\omega(\V{k})}\left\|(\slashed{k}+m)\big|_{k^0=-\omega(\V{k})}\right\|_2 |\cF(g)(-\omega(\V{k}),\V{k})|\le \\
	&\le \left(A\|\mathfrak{g}_\varepsilon^n\|_\infty\int_{\R^3}\frac{d^3\V{k}}{(2\pi)^2}\frac{\omega(\V{k})\|\gamma^0\|_2+\sum_{i=1}^3|k^i|\|\gamma^i\|_2+m}{2\omega(\V{k})(1+(\omega(\V{k})^2+\V{k}^2)^2)} \right)\|\cF(g)\|_{4,0} \:.
	\end{split}
	\end{equation*}
	The integral between parentheses is well-defined and convergent. Moreover, notice that
	$
	\|\mathfrak{g}_\varepsilon^2\|_\infty\le1
	$
	uniformly in $\varepsilon$. 
	Now, exploiting \cite[Lemma 8.2.2 and eq.~(8.2.2)]{friedlander2} (where a different convention for the Schwartz norm is adopted), one finds that that
	$$
	\|\cF(g)\|_{4,0}\le K \:\|g\|_{6,4}
	$$
	for some constant $K$. Combining the above, we have proved that for some constant $k$,
	\begin{equation}\label{firstinequality}
	|P^{n\varepsilon}(x,g)|\le k\|g\|_{6,4} \ \mbox{ for any $x\in\R^{1,3},\;\varepsilon\in (0,\varepsilon_{max})$ and $n=0,1,2$}\:.
	\end{equation}
As a first consequence of this inequality, we now prove point (iv). Given~$f,g$ as in the hypotheses,
the previous estimates imply that
\begin{equation}\label{secondinequality}
 \int_{\R^4}|f(x)||P^{n\varepsilon}(x,g)|\,d^4 x\le D\|g\|_{6,4}\|f\|_{\scL^1}<\infty \:.
\end{equation}
This proves that the integral in point (iv) is well-defined.

We now prove point (v). The first inequality was proven in~\eqref{firstinequality}. The second inequality follows
from~\eqref{secondinequality} together with the general inequality~$\|f\|_{\scL^1}\le A\:\|f\|_{6,0}$
(valid for some constant $A>0$), which can be found in~\cite[eq.~(8.8.2)]{friedlander2}.

It remains to prove point~(vi). As a first step, Lebesgue's dominated convergence theorem implies that (see again \eqref{defPf})
	\begin{equation}\label{firstconvergence}
	\begin{split}
	P^{n\varepsilon}(x,g)&=-\int_{\R^3}\frac{d^3\V{k}}{(2\pi)^2}\,\mathfrak{g}_\varepsilon^n(\V{k})\, p_-(\V{k})\,\gamma^0\,\cF(g)(k)\,e^{-ik\cdot x} \stackrel{\varepsilon\searrow 0}{\longrightarrow}\\
	&\stackrel{\varepsilon\searrow 0}{\longrightarrow} -\int_{\R^3}\frac{d^3\V{k}}{(2\pi)^2}\, p_-(\V{k})\,\gamma^0\,\cF(g)(k)\,e^{-ik\cdot x}=P(x,g),
	\end{split}
	\end{equation}
which proves the first statement.
In order to prove the second statement, notice that $|P^{n\varepsilon}(x,g)-P(x,g)|\le C$ uniformly in $x$ and $\varepsilon$ for every $n=0,1,2$ as a consequence of \eqref{firstinequality}.
 This, together with \eqref{firstconvergence} as well as the fact that $f\in\scL^1$ and Lebesgue's dominated convergence theorem, we conclude that
	\begin{equation*}
	\begin{split}
	P^{n\varepsilon}(f,g)-P(f,g)&=\int_{\R^4}f(x)^\dagger(P^{n\varepsilon}(x,g)-P(x,g))\to 0.
	\end{split}
	\end{equation*}
The proof is complete.
\end{proof}

\begin{proof}[Proof of Remark \ref{remarkbehavior}.]
We only analyze the first summand (the proof for the second summand is analogous). Writing the integral $\|\mathfrak{g}_\varepsilon^2\|_{\scL1}$ in spherical coordinates we have
\begin{align*}
	\int_{\R^3} \mathfrak{g}_\varepsilon(\V{k})^2\, d^3\V{k}&=C\, \int_{0}^\infty e^{-2\varepsilon\sqrt{r^2+m^2}}r^2dr\stackrel{s=\varepsilon r}{=} \varepsilon^{-3}\,\underbrace{C\, \int_{0}^\infty e^{-2\sqrt{s^2+(\varepsilon m)^2}} s^2 ds}_{:=g(\varepsilon)}
\end{align*}
for a strictly positive constant~$C$. Obviously, $g(\varepsilon)$ converges to a strictly positive number in the limit $\varepsilon\searrow 0$. This concludes the proof.
\end{proof}

\begin{proof}[Proof of Proposition \ref{proprietaAf}]
The proof of point (i) is straightforward. In preparation of the proof of point~(ii),
let $R_T:=\{|t|\le T \}$.
We first consider the case that~$f$ is real-valued. Then $A_f$ is selfadjoint and
\begin{align*}
	\|A_f^\varepsilon\|&=\sup_{\|u\|=1}|\la u|A_f^\varepsilon u\ra |=\sup_{\|u\|=1}\left|\int_{R_T}\,f(x)\,\Sl\gR_{\varepsilon} u(x)|\gR_{\varepsilon} u(x)\Sr\, d^4x\right|\le\\
	&\le  \sup_{\|u\|=1}\int_{R_T}\left|f(x)\,\Sl \gR_{\varepsilon} u(x)|\gR_{\varepsilon} u(x)\Sr\right|\, d^4x\le\\
	&\le \|f\|_\infty \sup_{\|u\|=1}\int_{R_T}\left|\Sl \gR_{\varepsilon} u(x)|\gR_{\varepsilon} u(x)\Sr\right|\, d^4x\le\\
	&\le \|f\|_\infty \sup_{\|u\|=1} \int_{R_T}|\gR_{\varepsilon} u(x)|^2
	\, d^4x\stackrel{(*)}{=} \|f\|_\infty\, 2T \sup_{\|u\|=1} \|\gR_{\varepsilon} u\|^2 \le 2T \|f\|_\infty \:,
\end{align*}
where in the last inequality we used $\|\gR_\varepsilon\|\le 1$, and $(*)$ follows from current conservation (see \cite[Lemma 2.7]{oppio}).
	Now, let $f=u+iv$ be arbitrary, with $u,v$ real valued. Then 
	$$
	\|A_f\|\le \|A_u\|+\|A_v\|\le 2T(\|u\|_\infty+\|v\|_\infty)\le 4T \|f\|_\infty \:.
	$$

After these preparations, we come to the proof of point (ii).
	By \cite[Proposition 4.41]{moretti-book} a bounded operator $T$ is trace-class if and only if  the series $\sum_{n\in\N}|\la e_n|T e_n\ra |$ converges for every Hilbert basis $(e_n)_n$. Any operator $ F^\varepsilon(x)$
has finite rank and trace $\mbox{tr}_{\rm{vac}}^\varepsilon$. Now, 
for any given Hilbert basis $(e_n)_{n\in\N}$ and $N\in\N$,
\[ \sum_{n=1}^N |\la e_n| A_f^\varepsilon e_n\ra |\le \sum_{n=1}^N \int_\Omega |f(x)| |\la e_n| F^\varepsilon(x)e_n\ra|\, d^4x=\int_\Omega |f(x)| \sum_{n=1}^N  |\la e_n| F^\varepsilon(x)e_n\ra |\, d^4x \:. \]
Here the function~$|\la e_n| F^\varepsilon(\cdot)e_n\ra |$ is continuous and therefore measurable.
Using the decomposition~\eqref{decompositionF}, we have
	\begin{equation}\label{equationtraceclass}
	\begin{split}
	&\sum_{n=1}^N|\la e_n| F^\varepsilon(x)e_n\ra|\le\sum_{n=1}^N|\la e_n| F_+^\varepsilon(x)e_n\ra|+\sum_{n=1}^N|\la e_n| F_-^\varepsilon(x)e_n\ra |=\\
	&=\sum_{n=1}^N\la e_n| F_+^\varepsilon(x)e_n\ra-\sum_{n=1}^N\la e_n| F_-^\varepsilon(x)e_n\ra \stackrel{N\to\infty}{\longrightarrow} 4\pi\nu^+(\varepsilon)-4\pi\nu^-(\varepsilon) \:.
	\end{split}
	\end{equation}
Applying Beppo Levi's convergence theorem, we conclude that
	$$
	\sum_{n=1}^\infty|\la e_n|A_f^\varepsilon e_n\ra |\le 4\pi\:\big(\nu^+(\varepsilon)-\nu^-(\varepsilon)
	\big)\:\|f\|_{\scL^1}<\infty.
	$$
	Since the basis $(e_n)_n$ is arbitrary, we proved that $A_f^\varepsilon$ is trace-class. In order to prove the last statement of the theorem, in any basis $(e_n)_n$ we have
\[ \sum_{n=1}^N\la e_n|A_f^\varepsilon e_n\ra =\int_\Omega f(x)\sum_{n=1}^N \la e_n| F^\varepsilon(x) e_n\ra\, d^4x \:. \]
	At this point, the result follows directly from~\eqref{equationtraceclass} and Lebesgue's dominated convergence theorem.
\end{proof}

\begin{proof}[Proof of Theorem \ref{theoremeactionAf}]
Given~$\psi\in\hat{P}_-(S(\R^3,\C^4))$, we consider the corresponding solution (see \ref{expressionEonS})
	$$	
	u_\psi:=\hat{\mathrm{E}}(\psi)\in \scH_m^-\cap C^\infty(\R^{1,3},\C^4),\quad u_\psi(x)= \int_{\R^3}\frac{d^3\V{k}}{(2\pi)^{3/2}}\, \psi(\V{k})\,e^{-ik\cdot x} \:,
	$$
	where $k\cdot x$ here is an abbreviation for $-\omega(\V{k})t-\V{k}\cdot\V{x}$. 
Choosing the function~$\varphi$ as in the assumption, we obtain
(using again $k:=(-\omega(\V{k}),\V{k})$ and the fact that $\psi(\V{k})=p_-(\V{k})\psi(\V{k})$)
\begin{align*}
	&\la u_\psi|A^\varepsilon_f\big(P(\,\cdot\,,\varphi)\big)\ra=-\int d^4x\, f(x)\, \Sl\gR_{\varepsilon} u_\psi(x)|\gR_{\varepsilon}\big(P(\,\cdot\,,\varphi)\big)(x)\Sr=\\
	&=-\int_{\R^4}d^4x\, f(x) \int_{\R^3}  \frac{d^3\V{k}}{(2\pi)^{3/2}}\,\mathfrak{g}_\varepsilon(\V{k})\,\psi(\V{k})^\dagger\gamma^0\int_{\R^4}\frac{d^4q}{(2\pi)^2}\, \mathfrak{G}_\varepsilon(q)\,\hat{P}(q)\,\hat{\varphi}(q)\,e^{-i(q-k)\cdot x}=\\
	&=\int_{\R^3}d^3\V{k}\, \psi(\V{k})^\dagger\left(- \frac{\mathfrak{g}_\varepsilon(\V{k})}{(2\pi)^{3/2}}\,  p_-(\V{k})\gamma^0\int_{\R^4}d^4q\, \mathfrak{G}_\varepsilon(q)\,\hat{P}(q)\,\hat{\varphi}(q)\, \hat{f}(k-q)\right)=\\
	&=\int_{\R^3}d^3\V{k}\, \psi(\V{k})^\dagger\underbrace{\left(- \frac{\mathfrak{g}_\varepsilon(\V{k})}{(2\pi)^{3/2}}\,  p_-(\V{k})\gamma^0\,\big(\cF(P^\varepsilon(\,\cdot\,,\varphi))*\cF(f)(-\omega(\V{k}),\V{k})\big)\right)}_{\phi(\V{k})} \:.
\end{align*}
	Note that $\phi\in \hat{P}_-(S(\R^3,\C^4))$. Indeed, $f\,P^\varepsilon(\,\cdot\,\,\varphi)\in C_0^\infty(\R^{1,3},\C^4)$ and therefore
	$$
	\cF(P^\varepsilon(\,\cdot\,,\varphi))*\cF(f)=(2\pi)^2\,\cF(f\,P^\varepsilon(\,\cdot\,\,\varphi))\in\mathcal{S}(\R^{1,3},\C^4) \:.
	$$
	As $\psi$ is arbitrary and the solutions $u_\psi$ span a dense subset of $\scH_m^-$, Parseval's identity yields
	\begin{equation*}
	\begin{split}
	A^\varepsilon_f\big(P(\,\cdot\,,\varphi)\big)(x)&= u_\phi(x)=\int_{\R^3}\frac{d^3\V{k}}{(2\pi)^{3/2}}\, \phi(\V{k})\,e^{-ik\cdot x}=\\
	&=-2\pi\,\int_{\R^3} \frac{d^3\V{k}}{(2\pi)^{2}} \,\mathfrak{g}_\varepsilon(\V{k})\, p_-(\V{k})\gamma^0\, \cF(fP^\varepsilon(\,\cdot\,,\varphi))(k)e^{-ik\cdot x}=\\
	&=2\pi\, P^\varepsilon \big(\,\cdot\,,fP^\varepsilon(\,\cdot\,,\varphi) \big) \:,
	\end{split}
	\end{equation*}
where in the last step we used \eqref{defPf}. This concludes the proof.
\end{proof}

\begin{proof}[Proof of Proposition~\ref{prpregularized}.]
As already mentioned, the proof of this proposition relies on the results of Section \ref{seccommunreg}. We divide the proof in two parts, corresponding to points~(i) and~(ii) in the proposition.
We shall make use of the relation~$A_f^\varepsilon=\gR_{\varepsilon} A_f\gR_{\varepsilon}$ (see Theorem \ref{properreg}~(i)).
\vspace{0.1cm}
\begin{itemize}[leftmargin=2.5em]
	\item[{\rm{(i)}}] \textit{The set $\Omega$ does not intersect the null cone.} \\[0.3em]
We first observe that for any fixed $f\in C_0^\infty(\Omega,\C)$ and $\chi,\zeta\in\C^4$,
the inequality
	\begin{equation*}
	\begin{split}
	|\langle u_{0,\chi}^\varepsilon|A_f^\varepsilon\, u_{0,\zeta}^\varepsilon\rangle|&=|\langle \gR_\varepsilon (u_{0,\chi}^\varepsilon)|A_f^\circ\, \gR_{\varepsilon} (u_{0,\zeta}^\varepsilon)\rangle|=|\langle u_{0,\chi}^{2\varepsilon}|A^\circ_f\, u_{0,\zeta}^{2\varepsilon}\rangle|\le c(f)|\chi||\zeta|
	\end{split}
	\end{equation*}
holds uniformly in~$\varepsilon$ for a positive constant $c(f)$, where in the last step we used Theorem~\ref{teoremalmite2}. Similarly, using Lemma \ref{teoremalmite}, it follows that
	$$
	\|A_f^\varepsilon\, u_{0,\chi}^\varepsilon\|\le \|A^\circ_f\,u_{0,\chi}^{2\varepsilon}\|\le c'(f)|\chi|
	$$
	uniformly in $\varepsilon$ for a positive constant $c'(f)$.
	Therefore, for any choice of $f_1,\dots,f_n$ with $n\ge 2$, we have
	\begin{equation*}
		\begin{split}
			|\langle u_{0,\chi}^\varepsilon|A_{f_1}^\varepsilon\cdots A_{f_n}^\varepsilon u_{0,\zeta}^\varepsilon\rangle|&\le c'(\overline{f_1})|\chi|\,\|A_{f_2}^\varepsilon\|\cdots\|A_{f_{n-1}}^\varepsilon\|\,c'({f_n})|\zeta|\le\\
			&\le c''(f_1,\dots,f_n)|\chi||\zeta| \:,
		\end{split}
	\end{equation*}
	again uniformly in $\varepsilon$. In the last step we used Proposition \ref{proprietaAf}~(ii) in 
	order to estimate the norm of the operators $A_{f_i}^\varepsilon$.\\[-0.5em]
	
\item[{\rm{(ii)}}]  \textit{The set $\Omega$ does intersect the null cone.} \\[0.3em]
Choosing the function $f$ and spinors $\chi,\zeta$ as in Theorem \ref{theoremblowup}, 
the results of this theorem imply that, for a positive constant $c$ and for all sufficiently small $\varepsilon$,
	$$
	|\langle u_{0,\chi}^\varepsilon|A_f^\varepsilon\, u_{0,\zeta}^\varepsilon\rangle|=|\langle u_{0,\chi}^{2\varepsilon}|A^\circ_f\, u_{0,\zeta}^{2\varepsilon}\rangle|\ge \frac{1}{c\,\varepsilon^2}.
	$$
\end{itemize}
\vspace{-0.1em}
\end{proof}

\begin{proof}[Proof of Theorem~\ref{prpiotaes}.]
	Also in this proof we make use of some statements and results of Section \ref{secunregularized}.  Again, note that $A_f^\varepsilon=\gR_{\varepsilon}\,A^\circ_f\,\gR_{\varepsilon}$ as shown in Theorem \ref{properreg}~(i).
For the proof it suffices to consider the monomials~$A=A_{f_1}^\varepsilon\cdots A_{f_n}^\varepsilon$. We have
\begin{align*}
&\frac{\|A\, F^\varepsilon(0)\|}{\|F^\varepsilon(0)\|}=\sup_{u\neq 0}\frac{\|A\, F^\varepsilon(0) u\|}{\|F^\varepsilon(0)\|\|u\|}\le \sup_{u\neq 0}\frac{\|A\, F^\varepsilon(0) u\|}{\|F^\varepsilon(0) u\|}=\sup_{w\in S_0}\frac{\|A\, w\|}{\|w\|}=\\
			&= \sup_{\chi\neq 0}\frac{\|A\, u_{0,\chi}^\varepsilon\|}{\|u_{0,\chi}^\varepsilon\|}=\sup_{\chi\neq 0}\frac{\|A_{f_1}^\varepsilon\cdots A_{f_n}^\varepsilon\,u_{0,\chi}^\varepsilon\|}{\|u_{0,\chi}^\varepsilon\|}\le\|A_{f_1}^\varepsilon\cdots A_{f_{n-1}}^\varepsilon\|\sup_{\chi\neq 0}\frac{\|A_{f_n}^\varepsilon\,u_{0,\chi}^\varepsilon\|}{\|u_{0,\chi}^\varepsilon\|}\le\\
			&\le c(f_1,\dots,f_{n-1})\,\sup_{\chi\neq 0}\frac{\|A_{f_n}\,u_{0,\chi}^{2\varepsilon}\|}{\|u_{0,\chi}^\varepsilon\|}\le\  \sqrt{2\pi}\;\frac{c(f_1,\dots,f_{n-1})\,c(f)\,|\chi|}{\sqrt{|\nu^-(\varepsilon)|}|\chi|} \:,
\end{align*}
where in the last inequality we used Lemma \ref{teoremalmite} and~\eqref{estimatenormlocstate}. The claim follows by noticing that $|\nu^-(\varepsilon)|\sim \varepsilon^{-3}$.
\end{proof}

\section{Proof of Irreducibility of the Spacetime Point Operators} \label{secappreg}
In this appendix, we give the proof of Theorem~\ref{maintheorem}.
Throughout the proof we denote the image of the local correlation operator $F^\varepsilon(x)$ by~$S_x$.
Moreover, we let~$\pi_x$ be the orthogonal projection operator on $S_x$.
Also, $$\mathrm{P}(x,y):=\pi_x\,F^\varepsilon(y)\!\restriction_{S_y},$$ whereas~$P^{n\varepsilon}(x,y)$ denotes the distribution kernel introduced in \eqref{bidistributionP}.
The two objects are related to each other by \eqref{id3}, i.e.
\begin{equation}\label{identificationPs}
\Phi_x\, \mathrm{P}\big(F^\varepsilon(x), F^\varepsilon(y) \big)\,\Phi_y^{-1}=2\pi\,P^{2\varepsilon}(x,y)
\end{equation}
with (see \eqref{ssp} and \eqref{sspMink})
\begin{equation}\label{isometry}
\Phi_z: (S_z,\Sl\cdot|\cdot\Sr_z)\rightarrow (\C^4,\Sl\cdot|\cdot\Sr)
\end{equation}
the isometry of indefinite inner-product spaces defined in \eqref{id2}. Finally, it holds that
\begin{equation}\label{expansionP}
\begin{split}
P^{2\varepsilon}(x,y)=\sum_{j=0}^3v_j(x-y)\: \gamma^j+\beta(x-y) \:,
\end{split}
\end{equation}
for smooth functions $v_j,\beta$ as in \eqref{espressioneP} and \eqref{functionsvbeta}.

Given the length of the proof, in order to make the exposure more digestible, we split the
considerations into several steps.
\begin{Part}{\em
Let~$\Sigma$ be either a Cauchy surface~$\{t=\mbox{const}\}$ or a {\em{connected}} open set. We want to show that
\[ 
(\mycal{X}_\Sigma^\varepsilon)'=\C\,\bI.
\] 
So, consider any $Q\in  ({\mycal X}_\Sigma^\varepsilon)'$.  We need to prove that $Q$
 is a multiple of the identity.  
 \begin{center}
 \textit{We first consider the case that~$Q$ is selfadjoint.}
 \end{center}
Since $Q$ commutes with $F^\varepsilon(z)$ for every $z\in\Sigma$, we see that $Q(S_z)\subset S_z$,
making it possible to consider the restriction $Q_z:=Q\!\restriction_{S_z} : S_z \rightarrow S_z$.
Using the isometry $\Phi_z$ in \eqref{isometry}, we can construct the operators
$$
\tilde{Q}_z:\C^4\rightarrow\C^4\quad \tilde{Q}_z:=\Phi_z\, Q_z\,\Phi_z^{-1}.
$$}
\end{Part}

\begin{Part}
For every $x,y\in\Sigma$,
\begin{equation}\label{equazionedaannullare}
\begin{split}
\sum_{j=0}^3v_j(x-y)(\tilde{Q}_x\gamma^j-\gamma^j \tilde{Q}_y)=\beta(x-y)(\tilde{Q}_y-\tilde{Q}_x) \:.
\end{split}
\end{equation}
\end{Part}
\begin{proof}
	Since $Q$ leaves $S_z$ invariant, we have $Q\,\pi_z=\pi_z\, Q\,\pi_z$ and therefore 
	$$
	\pi_z\, Q = (Q\,\pi_z)^\dagger=(\pi_z\, Q\, \pi_z)^\dagger = \pi_z\, Q\,\pi_x = Q\,\pi_z.
	$$
	In particular, it follows that, for all $x,y\in\Sigma$,
	\begin{equation}\label{commutQP}
	Q_x\, \mathrm{P}(x,y)= Q\,\pi_x\,  F^\varepsilon(y)\,\pi_y= \pi_x\,  F^\varepsilon(y)\, Q\,\pi_y= \mathrm{P}(x,y)\, Q_y
	\end{equation}
	The identities~\eqref{commutQP} and~\eqref{identificationPs} imply that 
	\begin{equation*}
	\begin{split}
	\tilde{Q}_x\, P^{2\varepsilon}(x,y) =P^{2\varepsilon}(x,y)\,\tilde{Q}_y\qquad \mbox{for all } x,y\in\Sigma\:.
	\end{split}
	\end{equation*}
	Applying~\eqref{expansionP}, we get for all $x,y\in\Sigma$
	\begin{equation*}
	\sum_{j=0}^3v_j(x-y)\tilde{Q}_x\gamma^j+\beta(x-y)\tilde{Q}_x=\sum_{j=0}^3v_j(x-y)\gamma^j\tilde{Q}_y+\beta(x-y)\tilde{Q}_y,
	\end{equation*}
	and the claim follows.
\end{proof}
\begin{Part}
	For any $x\in\Sigma$, 
	\begin{equation}\label{commutQgamma}
	\tilde{Q}_x\gamma^0=\gamma^0\tilde{Q}_x\quad\mbox{and}\quad(\tilde{Q}_x)^\dagger=\tilde{Q}_x
	\end{equation}
	(where $^\dagger$ denotes the adjoint with respect to the standard inner product).
\end{Part}
\begin{proof}
	Choosing $x=y$ in~\eqref{equazionedaannullare} we get
\[
	v_0(0) (\tilde{Q}_x\gamma^0-\gamma^0\tilde{Q}_x)=0\ \Longrightarrow\ \tilde{Q}_x\gamma^0=\gamma^0\tilde{Q}_x\quad \mbox{for all }x\in\Sigma \:, \]
	where we used that $v_i(0)=0$ and $v_0(0)\neq 0$ (see Lemma \ref{lemmavbeta}).  Now, choose $\psi,\phi\in\C^4$, then
	\begin{equation*}
	\begin{split}
	\qquad(\tilde{Q}_x\psi)^\dagger\phi&=(\tilde{Q}_x\psi)^\dagger\gamma^0\gamma^0\phi=\Sl\tilde{Q}_x \psi|\,\gamma^0\phi\Sr\stackrel{(1)}{=}\Sl Q_x\,\Phi_x^{-1}(\psi)|\Phi_x^{-1}(\gamma^0\phi)\Sr_x\stackrel{(2)}{=}\\
	&=-\langle Q\,\Phi_x^{-1}(\psi)| F^\varepsilon(x)\,\Phi_x^{-1}(\gamma^0\phi)\rangle=-\langle\Phi_x^{-1}(\psi)| F^\varepsilon(x)\,Q\,\Phi_x^{-1}(\gamma^0\phi)\rangle=\\
	&=\Sl \Phi_x^{-1}(\psi)|\,Q_x\,\Phi_x^{-1}(\gamma^0\phi)\Sr_x=\Sl\psi|\,\tilde{Q}_x\gamma^0 \phi\Sr=\Sl \psi|\, \gamma^0 \tilde{Q}_x\,\phi\Sr=\\
	&=\psi^\dagger (\tilde{Q}_x\phi) \:,
	\end{split}
	\end{equation*}
	where in $(1)$ we used the unitarity of the mapping
	$
	\Phi_z
	$, while in $(2)$ we used {\eqref{ssp}}. Since~$\phi,\psi$ are arbitrary, the result follows.
\end{proof}

\begin{Part}
	It holds that $\tilde{Q}_x=\tilde{Q}_y$ for any $x,y\in\Sigma$.
\end{Part}
\begin{proof}
Taking the adjoint~$^\dagger$ of both sides of~\eqref{equazionedaannullare}, for every $x,y\in\Sigma$ we get
	\begin{equation*}
	\begin{split}
	\overline{v_0(x-y)}(\gamma^0\tilde{Q}_x-\tilde{Q}_y\gamma^0)-\sum_{\alpha=1}^3\overline{v_\alpha(x-y)}(\gamma^\alpha\tilde{Q}_x-\tilde{Q}_y\gamma^\alpha)=\overline{\beta(x-y)}(\tilde{Q}_y-\tilde{Q}_x) \:,
	\end{split}
	\end{equation*}
	where we used that $(\gamma^0)^\dagger=\gamma^0$ and $(\gamma^\alpha)^\dagger=-\gamma^\alpha$. Using again Lemma~\ref{lemmavbeta}, this expression can be rewritten as
	\
	\begin{equation}\label{eqA}
	-v_0(y-x)(\tilde{Q}_y\gamma^0-\gamma^0\tilde{Q}_x)+\sum_{\alpha=1}^3v_\alpha(y-x)(\tilde{Q}_y\gamma^\alpha-\gamma^\alpha\tilde{Q}_x)=-\beta(y-x)(\tilde{Q}_x-\tilde{Q}_y) \:.
	\end{equation}
	This identity can be compared with~\eqref{equazionedaannullare} with~$x$ and~$y$ interchanged,
	\begin{equation}\label{eqB}
	v_0(y-x)(\tilde{Q}_y\gamma^0-\gamma^0\tilde{Q}_x)+\sum_{\alpha=1}^3v_\alpha(y-x)(\tilde{Q}_y\gamma^\alpha-\gamma^\alpha\tilde{Q}_x)=\beta(y-x)(\tilde{Q}_x-\tilde{Q}_y) \:.
	\end{equation}
Adding and the subtracting~\eqref{eqA} and~\eqref{eqB}, we get:
	\begin{equation}\label{identityab}
\mbox{for all }x,y\in\Sigma:\quad 
	\begin{cases}
	\text{(a)}&\quad \sum_{\alpha=1}^3v_\alpha(y-x)(\tilde{Q}_y\gamma^\alpha-\gamma^\alpha\tilde{Q}_x)=0,\\[0.3em]
	 \text{(b)}&\quad v_0(y-x)(\tilde{Q}_y\gamma^0-\gamma^0\tilde{Q}_x)=\beta(y-x)(\tilde{Q}_x-\tilde{Q}_y) \:.
	\end{cases}
	\end{equation}
	Using the fact that $\tilde{Q}_z$ commutes with $\gamma^0$ for every $z\in\Sigma$ (see \eqref{commutQgamma}), identity (b) above can be restated as
	\begin{equation}\label{identity}
	(v_0(y-x)\gamma^0+\beta(y-x))(\tilde{Q}_y-\tilde{Q}_x)=0\quad\mbox{for all $x,y\in\Sigma$\:.}
	\end{equation}
	For any fixed $y\in\Sigma$, the functions $\R^{1,3}\ni x\mapsto v_0(y-x)$ and $x\mapsto \beta(y-x)$ are both continuous.  Moreover, we know that the function $\R^{1,3}\ni x\mapsto |v_0(y-x)|-|\beta(y-x)|$ is strictly positive at~$x=y$ (see Lemma \ref{lemmavbeta}). Therefore, there exists an open ball $B_\delta(y)\subset\R^{1,3}$  such that
	\begin{equation}\label{neighbourproperty}
	\begin{rcases}
\ |v_0(y-x)|-|\beta(y-x)|>0\quad \\[0.1cm]
\ 	v_0(y-x)\neq 0
	\end{rcases} \ \mbox{ for all } x\in B_\delta(y) \:.
	\end{equation}
The next step is to prove that
\beq \label{Qlocal}
\tilde{Q}_x= \tilde{Q}_y \qquad \text{for every $x\in B_\delta(y)\cap\Sigma$}\:.
\eeq
To this end, fix $x\in B_\delta(y)\cap\Sigma$ and assume conversely that $\tilde{Q}_x\neq \tilde{Q}_y$. Then there exists a vector $\psi\in\C^4$ such that $\phi:=(\tilde{Q}_y- \tilde{Q}_x)\psi\neq 0$. In particular, \eqref{identity} implies
that
	$$
	\left(\gamma^0+\frac{\beta(y-x)}{v_0(y-x)}\bI_4\right)\phi=0 \:,
	$$
	which means that $-\beta(y-x)/v_0(y-x)$ is an eigenvalue of $\gamma^0$, i.e.\ $\beta(y-x)=\pm v_0(y-x)$. This is impossible because $|v_0(y-x)|>|\beta(y-x)|$, as shown in \eqref{neighbourproperty}. This implies that $\tilde{Q}_x=\tilde{Q}_y$ for any $x\in B_\delta(y)\cap\Sigma$. 

We point out that the radius of the neighborhood $B_\delta(y)$ can be fixed independently of the point $y$ because the functions~$v_0$ and $\beta$ depend on $x$ and~$y$ only through the difference vector~$y-x$. Exploiting this fact, we now want to infer that actually 
	\begin{equation}\label{equalityQ}
	\tilde{Q}_x=\tilde{Q}_y\quad \mbox{for any pair } x,y\in\Sigma.
	\end{equation}
In order to prove this, notice that, by connectedness, there always exists a continuous path in $\Sigma$ joining $x$ and $y$. The support of this path is compact and thus it can be covered by a finite number of spheres of radius~$\delta$. The claim follows by applying the local identity~\eqref{Qlocal}.
\end{proof}	

\begin{Part}
	For any $x\in\Sigma$ and $\alpha=1,2,3$,
	\begin{equation}\label{commutationalpha}
	\tilde{Q}_x\,\gamma^\alpha=\gamma^\alpha\,\tilde{Q}_x \:.
	\end{equation}
\end{Part}
\begin{proof}	
Substituting~\eqref{equalityQ} into identity (a) in \eqref{identityab}, for any $y\in\Sigma$ we get
	\begin{equation}\label{simplified1}
	\sum_{\alpha=1}^3 v_\alpha(y-x) (\tilde{Q}_y\gamma^\alpha-\gamma^\alpha\tilde{Q}_y)=0\qquad\mbox{for all }x\in\Sigma \:.
	\end{equation}
	Now, let $B_\delta(y)\subset\R^{1,3}$ be a neighborhood of $y$ such that $B_\delta(y)\subset\Sigma$ in the case $\Sigma$ open. Fix $\alpha=1$ and define the vectors
	$$
	x_\lambda :=y-\lambda e_1\in B_\delta(y)\cap\Sigma\qquad \mbox{for every $\lambda\in(-\delta,\delta)$}.
	$$  
Then
	$$
	v_\alpha(y-x_\lambda)=0\qquad \mbox{for every $\lambda\in(-\delta,\delta)$}\quad \mbox{if $\alpha\neq 1$}\:.
	$$ 
	This follows from the fact that the function in the integral expression of $v_\alpha(y-x_\lambda)$ (see \eqref{functionsvbeta}) is odd with respect to the transformation $(k^1,k^2,k^3)\mapsto (k^1,-k^2,-k^3)$.
	Thus, choosing $x=x_\lambda$, identity~\eqref{simplified1} simplifies to
	\begin{equation}\label{factortoremove}
	v_1(y-x_\lambda)(\tilde{Q}_y\gamma^1-\gamma^1\tilde{Q}_y)=0\quad\mbox{for all $\lambda\in (-\delta,\delta)$\:.}
	\end{equation}
	We now show the function $\lambda\mapsto v(y-x_\lambda)$ does not vanish identically in $(-\delta,\delta)$. This will imply that $\tilde{Q}_y$ must commute with $\gamma^1$.  We have
	\begin{equation*}
	\begin{split}
	v_1(y-x_\lambda)&\!=\!-\frac{1}{2}\int_{\R^3}\frac{d^3\V{k}}{(2\pi)^4}\,\frac{k^1}{\omega(\V{k})}\,\mathfrak{g}_\varepsilon^2(\V{k})\,e^{i\lambda k^1}\!=\\
	&=-\frac{1}{2}\!\int_{\R}d k^1 k^1\,e^{i\lambda k^1}\!\int_{\R^2}\frac{d k^2\, d k^3}{(2\pi)^{4}}\frac{e^{-2\varepsilon\sqrt{(k^1)^2+(k^2)^2+(k^3)^2+m^2}}}{\omega(k^1,k^2,k^3)}=\\
	&=\int_{\R}dk^1 k^1 e^{i\lambda k^1} f(k^1) \:,
	\end{split}
	\end{equation*}
	where the function~$f$ is defined as the double integral in the second line above. This function is clearly strictly positive on $\R$.
	At this point notice that
	$$
	v_1(0)=0\ \mbox{ and }\ \frac{d}{d\lambda}\bigg|_0 v_1(y-x_\lambda)=i\int_{\R}dk^1 (k^1)^2 f(k^1) \neq 0 \:,
	$$
because~$f$ is strictly positive. 
	Therefore, there exists at least one $\lambda$ in the vicinity of $0$ such that $v_1(y-x_\lambda)\neq 0$. Using this $\lambda$ in \eqref{factortoremove} yields $\tilde{Q}_y\gamma^1=\gamma^1 \tilde{Q}_y$. 
This calculation applies in an obvious way to $\alpha=2,3$ as well, concluding the proof.
\end{proof}

\begin{Part}\label{finalpart}
	There exists $a\in\R$ such that $Q=a\,\bI$.
\end{Part}
\begin{proof}
	Combining \eqref{commutQgamma} with \eqref{commutationalpha}, we conclude that
	$$
	\tilde{Q}_x \,\gamma^j=\gamma^j \,\tilde{Q}_x\qquad\mbox{for all }j=0,1,2,3\ \mbox{and all }x\in\Sigma.
	$$
	The gamma matrices form an irreducible set of matrices on $\C^4$. Therefore, for every $x\in\Sigma$, there must exist some complex number $a_x\in\C$ such that
	$$
	\tilde{Q}_x=a_x\,\bI_{\C^4}\Longrightarrow Q_x=a_x\, \bI_{S_x}\:.
	$$
	Because the operator~$Q_x$ is selfadjoint, we see that~$a_x \in \R$. We saw above that $\tilde{Q}_x=\tilde{Q}_y$ for every $x,y\in\Sigma$. Therefore there must exist $a\in\R$ such that:
	$$
	a_x=a\in\R\quad \mbox{for all }x\in\Sigma \:.
	$$
	This gives $Qu=au$ for any $u\in\bigcup_{x\in\Sigma}S_x$. Since the latter set is dense in $\scH_m^-$ (see Lemma \ref{lemmasuffiienza}), by continuity we get
	$
	Q=a\bI.
	$
\end{proof}

\begin{Part}
	{\em 
	In order to conclude the main proof, we need to drop the condition of self-adjointness of $Q$. Therefore,
	\begin{center}
		\textit{We now let $Q\in (\mycal{X}_\Sigma^\varepsilon)'$ be arbitrary.}
	\end{center}
	 Any operator $Q\in(\mycal{X}_\Sigma^\varepsilon)'$ can be decomposed as~$Q=Q_1+iQ_2$ with $Q_i$ self-adjoint operators given by
$$
Q_1:=\frac{1}{2}(Q+Q^*)\quad \text{and} \quad Q_2:=\frac{1}{2i}(Q-Q^*) \:.
$$
The relation~$Q F^\varepsilon(x)= F^\varepsilon(x)Q$ also implies that~$Q^* F^\varepsilon(x)= F^\varepsilon(x)Q^*$.
Therefore, both~$Q_1$ and~$Q_2$ commute with the local correlation operator $ F^\varepsilon(x)$. Applying Part \ref{finalpart}
 we get
 $$
Q=(a+bi)\,\bI\quad\mbox{for some $a,b\in\R$}\:.
$$
}
\end{Part}

\begin{Part}
	{\em
	It remains to consider the case that~$\Sigma$ is a {\em{disconnected}} open set.
	In this case, $\Sigma$ clearly contains a connected open subset, say an open ball $B\subset\Sigma$. Therefore, $ {\mycal X}_B^\varepsilon \subset {\mycal X}_\Sigma^\varepsilon$ and therefore $ ({\mycal X}_\Sigma^\varepsilon)'\subset ({\mycal X}_B^\varepsilon)'=\C\bI$.
	This concludes the proof of Theorem~\ref{maintheorem}.
}
\end{Part}

\section{Proof of Irreducibility of the Unregularized Local Algebras} \label{secappunreg}
In this appendix, we give the proof of Theorem~\ref{irreducibilityunreg}.
In order to make the exposure more digestible, we again split the argument into several steps.

\begin{Part}{\em
		Let~$\Omega\subset\R^{1,3}$ be open set. We want to show that
\[
		(L^\circ_\Omega)'=\C\,\bI. \]
				So, consider any $B\in  (L^\circ_\Omega)'$.  We need to prove that $B$
		is a multiple of the identity.  
		\begin{center}
			\textit{We first consider the case that~$B$ is selfadjoint.}
		\end{center}
	}
\end{Part}

\begin{Part}
	There exists a connected open set $U\subset\Omega$ and a measurable weakly-diffe\-ren\-tia\-ble matrix-valued function
	$$
	M:U\ni z\mapsto M(z)\in M(4,\C) 
	$$
	such that, for every $u\in\scH_m^-$,
	\begin{equation}\label{expressionB}
	Bu(z)= M(z)\,u(z)\quad\mbox{a.e. on }U.
	\end{equation}
\end{Part}
\begin{proof}
Let~$x_0\in\Omega$. By regularity (see the end of Section~\ref{subsectionCFSM}), we know that there exist four solutions~$u_\mu\in\scH_m^-\cap C^\infty(\R^{1,3},\C^4)$ such that~$u_\mu(x_0)=\mathfrak{e}_\mu$ with~$\mu=1,2,3,4$. By continuity, there is an open neighborhood $U:=B_r(x_0)\subset\Omega$ such that
	$$
	\{u_\mu(x)\:|\:\mu=1,2,3,4\}\quad\mbox{are linearly independent for every $x\in U$}.
	$$
Assumption~$B\in (L^\circ_\Omega)'$ can be restated as
	\begin{equation}\label{commutation}
	A^\circ_f\,B=B\,A^\circ_f\quad\mbox{for every } f\in C_0^\infty(\Omega,\C).
	\end{equation}
	Choosing an element~$v\in\scH_m^-$, 
	 identity~\eqref{commutation} implies that, for every $f\in C_0^\infty(\Omega,\C)$,
	\begin{equation}\label{identityintegrals}
	\begin{split}
	\int_{\R^4}f(z)\,\Sl u_\mu(z) \,|\, Bv\, (z)\Sr\, d^4z&=-\la u_\mu \,|\, A^\circ_f \,B v\ra =-\la B u_\mu \,|\,A^\circ_f\,v\ra=\\
	&=\int_{\R^4}f(z)\,\Sl Bu_\mu(z) \,|\, v(z)\Sr\, d^4z \:.
	\end{split}
	\end{equation}
	As $f$ is arbitrary, applying Du Bois-Reymond's theorem, we conclude that there exists a set $N_v\subset\Omega$ of zero measure (which depends on $v$) such that
	$$
	\Sl u_\mu(z) \,|\, Bv(z)\Sr=\Sl Bu_\mu (z)\,|\, v(z)\Sr\quad\mbox{for all }\mu=1,2,3,4 \mbox{ and }z\in \Omega\setminus N_v.
	$$ 
	In particular, this identity holds on $U\setminus N_v$ where the vectors $u_\mu(z)$ are linearly independent.
	
	Applying the Gram-Schmidt orthonormalization process, it is possible to obtain at every point $z\in U$ an orthonormal basis of $\C^4$ (with respect to the standard positive inner product):
\[ 
	w_\nu (x)=\sum_{\mu=1}^4 A_{\nu\mu}(z)\,u_\mu(z) \:. \]
	The explicit form of the coefficients $A_{\nu\mu}$ can be obtained by looking at the non-recursive form of the orthonormalization method. It can be seen directly from this  form that that the coefficients $A_{\nu\mu}$ are smooth on $U$. As a consequence, also the functions $w_\nu$ are smooth on $U$.
Now, for any $z\in U\setminus N_v$ we have
\begin{align*}
			w_\nu(z)^\dagger\, \gamma^0\, (Bv(z))&=\sum_{\mu=1}^4(A_{\nu\mu}(z))^*\,u_\mu(z)^\dagger\,\gamma^0\,(Bv(z))= \\ &=\sum_{\mu=1}^4 (A_{\nu\mu}(z))^*\:\Sl u_\mu(z)^\dagger|Bv(z)\Sr=\\
			&=\sum_{\mu=1}^4 (A_{\nu\mu}(z))^* \:\Sl Bu_\mu(z)|v(z)\Sr \:.
\end{align*}
In other words, 
	\begin{equation*}
		\begin{split}
			\gamma^0\, (Bv(z))&=\sum_{\nu=1}^4\big(w_\nu(z)^\dagger\, \gamma^0\, (Bv(z))\big)w_\nu(z)=\sum_{\nu,\mu=1}^4 (A_{\nu\mu}(z))^*\Sl Bu_\mu(z)|v(z)\Sr\,w_\nu(z)=\\
			&=\bigg(\sum_{\nu,\mu=1}^4 (A_{\nu\mu}(z))^*\,\,w_\nu(z)\, (Bu_\mu(z))^\dagger\,\gamma^0\,\bigg)v(z).
		\end{split}
	\end{equation*}
The expression in parentheses is clearly measurable on $U\setminus N_v$. If we set it to zero on $N_v$ and multiply everything from the left by~$\gamma^0$, we conclude that there exists a matrix-valued measurable function $M$ on $U$ such that~$Bv(z)=M(z)\, v(z)$ almost everywhere on $U$. Notice that the function $M$ is locally integrable and also weakly differentiable on $U$. This follows from the local integrability and weak differentiability of $B u_\mu$ (as an element of $\scH_m^-$) and the regular differentiability of $A_{\nu\mu}$ and $w_\mu$ on the entire set $\Omega$.

The above construction depends on the specific choice of the function~$v \in \scH_m^-$
only via the set~$N_v$. Therefore, if $M$ is the matrix-valued function constructed above, we see that, in fact,
\[
	Bu(z)= M(z)\,u(z)\quad\mbox{a.e. on }U\quad\mbox{for every }u\in\scH_m^- \:, \]
concluding the proof.
\end{proof}

\begin{Part}
	The matrix $M$ can be chosen to be symmetric on $U$ with respect to the spin scalar product $\Sl\cdot|\cdot\Sr$, i.e.
	$$
	M(z)^*=M(z)\quad\mbox{for all $z\in U$.}
	$$
\end{Part}
\begin{proof}
The claim is a direct consequence of the self-adjointness of $B$.
Consider again the functions $u_\mu$ defined at the beginning of the proof. Applying the same argument as
in~\eqref{identityintegrals} with $v$ replaced by $u_\nu$ we conclude that, for some set $N$ of measure zero,
\begin{align*}
\Sl u_\mu(z)\,|\,M(z) \,u_\nu(z)\Sr &=\Sl u_\mu(z)\,|\, Bu_\nu(z)\Sr \\
&=\Sl Bu_\mu(z) \,|\, u_\nu(z)\Sr =\Sl M(z)\,u_\mu(z)|u_\nu(z)\Sr
\end{align*}
for every $z\in U\setminus N$. As the vectors $u_\mu(z)$ define a basis at every point $z\in U$, the above identities implies that
$$
M(z)=M(z)^*\ \mbox{on }U\setminus N.
$$
By setting $M(z)=0$ on $N$ we can always arrange that $M(z)=M(z)^*$ for every $x\in U$.	
\end{proof}

\begin{Part}
	Let $\partial_j M$ be a representative of the weak derivative of $M$ on $U$. Then the matrices defined for
	every~$j=0,1,2,3$ and for all~$z\in U$ by 
	\begin{equation}\label{defA}
	A^j(z):=[\gamma^j,M(z)]\quad\mbox{and}\quad B(z):= \slashed{\partial}M(z)
	\end{equation}
	have the properties:
	\begin{equation}\label{antiadjoint}
	A^j(z)^*=-A^j(z),\quad \{\gamma^j,A^j(z)\}=0.
	\end{equation}
Moreover,  there is a set $N\subset U$ of measure zero such that, for all $z\in U\setminus N$,
	\begin{equation}\label{matrixzero}
\big(A^j(z)k_j+iB(z)\big)(\slashed{k}+m)=0 \quad
\mbox{for all }\V{k}\in\R^3,
\end{equation}
where  $k^0=-\omega(\V{k})$.
\end{Part}
\begin{proof}
The relations~\eqref{antiadjoint} follow directly from the properties of the Dirac matrices and the fact that the matrices $M(z)$ and $\gamma^j$ are symmetric with respect to the spin scalar product.

Now consider any $u\in\scH_{m}^-\cap C^\infty(\R^{1,3}\C^4)$. Then the function~$Bu$ is weakly differentiable and satisfies the Dirac equation weakly. Therefore,
	\begin{align}
	0&=i\gamma^j\partial_j (Bu)(z)-m(Bu)=i\gamma^j\partial_j (Mu)-m(Mu)= \notag \\
	&=i\gamma^j\big((\partial_j M)u+M\partial_j u\big)-i M\gamma^j\partial_j u= \notag \\
	&=i[\gamma^j,M]\,\partial_j u+i(\slashed{\partial}M) u \label{diracequation}
	\end{align}
	almost everywhere on $U$.
	Using~\eqref{defA}, equation~\eqref{diracequation} can be restated as follows: For any $u\in\scH_m^-\cap C^\infty(\R^{1,3},\C^4)$ there exists a null-measure set $N_u$ such that
	$$
	A^j(z)\,\partial_j u(z)+B(z)u(z)=0\ \mbox{ for every $z\in U\setminus N_u$}\:.
	$$
	We now consider a countable subset $\mathcal{D}\subset\mathcal{S}(\R^3,\C^4)$ which is dense in $L^2(\R^3,\C^4)$. 
	Choosing~$\varepsilon>0$, the set
	$$
	\mathcal{E}:=\gR_\varepsilon \big(\hat{\mathrm{E}}(\hat{P}_-(\mathcal{D})) \big)\subset \H_m^-\cap C^\infty(\R^{1,3},\C^4)
	$$
	is again countable.
	The elements of this subspace can be written as 	(see \eqref{expressionEonS} and Definition \ref{defreg})
	\begin{equation}\label{explicitexpression}
	u_\varphi(x):=\int_{\R^3}d^3\V{k}\,\mathfrak{g}_\varepsilon(\V{k})\, p_-(\V{k})\,\varphi(\V{k})\, e^{-ik\cdot x},\quad\varphi\in \mathcal{D} \:,
	\end{equation}
 (where $k_0=\omega(\V{k})$, as usual). Since $\mathcal{E}$ is countable, the set 
	$$
	N:=\bigcup\{N_u\:|\: u\in\mathcal{E}\}\subset\Omega
	$$
	has vanishing measure. Therefore, we conclude that, for every $u\in \mathcal{E}$,
	\begin{equation}\label{vanishingidentity}
	A^j(z)\,\partial_j u(z)+B(z)\,u(z)=0\ \qquad \mbox{for all } z\in U\setminus N\:.
	\end{equation}
	At this point, using the explicit expression in \eqref{explicitexpression}, we can restate \eqref{vanishingidentity} as
	\begin{equation}\label{integralvanish}
	\int_{\R^3}d^3\V{k}\,\underbrace{\mathfrak{g}_\varepsilon(\V{k})\big(A^j(z)(-ik_j)+B(z)\big) p_-(\V{k})}_{:= D(\V{k},z)}\varphi(\V{k})\, e^{-ik\cdot x}=0
	\end{equation}
	 for all $z\in U\setminus N$ and all $\varphi\in \mathcal{D}$.
For any~$z$, the rows of the 
matrix $D(\,\cdot\,,z)$ define functions of $\mathcal{S}(\R^3,\C^4)$.
Therefore, using identity \eqref{integralvanish},  the denseness of $\mathcal{D}$ in $L^2(\R^3,\C^4)$, the fact that the regularization factor is strictly positive and the definition of $p_-(\V{k})$ (see \eqref{Pmpdef})
it can be inferred that
	\begin{equation*}
	\big(A^j(z)k_j+iB(z)\big)(\slashed{k}+m)=0 \qquad
\mbox{for all }\V{k}\in\R^3\ \mbox{ and }z\in U\setminus N \:.
	\end{equation*}
This concludes the proof.
\end{proof}

\begin{Part}
	There exist two real-valued measurable functions $\lambda$ and $b$ on $U$ such that
	\begin{equation}\label{finalformM}
	M=b\,\bI_4+i\lambda\,\gamma^5\quad\mbox{a.e. on $U$}.
	\end{equation}
\end{Part}
\begin{proof}
For simplicity of notation, from now on we shall omit the argument~$z$. We fix $\alpha=1,2,3$ and choose $k=-\sqrt{r^2+m^2}\,e_0-r\,e_\alpha$. Then
identity \eqref{matrixzero} reduces to
\[ \big(-\sqrt{r^2+m^2}A^0+rA^\alpha+iB \big) \big(-\sqrt{r^2+m^2}\gamma^0+r\gamma^\alpha+m \big)=0 \:, \]
valid for arbitrary $r\in\R$. We now multiply out and decompose into the symmetric and anti-symmetric
parts under the inversion~$r \to -r$. The anti-symmetric part reads
\[ 0= -r\sqrt{r^2+m^2}\:\big(A^0\gamma^\alpha+A^\alpha\gamma^0 \big)+r\:\big(mA^\alpha+iB \gamma^\alpha \big) \:. \]
Analyzing the $r$-dependence, one sees that both matrix coefficients must vanish, i.e.\
\begin{equation*}
A^0\gamma^\alpha+A^\alpha\gamma^0=0 = mA^\alpha+iB\gamma^\alpha \:.
\end{equation*}
From these identities we obtain
\begin{equation}\label{identityAB}
A^\alpha=A^0\gamma^0\gamma^\alpha\quad\mbox{and}\quad B=-imA^\alpha \gamma^\alpha
\qquad\mbox{for all }\alpha=1,2,3\:.
\end{equation}
The first identity holds trivially if~$\alpha$ is set to zero, giving rise to
\beq \label{Ajid}
A^j=A^0\gamma^0\gamma^j \qquad \text{for $j=0,\ldots, 3$}\:.
\eeq
Using that the matrix~$A^0$ is anti-symmetric with respect to the spin scalar product
and anti-commutes with~$\gamma^0$ (see~\eqref{antiadjoint}), it follows that
\[ (A^0 \gamma^0)^* = -\gamma^0 A^0 = A^0 \gamma^0 \:. \]
Taking the adjoint of~\eqref{Ajid} (again with respect to the spin scalar product)
and using that the~$A^j$ are anti-symmetric, we obtain
\[ -(A^0\gamma^0)\, \gamma^j = -A^j = (A^j)^* = \gamma^j\, (A^0\gamma^0) \:. \]
Hence the matrices~$(A^0\gamma^0)$ and~$\gamma^j$ anti-commute. As a consequence,
the matrices~$\gamma^5 (A^0\gamma^0)$ and~$\gamma^j$ commute,
\[ \big[ \gamma^5 (A^0\gamma^0),\, \gamma^j \big] = 0 \qquad \text{for $j=0,\ldots, 3$}\:. \]
Since the Dirac matrices are irreducible, it follows that the matrix~$\gamma^5 (A^0\gamma^0)$
is a multiple of the identity. As a consequence,
\[ A^0 = a\, \gamma^5 \gamma^0 \qquad \text{with~$a \in \C$}\:, \]
and substituting into~\eqref{Ajid} gives
\beq \label{Ajrel}
A^j= a\, \gamma^5 \gamma^j \qquad \text{for $j=0,\ldots, 3$}\:.
\eeq
Using these relations on the right side of~\eqref{identityAB} gives
\beq \label{formB}
B= im a\,\gamma^5 \:.
\eeq
Likewise, applying~\eqref{Ajrel} on the left side of~\eqref{defA} gives
\[ a\, \gamma^5 \gamma^j = \big[ \gamma^j, M] \:, \]
implying that
\[ \Big[ \gamma^j, M + \frac{a}{2}\: \gamma^5 \Big] = 0 \:. \]
Again using that the Dirac matrices are irreducible, we conclude that
\[ M = -\frac{a}{2}\: \gamma^5 + b\, \bI_4 \qquad \text{with~$a,b \in \C$}\:. \]
The relation~$M^*=M$ implies that $a \in i \R$ and~$b \in \R$. The claim follows by setting $\lambda:=ia/2$.
\end{proof}

\begin{Part}
	There exists $b_0\in\R$ such that $B=b_0\,\bI$.
\end{Part}
\begin{proof}
The functions $\lambda,b$ in \eqref{finalformM} are weakly differentiable on $U$, as a consequence of the weak differentiability of $M$.
Using the explicit expression of $B$ in \eqref{formB} and substituting~\eqref{finalformM} into the general definition of $B$ in \eqref{defA}, we obtain
\begin{align*}
0&=B-\gamma^j\partial_jM=2m \lambda \gamma^5-(\gamma^j\partial_jb)\,\bI_4-i(\gamma^j\partial_j\lambda)\,\gamma^5=\\
&=(-\partial_jb)\,\gamma^j+(2 m\lambda)\,\gamma^5+(-i\partial_j\lambda)\, \gamma^j\gamma^5 \:,
\end{align*}
again almost everywhere on $U$.
Since the matrices $\gamma^j,\gamma^5$ and $\gamma^j\gamma^5$ are linearly independent, we conclude that
\[ \partial_jb=0=\lambda\qquad\mbox{a.e. on $U$ for any $j=0,1,2,3$} \:. \]
As the set $U$ is connected,
we conclude that there exists $b_0\in\R$ such that $b=b_0$ almost everywhere on $U$ and
$$
M=b_0\,\bI_4\qquad\mbox{a.e. on $U$}\:.
$$
If now we go back to \eqref{expressionB}, we see that 
$$
B u=b_0\, u\qquad\mbox{a.e. on $U$ for every $u\in\scH_m^-$}\:.
$$
Since both sides of the above identity define elements of $\scH_m^-$ and $U$ is open, Hegerfeldt's theorem (see Proposition \ref{prphegerfeldt})  implies that 
$$
Bu=b_0\,u\ \text{ on all of Minkowski space $\R^{1,3}$.}
$$ 
The arbitrariness of $u$ gives $B=b_0\,\bI$.	
\end{proof}

\begin{Part}
{\em
		So far, we only considered the special case of a selfadjoint operator $B$. 
		\begin{center}
			\textit{We now let $B\in (L_\Omega^\circ)'$ be arbitrary.}
		\end{center}
		This operator can be decomposed into self-adjoint and anti-selfadjoint components,
	$$
	B=\frac{B+B^\dagger}{2}+i\frac{B-B^\dagger}{2i}:=B_1+iB_2 \:.
	$$
	As the set $L_\Omega$ is $^*$-closed, it follows that also the selfadjoint operators $B_1$ and~$B_2$ commute with the operators $A^\circ_f$. Applying the argument above to the two individual operators, we conclude that there exists a complex number $b_0+ib_1$ such that
	$$
	B=(b_0+ib_1)\bI.
	$$
	This concludes the proof of Theorem~\ref{irreducibilityunreg}.
}
\end{Part}	
\vspace{0.4cm}


\Thanks {{\em{Acknowledgments:}} We would like to thank J\"urg Fr\"ohlich, Maximilian Jokel, Christoph Langer, Valter Moretti, Claudio Paganini and the referee for helpful discussions.


\begin{thebibliography}{10}

\bibitem{cfsweblink}
\emph{Link to web platform on causal fermion systems:
  \href{https://www.causal-fermion-system.com}{www.causal-fermion-system.com}}.

\bibitem{CB}
C.~Beck, \emph{Localization -- {L}ocal {Q}uantum {M}easurement and
  {R}elativity}, Dissertation, Ludwig-Maximilians-Universit{\"a}t M{\"u}nchen
  (2020).

\bibitem{dixmier}
J.~Dixmier, \emph{{V}on {N}eumann {A}lgebras}, North-Holland Mathematical
  Library, vol.~27, North-Holland Publishing Co., Amsterdam-New York, 1981,
  With a preface by E. C. Lance, Translated from the second French edition by
  F. Jellett.

\bibitem{pfp}
F.~Finster, \emph{The {P}rinciple of the {F}ermionic {P}rojector},
  hep-th/0001048, hep-th/0202059, hep-th/0210121, AMS/IP Studies in Advanced
  Mathematics, vol.~35, American Mathematical Society, Providence, RI, 2006.

\bibitem{srev}
\bysame, \emph{A formulation of quantum field theory realizing a sea of
  interacting {D}irac particles}, arXiv:0911.2102 [hep-th], Lett. Math. Phys.
  \textbf{97} (2011), no.~2, 165--183.

\bibitem{qft}
\bysame, \emph{Perturbative quantum field theory in the framework of the
  fermionic projector}, arXiv:1310.4121 [math-ph], J. Math. Phys. \textbf{55}
  (2014), no.~4, 042301.

\bibitem{cfs}
\bysame, \emph{The {C}ontinuum {L}imit of {C}ausal {F}ermion {S}ystems},
  arXiv:1605.04742 [math-ph], Fundamental Theories of Physics, vol. 186,
  Springer, 2016.

\bibitem{eth-cfs}
F.~Finster, J.~Fr\"ohlich, C.~Paganini, and M.~Oppio, \emph{Causal fermion
  systems and the {ETH} approach to quantum theory}, arXiv:2004.11785
  [math-ph], to appear in Discrete Contin. Dyn. Syst. Ser. S (2020).

\bibitem{review}
F.~Finster and M.~Jokel, \emph{Causal fermion systems: An elementary
  introduction to physical ideas and mathematical concepts}, arXiv:1908.08451
  [math-ph], {P}rogress and {V}isions in {Q}uantum {T}heory in {V}iew of
  {G}ravity (F.~Finster, D.~Giulini, J.~Kleiner, and J.~Tolksdorf, eds.),
  Birkh\"auser Verlag, Basel, 2020, pp.~63--92.

\bibitem{fockbosonic}
F.~Finster and N.~Kamran, \emph{Complex structures on jet spaces and bosonic
  {F}ock space dynamics for causal variational principles}, arXiv:1808.03177
  [math-ph], to appear in Pure Appl. Math. Q. (2020).

\bibitem{dice2014}
F.~Finster and J.~Kleiner, \emph{Causal fermion systems as a candidate for a
  unified physical theory}, arXiv:1502.03587 [math-ph], J. Phys.: Conf. Ser.
  \textbf{626} (2015), 012020.

\bibitem{friedlander2}
F.G. Friedlander, \emph{Introduction to the {T}heory of {D}istributions},
  second ed., Cambridge University Press, Cambridge, 1998, With additional
  material by M. Joshi.

\bibitem{froehlich2019review}
J.~Fr{\"o}hlich, \emph{A brief review of the ``{ETH}-approach to quantum
  mechanics''}, arXiv:1905.06603 [quant-ph] (2019).

\bibitem{froehlich2019relativistic}
\bysame, \emph{Relativistic quantum theory}, arXiv:1912.00726 [quant-ph]
  (2019).

\bibitem{hegerfeldt1974remark}
G.C Hegerfeldt, \emph{Remark on causality and particle localization}, Physical
  Review D \textbf{10} (1974), no.~10, 3320.

\bibitem{moretti-book}
V.~Moretti, \emph{Spectral {T}heory and {Q}uantum {M}echanics}, Unitext, vol.
  110, Springer, Cham, 2017, Mathematical foundations of quantum theories,
  symmetries and introduction to the algebraic formulation, Second edition.

\bibitem{oppio}
M.~Oppio, \emph{On the mathematical foundations of causal fermion systems in
  {M}inkowski spacetime}, arXiv:1909.09229 [math-ph], to appear in Ann. Henri
  Poincar\'e (2020).

\bibitem{thaller}
B.~Thaller, \emph{The {D}irac {E}quation}, Texts and Monographs in Physics,
  Springer-Verlag, Berlin, 1992.

\end{thebibliography}
\providecommand{\bysame}{\leavevmode\hbox to3em{\hrulefill}\thinspace}
\providecommand{\MR}{\relax\ifhmode\unskip\space\fi MR }
\providecommand{\MRhref}[2]{%
  \href{http://www.ams.org/mathscinet-getitem?mr=#1}{#2}
}
\providecommand{\href}[2]{#2}

\end{document}